\newif\ifready\readytrue
\newif\ifsubmit\submitfalse

\ifsubmit
    \documentclass[sigconf]{acmart}
\else
    \documentclass{article}
    \usepackage[utf8]{inputenc}
    \usepackage[top=1in,left=1in,right=1in,bottom=1in]{geometry}
\fi
\usepackage{graphicx}
\usepackage{graphics}
\usepackage{booktabs}
\ifsubmit
\else
\usepackage{amsthm}
\usepackage{amssymb}
\fi
\usepackage{amsmath}
\usepackage{bm}
\ifsubmit
\else
\usepackage[colorlinks=true,linkcolor=blue,citecolor=blue,allcolors=blue]{hyperref}
\fi
\hypersetup{
    colorlinks=true,
    linkcolor=blue,
    filecolor=blue,
    urlcolor=blue,
    allcolors=blue
}
\usepackage{relsize}
\usepackage[noend,ruled,linesnumbered]{algorithm2e}
\SetKwProg{Fn}{Function}{}{}
\usepackage[noend]{algpseudocode}
\usepackage{multicol}
\usepackage[most]{tcolorbox}
\usepackage{color}
\usepackage{pgfplots}
\usepackage{subcaption}
\usepackage{tablefootnote}
\pgfplotsset{
    compat=1.3,
    legend image code/.code={
        \draw [#1] (0cm,-0.1cm) rectangle (0.6cm,0.1cm);
    },
}

\usepackage{thm-restate}
\usepackage{enumerate}
\usepackage{enumitem}
\setlist{noitemsep,topsep=0pt,parsep=0pt,partopsep=0pt}
\usepackage{mathtools}
\usepackage{xspace}

\usepackage[capitalize,nameinlink]{cleveref}

\theoremstyle{plain}
\newtheorem{theorem}{Theorem}[section]

\newtheorem{invariant}{Invariant}
\newtheorem{lemma}[theorem]{Lemma}
\newtheorem{claim}[theorem]{Claim}
\newtheorem{corollary}[theorem]{Corollary}
\newtheorem{definition}[theorem]{Definition}

\newcommand{\defn}[1]{\textbf{\textit{#1}}}

\crefname{theorem}{Theorem}{Theorems}
\Crefname{lemma}{Lemma}{Lemmas}
\Crefname{claim}{Claim}{Claims}
\Crefname{observation}{Observation}{Observations}
\Crefname{algorithm}{Algorithm}{Algorithms}
\Crefname{myalgctr}{Algorithm}{Algorithms}
\Crefname{challenge}{Challenge}{Challenges}

\algrenewcommand\algorithmicindent{1em}%

\newcommand{\tO}{\widetilde{O}}

\newcommand{\kc}{$k$-core\xspace}
\mathchardef\mhyphen="2D

\DeclarePairedDelimiter\ceil{\lceil}{\rceil}
\DeclarePairedDelimiter\floor{\lfloor}{\rfloor}
\usepackage{colortbl}

\DeclareMathOperator{\poly}{poly}

\definecolor{mygreen}{RGB}{20,140,80}
\definecolor{linkcolor}{RGB}{0,0,230}
\definecolor{mylightgray}{RGB}{230,230,230}
\definecolor{verylightgray}{RGB}{245,245,245}

\newcommand{\etal}[0]{et al.}

\newcounter{myalgctr}

\newtcolorbox{OuterBox}[1][]{%
    breakable,
    enhanced,
    frame hidden,
    interior hidden,
    left=-5pt,
    right=-5pt,
    top=-5pt,
    float=p,
    boxsep=0pt,
    arc=0pt
#1}%

\newtcolorbox{InnerBox}[1][]{%
    enforce breakable,
    enhanced,
    colback=gray,
    colframe=white,
#1}%

\newenvironment{tbox}{
\vspace{0.2cm}
\begin{tcolorbox}[width=\columnwidth,
                  enhanced,
                  boxsep=2pt,
                  left=1pt,
                  right=1pt,
                  top=4pt,
                  boxrule=1pt,
                  arc=0pt,
                  colback=white,
                  colframe=black,
	              breakable
                  ]%
}{
\end{tcolorbox}
}

\newcommand{\tboxhrule}[0]{\vspace{0.1cm} {\color{black} \hrule} \vspace{0.2cm}}

\newenvironment{titledtbox}[1]{\begin{tbox}#1 \tboxhrule}{\end{tbox}}

\newcommand{\primal}{\textsc{Primal}}
\newcommand{\dual}{\textsc{Dual}}
\newcommand{\optdensity}{D^{*}}

\newcommand{\hw}{\hat{w}}
\newcommand{\zfactor}{Z}

\newcommand{\halpha}{\hat{\alpha}}

\newcommand{\core}{k}

\newcommand{\level}{\text{level}\xspace}

\newcommand{\eps}{\varepsilon}

\newcommand{\gn}{\mathcal{F}}

\newcommand{\kest}{\hat{\core}}

\newcommand{\whp}{whp\xspace}

\newcommand{\prob}{\mathsf{Pr}}
\newcommand{\geom}{\mathsf{Geom}}
\newcommand{\lf}{\psi}
\newcommand{\upexp}{(1+\lf)}
\newcommand{\upexpold}{(2+\lambda)(1+\lf)}

\newcommand{\noise}{\geom(\nparam)}
\newcommand{\hnup}{\widehat{\nup}}
\newcommand{\expect}{\mathbb{E}}

\newcommand{\nparam}{\nparaminside}

\newcommand{\lcur}{r}
\newcommand{\nparaminside}{\eps/(8\log^2n)}
\newcommand{\domain}{\mathcal{D}}

\newcommand{\approxfac}{\frac{8c\log^3n}{\eps}}

\newcommand{\densestsubfactor}{\dualadditivefactor}
\newcommand{\cutoffsubfactor}{\frac{c_2\log^4 n}{\eps}}
\newcommand{\phasesnoise}{\frac{\eps}{6T\log_{(1+\coren)}n}}
\newcommand{\cutoffnoise}{\frac{\eps}{6T(4T+1)\log_{(1+\coren)}n}}
\newcommand{\densecutoffnoise}{\frac{\eps}{3T(4T+1)\log_{(1+\coren)}n}}
\newcommand{\densestphases}{\frac{c_0 \log n}{\coren^3}}
\newcommand{\densestadd}{O\left(\frac{\log^4 n}{\eps}\right)}
\newcommand{\densetime}{O\left((n+m)\log^3 n\right)}

\newcommand{\coreadd}{O\left(\frac{\log^3 n}{\eps}\right)}

\newcommand{\kcoremultfactor}{1-12\coren}
\newcommand{\densestmultconst}{1/12}

\newcommand{\localadjust}{locally adjustable\xspace}
\newcommand{\LA}{Locally Adjustable\xspace}
\newcommand{\dense}{\rho}
\newcommand{\gcount}{s}
\newcommand{\hs}{\hat{s}}
\newcommand{\hatt}{\hat{t}}
\newcommand{\dedge}{X}

\newcommand{\degen}{d}

\newcommand{\stopf}{\mathtt{stop}}
\newcommand{\nodef}{\mathtt{update\mhyphen node \mhyphen state}}
\newcommand{\edgef}{\mathtt{update\mhyphen edge\mhyphen state}}
\newcommand{\outf}{\mathtt{out}}
\newcommand{\neighbf}{\mathtt{adj\mhyphen neighb}}
\newcommand{\adjf}{\mathtt{adj\mhyphen edge}}
\newcommand{\globalstopf}{\mathtt{global\mhyphen stop}}

\newcommand{\globaloutf}{\mathtt{global\mhyphen out}}
\newcommand{\outputs}{\mathcal{O}}
\newcommand{\probfactor}{\frac{1}{n^c}}
\newcommand{\probfactorminusone}{\frac{1}{n^{c-1}}}

\newcommand{\onoise}{\log^3(n)/\eps}

\newcommand{\zbound}{\frac{D^* - \densestsubfactor}{1+2\coren}}
\newcommand{\dualfactor}{(1+2\eta)z + \densestsubfactor}

\newcommand{\cl}{\approxfac}

\newcommand{\prev}{\tilde{p}}
\newcommand{\release}{\textbf{releases}\xspace}

\newcommand{\multfactor}{\phi}

\newcommand{\addfactor}{\zeta}
\newcommand{\order}{D}
\newcommand{\df}{\gs_f}
\newcommand{\gs}{GS}
\newcommand{\interns}{\mathcal{I}}
\newcommand{\hatn}{\hat{n}}
\newcommand{\he}{\hat{e}}

\newcommand{\dty}{unsettled\xspace}
\newcommand{\return}{\textbf{Return}\xspace}
\newcommand{\draw}{\geom(\nparam)}
\newcommand{\downexp}{(1+\lf)}

\newcommand{\const}{\coren}
\newcommand{\lbexp}{\left(\downexp^{g'} - \frac{8c\log^3 n}{\eps} - 1\right) \left(1 -\frac{1}{(2+\lambda)(1+\lf)}\right)}
\newcommand{\lbexpexpect}{\downexp^{g'} \left(1 -\frac{1}{(2+\lambda)(1+\lf)}\right)}
\newcommand{\lbabrv}{J}

\newcommand{\reals}{\mathbb{R}}
\newcommand{\norm}[1]{\left\Vert#1\right\Vert}

\newcommand{\range}{Range}
\newcommand{\rangeout}{\mathcal{Y}}
\newcommand{\adj}{\mathbf{a}}
\newcommand{\coren}{\eta}

\newcommand{\rout}{S}

\newcommand{\integers}{\mathbb{Z}}
\newcommand{\numlevels}{4\log^2 n}
\newcommand{\numgrouplevels}{2\log n}

\newcommand{\lvl}{r}
\newcommand{\mech}{\mathcal{M}}
\newcommand{\nup}{\mathcal{U}}
\newcommand{\dualadditivefactor}{\frac{c\log^3 n}{\eps \coren^3}}
\newcommand{\alg}{\mathcal{A}}

\ifready
\newcommand{\julian}[1]{}
\newcommand{\laxman}[1]{}
\newcommand{\qq}[1]{}
\newcommand{\shangdi}[1]{}
\newcommand{\jessica}[1]{}
\newcommand{\sr}[1]{{}}
\newcommand{\srnote}[1]{{}}
\newcommand{\qqnote}[1]{{}}
\newcommand{\lnote}[1]{{}}
\else
\newcommand{\julian}[1]{{\color{cyan} Julian: #1}}
\newcommand{\laxman}[1]{{\color{brown} Laxman: #1}}
\newcommand{\qq}[1]{{\color{magenta} Quanquan: #1}}
\newcommand{\shangdi}[1]{{\color{orange} Shangdi: #1}}
\newcommand{\jessica}[1]{{\color{purple} Jessica: #1}}
\newcommand{\sr}[1]{{\color{blue} Sofya: #1}}
\newcommand{\srnote}[1]{{\color{blue}\footnote{\color{blue} Sofya: #1}}}
\newcommand{\qqnote}[1]{{\color{magenta}\footnote{\color{magenta} Quanquan: #1}}}
\newcommand{\lnote}[1]{{\color{blue}\footnote{\color{brown} Laxman: #1}}}
\fi

\newcommand{\conffull}[2]{\ifdefined\confversion#1
\else#2\fi
}
\title{Differential Privacy from Locally Adjustable Graph Algorithms: $k$-Core
Decomposition, Low Out-Degree Ordering, and Densest Subgraphs}

\date{}

\author{Laxman Dhulipala, Quanquan C. Liu, Sofya Raskhodnikova, \\
Jessica Shi, Julian Shun, Shangdi Yu}

\begin{document}
\sloppy

\renewcommand{\algorithmicrequire}{\textbf{Input:}}
\renewcommand{\algorithmicensure}{\textbf{Output:}}
\algblock{ParFor}{EndParFor}
\algblock{Input}{EndInput}
\algblock{Output}{EndOutput}
\algblock{ReduceAdd}{EndReduceAdd}

\algnewcommand\algorithmicparfor{\textbf{parfor}}
\algnewcommand\algorithmicinput{\textbf{Input:}}
\algnewcommand\algorithmicoutput{\textbf{Output:}}
\algnewcommand\algorithmicreduceadd{\textbf{ReduceAdd}}
\algnewcommand\algorithmicpardo{\textbf{do}}
\algnewcommand\algorithmicendparfor{\textbf{end\ input}}
\algrenewtext{ParFor}[1]{\algorithmicparfor\ #1\ \algorithmicpardo}
\algrenewtext{Input}[1]{\algorithmicinput\ #1}
\algrenewtext{Output}[1]{\algorithmicoutput\ #1}
\algrenewtext{ReduceAdd}[2]{#1 $\leftarrow$ \algorithmicreduceadd(#2)}
\algtext*{EndInput}
\algtext*{EndOutput}
\algtext*{EndIf}
\algtext*{EndFor}
\algtext*{EndWhile}
\algtext*{EndParFor}
\algtext*{EndReduceAdd}

\ifsubmit
    \begin{abstract}
Differentially private algorithms allow large-scale data analytics while preserving user privacy. Designing such algorithms for graph data is gaining importance with the growth of large networks that model various (sensitive) relationships between individuals. While there exists a rich history of important literature in this space, to the best of our knowledge, no results formalize a relationship between certain parallel and distributed graph algorithms and differentially private graph analysis. In this paper, we define  \emph{locally adjustable} graph algorithms and show that algorithms of this type can be transformed into differentially private algorithms. 

Our formalization is motivated by a set of results that we present in the central and local models of differential privacy for a number of problems, including $k$-core decomposition, low out-degree ordering, and densest subgraphs. First, we design an $\varepsilon$-edge differentially private (DP) algorithm that returns a subset of nodes that induce a subgraph of density at least $\frac{D^*}{1+\eta} - O\left(\poly(\log n)/\varepsilon\right),$ where $D^*$ is the density of the densest subgraph in the input graph (for any constant $\coren > 0$). This algorithm achieves a two-fold improvement on the multiplicative approximation factor of the previously best-known private densest subgraph algorithms while maintaining a near-linear runtime.

Then, we present an $\varepsilon$-locally edge differentially private (LEDP) algorithm for $k$-core decompositions. Our LEDP algorithm provides approximates the core numbers (for any constant $\eta > 0$) with $(2+\eta)$ multiplicative and $O\left(\poly\left(\log n\right)/\varepsilon\right)$ additive error. This is the first differentially private algorithm that outputs private $k$-core decomposition statistics. We also modify our algorithm  to return a differentially private low out-degree ordering of the nodes, where orienting the edges from nodes earlier in the ordering to nodes later in the ordering results in out-degree at most $O\left(d + \poly\left(\log n\right)/\varepsilon\right)$ (where $d$ is the degeneracy of the graph). A small modification to the algorithm also yields a $\varepsilon$-LEDP algorithm for $\left(4+\eta, O\left(\poly\left(\log n\right)/\varepsilon\right)\right)$-approximate densest subgraph  (which returns both the set of nodes in the subgraph and its density). Our algorithm uses $O(\log^2 n)$  rounds of communication between the curator and individual nodes. 
\end{abstract}

\fi

\maketitle

\ifsubmit
\else

\fi
\raggedbottom
\thispagestyle{empty}
\newpage
\tableofcontents
\thispagestyle{empty}
\newpage

\setcounter{page}{1}
\section{Introduction}
The \emph{$k$-core decomposition} and %
related objects---densest subgraph and low out-degree ordering---are among the most important and widely used graph statistics in a variety of
\conffull{communities.}{
communities including databases~\cite{Chu20, LiZZ19,ESTW19,BGKV14,MMSS20},
machine learning~\cite{Alvarez2005,Esfandiari2018,Ghaffari2019}, graph
analytics~\cite{Khaouid15, Kabir2017, dhulipala2017julienne,
dhulipala2018theoretically}, graph visualization~\cite{Alvarez2005,carmi2007,yangdefining2015,ZhangZhao10}, theoretical 
computer science~\cite{chan2021,Esfandiari2018,Ghaffari2019,Montresor2013,SCS20}, and other communities~\cite{GBGL20,Luo19,SGJ13}.} The $k$-core decomposition assigns a 
number to each node in a network which captures how well-connected it is to the rest of the network; it is useful in applications
where one wants to find ``influential'' nodes,
or to partition a network based on each node's influence. 
\conffull{Concrete
applications include diffusion protocols in epidemiological studies
\cite{ciaperoni2020,kitsak2010,LiuTang15,malliaros2016},
community detection and computing network centrality
measures (where centers tend to have larger core numbers) \cite{dourisboure09,Healy07,Mitzenmacher2015,WangCao18,ZhangYing17}, and %
network visualization~\cite{Alvarez2005,carmi2007,yangdefining2015,ZhangZhao10}.}{Concrete
applications include diffusion protocols in epidemiological studies
(where nodes with large core numbers tend to be ``super-spreaders'') 
\cite{ciaperoni2020,kitsak2010,LiuTang15,malliaros2016},
community detection and computing network centrality
measures (where centers tend to have larger core numbers) \cite{dourisboure09,Healy07,Mitzenmacher2015,WangCao18,ZhangYing17},%
network visualization and
modeling~\cite{Alvarez2005,carmi2007,yangdefining2015,ZhangZhao10},
protein interactions~\cite{Amin2006,baderautomated2003}, and
clustering (where $k$-core decompositions are often used for preprocessing to achieve
faster performance)
\cite{GiatsidisMTV14,lee2010}.}
Because of these applications, finding fast and scalable algorithms for
exact and approximate $k$-core decompositions is an active research area with many results spanning the past decade (see~\cite{MGPV20} for a survey
\conffull{\unskip).}{ of these results).
and a plethora of results published just
in the past year~\cite{chan2021,GPC21unifying,guo2021multi,linghu_anchored_2021,LZLZT21,LZXX21,Luo21,VAC21,VBK21,WZX22,zhou2021core}.}

However, one increasingly important topic that has been
overlooked thus far in the community is the privacy of the individuals whose data are
used for computing the core numbers. Privacy measures are particularly
important for statistics such as $k$-core numbers, since such statistics output a
value for each individual in the data set, allowing for more effective attacks
that can decipher individual links and also the
information of one or more individuals in the data set. 
Such ominous possibilities call for a formal study of $k$-core decomposition
algorithms that are \emph{privacy preserving}, which is the focus of this paper.

Given an undirected graph $G$ with $n$
nodes and $m$ edges, the $k$-core of the graph is the maximal subgraph $H
\subseteq G$ such that the induced degree of every node in $H$ is at least~$k$. The \emph{$k$-core decomposition} of the graph is a partition of the
nodes of the graph into layers such that %
layer $k$ contains all the nodes that belong
to the $k$-core but not the $(k + 1)$-core. 
We illustrate the kind of attack on a (private) input
data set that can occur provided the (anonymized)
$k$-core decomposition of a social network graph.
In particular, the $k$-core decomposition gives
much more information about the structure of the graph than many other graph statistics.
Consider a social network graph consisting of a 
$k$-clique and a $k$-ary tree. Suppose the attacker does not know the graph
but has access to the exact core numbers of 
individuals in the data set. From the core numbers, the attacker can 
determine the individuals in the clique and the
edges between them. Such an attack is not possible, for example, if instead of the core numbers, the attacker has access to the degree distribution
of the graph. In this case, the attacker cannot differentiate between the clique and the non-leaf 
nodes of the $k$-ary tree.

Consider a graph $G$ that, instead of friendships, represent sensitive
information such as HIV transmissions
\cite{Little14,civ526,Wertheim17} or
cryptocurrency payments~\cite{MB19}.  
Two members of a clique in $G$, users $A$ and $B$,
may not want others to know that they share an edge; however, the $k$-core
decomposition of $G$ reveals that they do indeed share an edge. %
Such an attack emphasizes the need for $k$-core decomposition algorithms that
provide privacy for individuals.

With this goal in mind,  we present novel \emph{differentially 
private} (DP) algorithms
for the $k$-core decomposition and related objects---densest subgraph and 
low out-degree ordering---that match the multiplicative
approximation bounds of the known 
non-private algorithms. Differential privacy~\cite{DMNS06}
is the gold standard for privacy in data analysis. 
We present results in two differential privacy models: \emph{central}~\cite{DMNS06}  
and \emph{local}~\cite{kasiviswanathan2011can}. 
The central model assumes a 
trusted curator that gathers and processes non-private data from users.
In contrast, the local differential privacy model assumes \emph{no} trusted
third-party. Each node represents an individual
device of a user (e.g., phone) which releases privatized data. A third-party can act as an \emph{untrusted curator} for computing
statistics on the released data.  
Such models are particularly important  as individuals become more
wary of central authorities; and it is also a liability for companies to 
keep such sensitive data. 
Local differential privacy is a stronger notion of privacy than 
central differential privacy 
because nobody besides the owner touches any private data. Furthermore, privacy
in the local model implies privacy in the central model. On the other hand, 
local differential privacy is more restrictive for the algorithm designer,
and, for some problems, locally differentially private algorithms must have larger error than DP algorithms. In fact, several
known error lower bounds exhibit gaps between the central and local
models that could be as large as polynomial in the size of the input
(see, e.g., \cite{beimel2008distributed,CSS12,DMNS06}).

\paragraph{Outline}
We provide all definitions and notation for this paper 
in~\cref{sec:prelims}. We summarize our contributions  and provide a
technical overview in~\cref{sec:contributions}. In~\cref{sec:alg}, we
present our locally differentially private \kc decomposition, low out-degree 
ordering, and densest subgraph algorithms. In~\cref{sec:densest}, 
we provide our DP densest subgraph algorithm that achieves
a better multiplicative factor approximation than our 
densest subgraph algorithm in the local model. 
Finally, in~\cref{sec:framework}, we provide
our privacy framework that allows us to
convert a class of algorithms that we call 
\emph{locally adjustable} into differentially
private algorithms. It is open whether one can show that, in general, the utility of 
locally adjustable algorithms does not degrade significantly during the transformation.

\paragraph{Related Works}

DP algorithms have been developed for graph statistics
such as subgraph counts~\cite{KarwaRSY14,BBS13,CZ13,IMC21communication,IMC21locally,10.1007/978-3-642-36594-2_26,Sun19analyzing,Zhang15}, degree
distribution~\cite{5360242,Day16,RS16,ZHANG2021102285}, minimum
spanning tree and clustering~\cite{NRS07,HL18,KS18,NSV16},
spectral properties~\cite{10.1007/978-3-642-37456-2_28,NEURIPS2019_e44e875c}, 
cut problems~\cite{GLM10,NEURIPS2019_e44e875c,EKKL20,KL10,Sealfon16,Stausholm21}, and %
parameter 
estimation~\cite{10.1145/2623330.2623683,Ye2020,Ye2020towards}.

\textit{DP Densest Subgraph.}
The currently best DP algorithms for densest subgraphs are due to  Nguyen and Vullikanti~\cite{NV21}
and Farhadi \etal~\cite{AHS21}. We describe them in detail in~\cref{sec:contributions}.

\textit{Non-Private Algorithms for $k$-Core Decomposition and Related Problems.} 
In the non-DP, fully-dynamic setting, Bhattacharya \etal~\cite{BHNT15},
Henzinger \etal~\cite{HNW20}, and Sawlani and Wang~\cite{sawlani2020near}
provide sequential algorithms that use $\poly(\log n)$ update time
and obtain a $(4+\coren)$-approximate densest subgraph, $O(1)$-approximate
low out-degree ordering, and $(1+\coren)$-approximate
densest subgraph, respectively, for any $\coren > 0$. 
Sun \etal~\cite{SCS20} provide
the first dynamic $(4+\coren)$-approximate
$k$-core decomposition in the sequential model, using a peeling algorithm. 
A similar technique is used by Chan \etal~\cite{chan2021} in
the distributed setting. Ghaffari \etal~\cite{Ghaffari2019} give various algorithms for $(1+\eps)$-approximate
$k$-core decomposition in the massively parallel computation (MPC) model. Finally, Liu \etal~\cite{LSYDS22} formulate 
a parallel, batch-dynamic $(4+\coren)$-approximate $k$-core decomposition algorithm which
we use in our work\footnote{\label{footnote:multiplicative-error} We only consider one-sided multiplicative error. 
Thus, we translate the approximation factors of~\cite{SCS20} and~\cite{LSYDS22}, which give \emph{two-sided} error,
to instead give one-sided error, resulting in an additional factor of $2$ in our statement of their approximation bounds.}.

\section{Preliminaries}\label{sec:prelims}

We provide both DP and LDP algorithms with a focus
on the $k$-core decomposition, low out-degree ordering, and the densest subgraph. 
We define these problems here.
All privacy tools presented in this section are used in both the DP and LDP settings.

\subsection{Graph Definitions}
We use $[n]$ to denote $\{1, \dots, n\}$.
We consider undirected graphs $G = (V, E)$ with $n = |V|$ nodes and $m = |E|$ edges. For ease of indexing, we set $V=[n].$
The set of neighbors of a node $i \in [n]$ is denoted $N(i)$, and the degree of node $i$
is denoted $\deg(i)$. Our algorithms take an input graph $G$ and output an 
approximate \defn{core number} for each node in the graph 
(\cref{def:approx-core number}),
an approximate \defn{densest subgraph} (\cref{def:densest-subgraph}), 
and an approximate \defn{low out-degree ordering} (\cref{def:low-outdegree}).

\begin{definition}[$(\multfactor, \addfactor)$-Approximate Core Number]\label{def:approx-core number}
    The \emph{$k$-core} of a graph $G = (V, E)$
    is a maximal subgraph $H$ of $G$ such
    that the induced degree of every node in $H$ is at least $k$. 
    A node $v \in V$ has \defn{core number} $\kappa$ if $v$ is
    part of the $\kappa$-core but not the $(\kappa+1)$-core. Let $\core(v)$ be the core number of $v$ and 
    $\kest(v)$ be an approximation of the core number of $v$, and let $\multfactor \geq 1, \addfactor \geq 0$.
    The core estimate $\kest(v)$ is a \defn{$(\multfactor,
    \addfactor)$-approximate core number} of $v$ if
    $\core(v) - \addfactor \leq \kest(v) \leq \multfactor
    \cdot \core(v) + \addfactor$.
\end{definition}

Whereas an algorithm that outputs the \emph{exact} $k$-core decomposition does not satisfy the definition of DP (or LDP),
 we obtain an LDP algorithm for
\emph{approximate} $k$-core decomposition which gives $(2+\coren, O(\log^3 n/\eps))$-approximate
core numbers for any constant $\coren > 0$.
We define the related concept of
an \defn{approximate low out-degree ordering} based on the definition of \defn{degeneracy}.

\begin{definition}[Degeneracy]\label{def:degeneracy}
An undirected graph $G = (V, E)$ is $\degen$-degenerate if every induced subgraph of $G$ has a node with 
degree at most $\degen$. The \emph{degeneracy} of $G$ is the smallest value of $\degen$ for which $G$ is $\degen$-degenerate.
\end{definition}

It is well known that degeneracy $\degen = \max_{v \in V}\{\core(v)\}$.

\begin{definition}[$(\multfactor, \addfactor)$-Approximate Low Outdegree Ordering]\label{def:low-outdegree}
    Let $\order = [v_1, v_2, \dots, v_n]$ be a total ordering of nodes in a
    graph $G = (V, E)$. %
    The ordering $\order$ is an \defn{$(\multfactor, \addfactor)$-approximate
    low out-degree ordering} if
    orienting edges from earlier nodes to later nodes in $D$ 
    produces outdegree at most
    $\multfactor \cdot \degen + \addfactor$.
\end{definition}

We denote the \defn{density} of a graph $G = (V, E)$ by $\dense(G) := \frac{|E|}{|V|}$.
The \defn{densest subgraph} problem is defined as follows.

\begin{definition}[Densest Subgraph]\label{def:densest-subgraph}
    The densest subgraph $S_{\max}$ of a graph $G = (V, E)$ is a 
    maximal induced subgraph with maximum density. %
\end{definition}

When defining an approximate densest subgraph, we 
remove the condition on maximality of the subgraph and 
require the density to be within the specified approximation 
factors.

\begin{definition}[$(\multfactor, \addfactor)$-Approximate Densest Subgraph]\label{def:approx-densest-subgraph}
    Let the density of the densest subgraph in $G$ be $D^*$ and $\multfactor \geq 1, \addfactor \geq 0$.
    A \defn{$(\multfactor, \addfactor)$-approximate densest subgraph} $S$
    has density $\dense\left(S\right)$ at least $\frac{D^*}{\multfactor} - \addfactor$.
\end{definition}

\subsection{Differential Privacy}

We consider two models of differential privacy: central~\cite{DMNS06} and local~\cite{kasiviswanathan2011can}.
In the central model, there is a {\em trusted} curator that has direct access to the input, whereas in the local model, the curator is not trusted and gets access only to outputs of private algorithms, called {\em randomizers}. Both notions of privacy require a definition of \emph{neighboring} inputs. We
focus on \defn{edge-neighboring} graphs, defined next.

\begin{definition}[Edge-Neighboring~\cite{NRS07}]\label{def:edge-adjacent-graphs}
Graphs $G_1 = (V_1, E_1)$ and $G_2 = (V_2, E_2)$ are
edge-neighboring if they differ in one edge, namely,
if $V_1 = V_2$ and the size of the symmetric difference of $E_1$ and $E_2$ is 1. 
\end{definition}

\begin{definition}[$\eps$-Edge Differential Privacy~\cite{NRS07}]\label{def:dp}
    Algorithm $\alg(G)$, that takes as input a graph $G$ and outputs 
    some value in 
    \conffull{range $R$,}{
    $\range(\alg)$\footnote{$\range(\cdot)$ denotes the 
    set of all possible outputs of a function.},} is \textbf{$\eps$-edge differentially
    private} ($\eps$-edge DP) if for all \conffull{$\rout \subseteq R$}{$\rout \subseteq \range(\alg)$} 
    and all edge-neighboring graphs $G$ and $G'$, 
    \begin{align*}
        \frac{1}{e^{\eps}} \leq \frac{\prob[\alg(G') \in \rout]}{\prob[\alg(G) \in \rout]} \leq e^{\eps}.
    \end{align*}
\end{definition}

The primary complexity measure for $\eps$-edge DP algorithms is the \emph{running time}.

\subsection{Local Edge Differential Privacy (LEDP)}
The LEDP model is an extension of the local differential privacy (LDP) model %
originally introduced
by~\cite{kasiviswanathan2011can}.
Below we use the definitions
given in~\cite{ELRS22} which are based on definitions
in~\cite{joseph2019role} for non-graph data. %
The LEDP was also defined in~\cite{IMC21communication,IMC21locally} for the special case of one round, and
\cite{IMC21communication} additionally provides a proposition on the sequential composition
of their LEDP algorithms.

Our LEDP algorithms are described in terms of an (untrusted) \emph{curator}, who does not have access to the graph's edges, and individual nodes. During each round, the curator first queries
a set of nodes for information.
Individual nodes, which have access only to their own (private) adjacency lists, 
then \emph{release} information via {\em local randomizers,} defined next. 

\begin{definition}[Local Randomizer (LR)]\label{def:local-randomizer}
    An \defn{$\eps$-local randomizer} $R: \adj \rightarrow \rangeout$ for node $v$ is an $\eps$-edge DP 
    algorithm that takes as input the set of its neighbors $N(v)$, represented by
    an adjacency list $\adj = (b_1, \dots, b_{|N(v)|})$. In other words, $$\frac{1}{e^{\eps}} \leq \frac{\prob\left[R(\adj') \in Y\right]}{\prob\left[R(\adj) \in Y\right]} \leq e^{\eps} $$ for all %
    $\adj$ and $\adj'$  where the symmetric difference
    is $1$ and all sets of outputs $Y \subseteq \rangeout$. The probability is taken over the
    random coins of $R$ (but \emph{not} over the choice of the input). 
\end{definition}

The information released via local randomizers is public to all nodes and the curator. 
The curator performs some computation on the released information and makes the result public. The overall computation is formalized via the notion of the transcript.

\begin{definition}[LEDP]\label{def:LEDP}
A \defn{transcript} $\pi$ is a vector consisting of 4-tuples $(S^t_U, S^t_R, S^t_\eps, S^t_Y)$ -- encoding the set of parties chosen, set of randomizers assigned, set of randomizer privacy parameters, and set of randomized outputs produced -- for each round $t$. Let $S_\pi$ be the collection of all transcripts and $S_R$ be the collection of all randomizers. Let $\perp$ denote a special character indicating that the computation halts.
A \defn{protocol} is an algorithm $\alg: S_\pi \to 
(2^{[n]} \times 2^{S_R} \times 2^{\mathbb{R}^{\geq 0}} \times 2^{\mathbb{R}^{\geq 0}})\; \cup \{\perp\}$
mapping transcripts to sets of parties, randomizers, and randomizer privacy parameters. The length of the transcript, as indexed by $t$, is its round complexity.

Given $\eps\geq 0$,  a randomized protocol $\alg$ on (distributed) graph $G$ is \defn{$\eps$-locally edge differentially private ($\eps$-LEDP)} if the algorithm that outputs the entire transcript generated by $\alg$ is $\eps$-edge differentially private on graph $G.$
 If $t=1$, that is, if there is only one round, then $\alg$ is called \defn{non-interactive}. Otherwise, $\alg$ is called \defn{interactive}.
\end{definition}

Since LEDP algorithms operate in the distributed setting,
the complexity measures that we care about for our LEDP algorithms is the
\defn{number of rounds} of communication and the \defn{node communication complexity}, or the maximum size of 
each message sent from a user to the curator. We assume each user can see the public information
for each round on a public ``bulletin board''. \conffull{Additional differential privacy tools
and definitions can be found in the full version of our paper.}{}

\conffull{}{
\subsection{Differential Privacy Tools}

We define these techniques for DP but the same cited sources also prove they hold under LDP.
Also for simplicity, we give these techniques in graph theoretic terms but they also hold 
for broader applications (e.g., for databases, etc.).
Our algorithms use the notion of \emph{sensitivity}, defined as follows.

\begin{definition}[Global Sensitivity~\cite{DMNS06}]\label{def:global-sensitivity}
    For a function $f: \domain \rightarrow \reals^d$, 
    where $\domain$ is the domain of $f$ and 
    $d$ is the dimension of the output, the 
    \defn{$\ell_1$-{sensitivity}} of 
    $f$ is $\df = \max_{G, G'} \norm{f(G) 
    - f(G')}_1$ for all pairs of $(G, G')$ of edge-neighboring graphs.
\end{definition}

Our algorithms use the \emph{symmetric geometric distribution} 
given in previous papers~\cite{BV18,CSS11,DMNS06,DNPR10,AHS21,SCRCS11}. 
One can think of the symmetric
geometric distribution as the ``discrete Laplace distribution.'' The 
advantage of such a distribution is that we avoid the rounding 
issues of continuous distributions and this distribution is 
sufficient for our needs.

\begin{definition}[Symmetric Geometric Distribution~\cite{BV18,SCRCS11}]\label{def:geom}
    The \defn{symmetric geometric distribution}, denoted $\geom(b)$, 
    with input parameter $b \in (0, 1)$, takes
    integer values $i$ where the probability mass function at $i$ is
    $\frac{e^b - 1}{e^b + 1} \cdot e^{-|i| \cdot b}$.
\end{definition}

We denote a random variable drawn from the distribution as $X \sim \geom(b)$.

\defn{With high probability} or \defn{\whp} is used in this paper to mean with probability at least $1 - \frac{1}{n^c}$ for any constant
$c \geq 1$. As with all private algorithms, privacy
is \emph{always} guaranteed and the approximation factors are guaranteed \whp. 
We can upper bound the symmetric geometric noise \whp using the following.

\begin{lemma}\label{lem:noise-whp-bound}
With probability at least $1 - \frac{1}{n^c}$ for any constant $c \geq 1$, we can upper 
bound $\dedge \sim \geom\left(x\right)$ by $|\dedge| \leq \frac{c \ln n}{x}$. 
\end{lemma}

\begin{proof}
The probability that $|\dedge| > \frac{c\ln n}{x}$ is given by 
$2\sum_{i = \frac{c\ln n}{x} + 1}^{\infty} \frac{e^x - 1}{e^x + 1} \cdot e^{-|i| \cdot x} = 
2 \cdot \frac{\exp\left(-x \cdot \frac{c\ln n}{x}\right)}{e^{x} + 1} = 
2 \cdot \frac{\exp\left(-c\ln n\right)}{e^{x} + 1} \leq \frac{1}{n^c}$. Thus, the 
probability that $|\dedge| \leq \frac{c \ln n}{x}$ is at least $1 - \frac{1}{n^c}$.
\end{proof}

The \defn{geometric mechanism}
uses the symmetric geometric distribution. 

\begin{definition}[Geometric Mechanism~\cite{BV18,CSS11,DMNS06,DNPR10}]\label{def:sgd-mech}
    Given any function $f: \mathcal{D} \rightarrow \integers^d$, where 
    $\mathcal{D}$ is the domain of $f$ and $\df$ is the $\ell_1$-sensitivity 
    of $f$, the geometric mechanism is defined as
    $\mech_G(x, f(\cdot), \eps) = f(x) + (Y_1, \dots, Y_d)$,
    where $Y_i\sim \geom(\eps/\df)$ are i.i.d.\ random variables
    drawn from $\geom(\eps/\df)$ and $x$ is
    a data set.
\end{definition}

\begin{lemma}[Privacy of the Geometric
    Mechanism~\cite{BV18,CSS11,DMNS06,DNPR10}]\label{lem:sgd-private}
    The geometric mechanism is $\eps$-edge DP. %
\end{lemma}

The \emph{composition theorem} guarantees privacy for 
the \emph{composition} of multiple
algorithms with privacy guarantees of their own. %
In particular, this theorem covers the use case where multiple DP
algorithms are used on the \emph{same} dataset. We also use a theorem for 
group differential privacy.

\begin{theorem}[Composition Theorem~\cite{DMNS06,DL09,DRV10}]\label{thm:composition}
    A sequence of DP algorithms, $(\alg_1, \dots, \alg_k)$, with privacy parameters $(\eps_1, \dots, \eps_k)$ form at worst an $\left(\eps_1 + \cdots + \eps_k\right)$-DP algorithm under \emph{adaptive composition} (where the adversary can adaptively select algorithms after
    seeing the output of previous algorithms).%
\end{theorem}

\begin{lemma}[Group Privacy~\cite{DMNS06}]\label{thm:group-dp} 
Let $k\in \mathbb{N}$. Every $\eps$-differentially private algorithm $\alg$ is $\left(k\eps\right)$-differentially 
private for groups of size $k$. That is, for all data sets $X, X'$ such that $\|X - X' \|_0 \leq k$ and all subsets $Y \subseteq \mathcal{Y}$,
\begin{equation*}
    \frac{1}{e^{k\eps}} \leq \frac{\Pr[\alg(X') \in Y]}{\Pr[\alg(X) \in Y]} \leq e^{k\eps}.
\end{equation*}
\end{lemma}

For completeness, we provide a proof of~\cref{thm:composition}
in~\cref{app:adaptive}.
Finally, the post-processing theorem states that the result of post-processing on the output
of an $\eps$-edge DP algorithm is $\eps$-edge DP.

\begin{theorem}[Post-Processing~\cite{DMNS06,BS16}]\label{thm:post-processing}
Let $\mech$ be an $\eps$-edge DP mechanism and $h$ be an arbitrary (randomized)
mapping from $\range(\mech)$ to an arbitrary
set. The algorithm $h \circ \mech$ is $\eps$-edge DP.\footnote{$\circ$ is notation 
for applying $h$ on the outputs
of $\mech$.}
\end{theorem}

}

\conffull{}{
\subsection{Notations}

All notations used in our paper are included in~\cref{tab:params}; notation used in a specific section
are labeled with that section.

\begin{table*}[t!]
    \centering
    \footnotesize
    \begin{tabular}{c l}
    \toprule
         Notation & Meaning \\
         \midrule
         $G = ([n], E)$ & initial input graph  \\
         $\eps > 0$ & privacy parameter\\
         $\coren > 0$ & approximation parameter\\
         $i$ & each node $i$ has a unique index $i \in [n]$\\
         $N(i)$ & $i$'s set of neighbors\\
         $[x]_a^b$ for $x,a,b\in\mathbb{R},$ where $a<b$ & $a$ if $x<a$; $x$ if $x\in[a,b]$; and $b$ if $x>b$\\
         $G[V']$ & induced subgraph by nodes in $V'$\\
         $\dense(G[V'])$ & density of induced subgraph by nodes in $V'$\\
         $\geom(b)$ & symmetric geometric distribution (\cref{def:geom}) with parameter $b$\\
         $[n]$ & integers in $\{1, \dots, n\}$\\
         \midrule
         \multicolumn{2}{c}{Notation in~\cref{sec:densest}}\\
         \midrule
         $\ell(e)$ & load on edge $e$\\
         $\hat{\alpha}_{eu}$ & new load added to edge $e$ from node $u$\\
         $T$ & number of phases\\
         \midrule
         \multicolumn{2}{c}{Notation in~\cref{sec:alg}}\\
         \midrule
         $\lambda \in (0, 1)$ & constant parameter used in upper bounds on induced degree \\
         $\lf \in (0, 1)$ & constant parameter used to determine number of levels \\
         $L_{r}$ & array of node levels publicized by curator for round $r$ \\
         $C$ & array containing core number estimates of each node\\
         $\gn(r)$ & function that outputs $\floor{r/(2\log_{(1+\lf)} n)}$\\
         $\level(i)$ & the level of node $i$\\
         $\delta(G[V'])$ & sum of degrees in induced subgraph of $V'$ divided by $|V'|$, i.e. $\frac{\sum_{v \in V'} \deg_{G[V']}(v)}{|V'|}$. \\
         $\nup_i(\lcur)$ & number of neighbors of node $i$ at level $\lcur$ or higher\\
         \bottomrule
    \end{tabular}
    \caption{Notation} 
    \label{tab:params}
\end{table*}
}
\section{Our Contributions and Technical Overview}\label{sec:contributions}
Our main contributions in this paper are
our locally private approximation
algorithms for $k$-core decomposition, low out-degree ordering, and densest subgraphs that achieve bicriteria approximations where the multiplicative factors \emph{exactly match}
the multiplicative approximations of the original non-private algorithms, and they only incur additional small $\frac{\poly(\log n)}{\eps}$
additive error. Our key observation is that the
release of 
noisy \emph{levels} in the
recent \emph{level data structures} of
\cite{BHNT15,HNW20,LSYDS22}
used for these problems allows for privacy
at the cost of only a $\poly(\log n)$ increase in the 
additive error.

We first give $\eps$-LEDP algorithms using 
a number of recent non-private 
algorithms for static and dynamic orientation and 
\kc decomposition algorithms~\cite{BHNT15,HNW20,chan2021,LSYDS22}.
Our algorithms are the first $\eps$-LEDP algorithms for approximate \kc decomposition, low out-degree orientation and densest subgraph. 
For the approximate \kc decomposition problem, we must return an approximate core number for \emph{every} node in the graph.
Even though a single edge addition can cause the \emph{exact} core number of every node to change, 
we are able to give a $\left(2+\coren, \coreadd\right)$-approximate $\eps$-LEDP algorithm.%

One of the challenges in adapting 
known techniques to obtain
LEDP algorithms for the $k$-core decomposition is that the vector of $k$-core numbers has high sensitivity.
Specifically, its global sensitivity is
$n$, as \emph{one} edge update can cause the $k$-core number of \emph{every}
node to increase or decrease by $1$. E.g., all nodes in a cycle have core number $2$, and the deletion of any edge decreases
the core number of \emph{every} node by $1$.
\conffull{}{\footnote{More details on these challenges can be found in~\cref{app:challenges}.}}

The key to our results is to instead consider a function that returns the 
\emph{induced degree} of each node in a special
subgraph of the input graph. The sensitivity of this function is $2$, since one additional edge can increase the degree of at most two nodes, each
by $1$. We bound the number of times the curator queries the value of this function by $O(\log^2 n)$ in the worst case, thus eliminating
the need for the sparse vector technique (SVT) used in previous works~\cite{AHS21,FHO21}. 

Our algorithm runs in $O(\log^2n)$ rounds where each node sends messages
with $O(1)$ bits to the curator. 
Our algorithm is based on a level data structure
investigated by a number of works~\cite{BHNT15,HNW20,LSYDS22}.
In these data structures, nodes are partitioned into levels. The key insight that allowed us to 
achieve privacy in the local model is that each node in our algorithm only requires knowledge of the levels 
their neighbors are on. Thus, nodes can publish a \emph{noisy level} for each phase. Using the noisy level
in multiple phases directly leads to core number estimates. Our algorithm matches the multiplicative approximation
factor of the best known algorithm for this problem by Chan \etal~\cite{chan2021}. We also give a version of our algorithm that uses 
only $O(\log n)$ rounds (with messages of size $O(\log n)$ sent to the 
curator). %
We present this algorithm and its analysis in~\cref{sec:alg}.

\begin{theorem}[$\eps$-LEDP $k$-Core Decomposition]\label{thm:ldp-static}
    There exists an $O\left(\log n\right)$ round
    $\eps$-LEDP
    algorithm that, given a constant $\coren > 0$, returns $\left(2+\coren, \coreadd\right)$-approximate core numbers with high probability.
\end{theorem}

As a consequence of our algorithm, we give an $\eps$-LEDP algorithm for \emph{approximate
low out-degree ordering}.
Edge orientations obtained from such orderings
are useful in graph algorithm design 
and contribute to a variety of faster algorithms
for fundamental problems, such as coloring~\cite{Matula83}, 
matching~\cite{CHS09}, 
triangle counting~\cite{CN85}, independent set~\cite{onak2020fully}, 
dominating set~\cite{AG09},
and many others. 
Furthermore, many real-world graphs have small 
degeneracy~\cite{bera2020degeneracy,eppstein2013listing,shun2015multicore,dhulipala2018theoretically}, %
making such algorithms practically useful.
Note that (perhaps obviously) 
an algorithm that explicitly outputs an orientation for every edge cannot be edge DP.
Instead,
the output of our private algorithm is an ordering on the vertices, and the edge orientation can be obtained by orienting each edge from the lower to the higher endpoint in the ordering. Private orderings have been considered in previous
works: e.g., Gupta \etal~\cite{GLM10} give an $\eps$-edge DP ordering for the vertex cover problem 
(where the earliest endpoint of an edge covers that edge).

\begin{theorem}[$\eps$-LEDP Low Out-Degree Ordering]\label{thm:ldp-ordering}
    There exists an $O(\log^2 n)$-round
    $\eps$-LEDP %
    algorithm that, given a constant $\coren > 0$ and a graph of degeneracy $\degen$, returns a total ordering of the nodes such that orienting edges from nodes
    earlier in the ordering to nodes later in the ordering results in out-degree at most $(4+\coren)\degen + O\left(\frac{\log^3 n}{\eps}\right)$, with high probability.
\end{theorem}

With
an additional tweak, our algorithm also yields 
the first result for the densest subgraph problem in the local model.

\begin{theorem}[$\eps$-LEDP Densest Subgraph]\label{thm:local-densest-subgraph}
    There exists an $O(\log n)$-round $\eps$-LEDP %
    algorithm that returns a set 
    of nodes %
    that induce a $\left(4+\coren, O\left(\frac{\log^3 n}{\eps}\right)\right)$-approximate
    densest subgraph. Our algorithm returns both the nodes in the densest subgraph as well as the approximate density.
\end{theorem}

Interestingly, our \emph{static} LEDP algorithms are inspired by a 
body of non-private works in the dynamic setting~\cite{BHNT15,HNW20,LSYDS22,SCS20}. 
We show that the locality of these dynamic algorithms and the bounded sequential dependency
paths allow us to provide our LEDP guarantees as well as minimize the number of rounds of 
communication. Our techniques  may potentially lead to other 
LEDP algorithms inspired by non-private dynamic algorithms.

There are currently no known $(1+\coren)$-approximate \kc decomposition algorithm 
(even in the non-private setting) that 
use $\poly(\log n)$ phases. (Non-private $(1+\coren)$-approximate
algorithms that use $\omega(\poly(\log n))$
phases do exist.) 
This bound on the
number of phases is essential
for obtaining our additive approximation guarantees. It is open whether we can obtain a framework without additive error dependence on the 
number of phases. Also, independently, it is interesting to see whether non-private
$(1+\coren)$-approximate \kc decomposition algorithms that take $\poly(\log n)$ phases exist.

Then, we give an $\eps$-edge DP densest subgraph algorithm based on the \emph{non-private} parallel
algorithm of Bahmani \etal~\cite{BGM14} together with the modifications made by Su and Vu~\cite{SuVu20} in their non-private
distributed algorithm.
Using this algorithm and our mechanism, we present a $\left(1+\coren, \densestadd\right)$-approximate $\eps$-edge DP
algorithm for the densest subgraph problem with runtime $O((n + m)\log^3 n)$ when $\coren > 0$ is constant. This improves on the multiplicative approximation factors of
the previous $(\eps, \delta)$-edge DP $O\left(2+\coren, \frac{\log (n)\log(1/\delta)}{\eps}\right)$-approximation results of Nguyen and Vullikanti~\cite{NV21} and 
the more recent $\eps$-edge DP $\left(2, O\left(\frac{\log^{2.5} (n)\log(1/\sigma)}{\eps}\right)\right)$-approximation algorithm of 
Farhadi \etal~\cite{AHS21} that obtains this approximation guarantee with probability $1-\sigma$. We achieve this improvement in the approximation factor with only an $O(\log^3 n)$ increase
in the runtime. Furthermore, when $\coren > 0$ is constant, 
our algorithm matches the $((n + m)\log^2 n)$ runtime 
up to a $O(\log n)$ factor
of the best known non-private algorithm of Chekuri \etal~\cite{CQT22}, which obtains a $(1+\coren)$-approximate
densest subgraph algorithm. We present our algorithm and its analysis in~\cref{sec:densest}.

\begin{theorem}[$\eps$-Edge DP Densest Subgraph]\label{thm:densest-subgraph}
    There exists an $\left(1+\coren, \densestadd \right)$-approximate algorithm for the densest subgraph problem
    that is $\eps$-edge DP and runs in $O((n + m)\log^3 n)$ worst-case time for constant $\coren > 0$ that returns an approximation factor 
    within the stated bounds with high probability.
\end{theorem}

Finally, we present our general framework. 
This framework is a formalization of a category of graph algorithms that we call \defn{\localadjust}. We present a simple mechanism for converting \localadjust graph algorithms 
into $\eps$-edge DP and $\eps$-LEDP algorithms. Our framework and privacy proofs are given in~\cref{sec:framework}. 

Three crucial observations about \localadjust algorithms allow us to obtain private algorithms with small error. 
First, such algorithms proceed in several phases where the state of each node or edge is updated via a function that uses
\emph{only} information from its immediate neighbors and incident edges. We are thus able to bound the sensitivity of such functions
by a small constant in edge-neighboring graphs. Second, these local functions often only require the \emph{number} of neighbors and incident edges
that satisfy a condition without requiring knowledge of the states of these neighbors. This property allows us to easily add
noise via the geometric mechanism to the \emph{count} of neighbors/incident edges that satisfy the condition. The noise can be viewed as either hiding actual edges or representing \defn{dummy edges} which may additionally satisfy the condition. In edge-neighboring graphs $G$ and $G'$,
a dummy edge incident to node $v$ can account for the additional neighbor of $v$ that is present in $G'$ but not in $G$.
Finally, we see that the additive error of each of our algorithms depends on its total number of \defn{phases} (or the number of times the
algorithm is run before returning an output). We show that state-of-the-art parallel and distributed algorithms naturally run in 
$\poly(\log n)$ phases, thus leading to only a $\poly(\log n)$ factor in the additive error.
It is as an
interesting open question whether general utility guarantees can be obtained for our framework.

\section{Local Algorithms for Core Decomposition, Low Out-Degree Ordering, and Densest Subgraph}\label{sec:alg}

In this section, we give local algorithms for core decomposition, low out-degree ordering, and densest subgraph.
Our local algorithms are based on extending algorithms for a  
recently developed non-private batch-dynamic level data structure to the private setting.
We start with our local core decomposition algorithm and its 
privacy guarantees and accuracy, and then present our local low
out-degree ordering and densest subgraph algorithms. 
\conffull{All proofs are provided in the full version of our
paper.}{}

\subsection{$\eps$-LEDP $k$-Core Decomposition}

We present our algorithm in this section, 
\conffull{and provide our analyses in the
full version of our paper.}{
provide its privacy analysis in~\cref{sec:privacy}, 
and finally analyze its approximation factor 
and round complexity in~\cref{sec:approx-runtime}.}
Our LEDP algorithm is inspired by the sequential level data structure
algorithms of~\cite{BHNT15,HNW20} and the parallel level data structure of~\cite{LSYDS22}
used to obtain the densest subgraphs and
bounded degeneracy orientations of a graph in the dynamic setting. We crucially use the
observation made by~\cite{LSYDS22} that the longest path of sequential dependencies
is $O(\log^2 n)$ for \emph{any number} of updates. We use this, along with 
our new private procedures for releasing \emph{multiple} outputs \emph{simultaneously} for 
a phase at once,  to prove both our privacy
guarantees and our $O(\log^2 n)$ rounds, $O(1)$ node communication complexity bound.
A simple extension allows us to obtain an LEDP algorithm that uses $O(\log n)$ rounds and $O(\log n)$ communication complexity per node.

\ifsubmit
\else
Finally, we are able to improve the multiplicative approximation factor 
from $(4+\coren)$ in~\cite{LSYDS22} to $(2+\coren)$ (see \cref{footnote:multiplicative-error})
because~\cite{LSYDS22} is a dynamic algorithm
while ours is a static algorithm. The additional factor of $2$ was useful in reducing the parallel work 
in the fully batch-dynamic algorithm (when accounting for both insertions
and deletions). In our static setting, we only need to maintain insertions (from inserting all the edges of the 
graph as a batch) which allows us to reduce the approximation factor by $2$.
However, in order to ensure our privacy guarantees in the LEDP model, we compute
a new noise for every node. In the sequential computation setting, this incurs an additional
factor of $n$ in the running time, which also means that we no longer need the amortized
analysis from previous works, since the additional factor of $n$ dominates
the running time. But in the distributed setting,
we are still able to show that the  worst-case number of rounds is $O(\log^2 n)$ (and $O(\log n)$
rounds with a modification).
\fi

\paragraph{Non-Private Dynamic Algorithms of~\cite{BHNT15,HNW20,LSYDS22,SCS20}}
The level data structure of~\cite{BHNT15,HNW20} partitions the nodes of the
graph into $O(\log^2 n)$ \emph{levels}.
The levels are partitioned into
\emph{groups} of $O(\log n)$ levels each. Nodes move up and down the levels
according to a set of rules on the induced degree of a node $i$
with respect to the number of neighbors in the same or higher level than $i$ (sometimes the set
of nodes in the level just below the level of $i$ is also considered).
Specifically, if the induced degree is too high, $i$ moves one level up;
if the induced degree is too low, $i$ moves one level down. One crucial 
aspect of this algorithm is that \emph{a node only needs to know the 
levels of its immediate neighbors.}

In addition to the above structure, we use
a crucial aspect of the algorithm of~\cite{LSYDS22} that 
allows us to move nodes from the same level
\emph{simultaneously} without causing additional nodes
below them to move. Since we only have $O(\log^2 n)$ levels, we can process them in a bottom-up fashion while accounting 
for all vertices in the current level of the iteration simultaneously,
achieving the $O(\log^2 n)$ round complexity of our $\eps$-LEDP algorithm.
A further modification that accounts for different groups \emph{simultaneously} allows us to 
decrease the round complexity.

\subsection{Detailed Algorithm}

\conffull{}{
Our algorithms use the notation
shown in~\cref{tab:params}.}
For the remainder of this section, $\log n$ means $\log_{(1+\lf)} n$ 
for a parameter $\lf > 0$ that affects our approximation factor.
There are $4\log^2 n$ levels in the structure. As in previous works~\cite{BHNT15,HNW20,LSYDS22},
we call this structure a \emph{level data structure}. Our descriptions use terminology 
from these previous works. However, because we are in the static setting, we are able to simplify
the algorithm of~\cite{LSYDS22} in some ways by considering the input 
graph as a single batch of edge insertions.

Levels in the level data structure are partitioned
into $2 \log n$ \defn{groups} of equal size. 
Each group $g_i$ contains $2\log n$ \emph{consecutive} levels. We number the
levels starting with $0$ as the bottommost level and $\numlevels - 1$ as the topmost level. 
Each group $g_i$ has an associated index $i$, and contains levels in $[i \cdot 2\log n, (i+1)\cdot 2\log n)$.
Let the group index that a level $\lcur$ belongs to be $\gn(\lcur)$. 
In other words, $\gn(\lcur) = f$ if $\lcur \in g_f$, which can be computed 
as $\gn(\lcur) = \floor{\lcur/2\log n}$. We denote the level of a node $i$ as $\level(i)$.
Finally, a node is \defn{\dty} if it must move to a higher level. 

We now describe our algorithm.
The pseudocode for our algorithm is given 
in~\cref{alg:insert}. 
The algorithm is given a sequence of adjacency lists, 
$(a_1, \dots, a_n)$, constant parameters $\lambda, \psi > 0$ that will determine
the approximation factor, and privacy parameter $\eps \in (0, 1)$ as input.
The algorithm outputs $\left(2+\coren, O\left(\frac{\log^3 n}{\eps}\right)\right)$-approximate
core numbers and a low out-degree ordering with the same guarantee on the out-degree. The approximation
parameter $\coren$ is determined from $\lambda$ and $\psi$.
All nodes start in level $0$ (\cref{line:initial}).
Throughout the algorithm, the curator maintains and publishes
the current levels of the nodes in a 
set of $\numlevels$ lists, each of size $n$ where the 
$i$-th index of a list contains the level of node $i$ (\cref{line:levels}).
\SetKwFunction{FnUpdateLevels}{UpdatenodeLevels}
\SetKwFunction{FnInsert}{Incremental}
\SetKwFunction{FnStatic}{LEDPCoreDecomp}
\SetKwFunction{FnLEDPDensestSubgraph}{LEDPDensestSubgraph}
\SetKwFunction{FnDelete}{Decremental}
\SetKwFunction{FnCoreNumber}{EstimateCoreNumbers}
\SetKwFunction{FnLowOutdegree}{EstimateOrdering}

\begin{algorithm}[!t]
\conffull{
\small}{}
    \textbf{Input:} Adjacency lists $(\adj_1, \dots, \adj_n)$, constant $\coren \in (0, 1)$, and
    privacy parameter $\eps \in (0, 1)$.\\ 
    \textbf{Output:} $\eps$-LEDP $\left(2+\coren, O\left(\frac{\log^3 n}{\eps}\right)\right)$-approximate core numbers and low out-degree ordering of each node in $G$.\\
    \Fn{\FnStatic{$(\adj_1, \dots, \adj_n), \eps, \coren$}}{
        Set $\psi = 0.5$ and $\lambda = \frac{2}{9} (2 \coren - 5)$.\\ %
        Curator initializes $L_0, \dots, L_{\numlevels - 1}$ with $L_\lcur[i] \leftarrow 0$ for every $i \in [n], \lcur \in [0, \dots, \numlevels - 1]$.\label{line:initial}\\
        \For{$\lcur = 0$ to $\numlevels - 2$}{\label{line:iterate}
            \For{$i = 1$ to $n$}{
                $L_{\lcur + 1}[i] \leftarrow L_{\lcur}[i]$.\label{line:old-level} \\
                \If{$L_\lcur[i] = \lcur$}{\label{line:level-cur-level} 
                    Let $\nup_i$ be the number of neighbors $j \in \adj_i$ where $L_\lcur[j] = \lcur$.\label{line:compute-up} \\
                    Sample $X \sim \geom(\nparam)$.\\
                    Compute $\hnup_i \leftarrow \nup_i + X$.\label{line:compute-noisy-up}\\
                    \If{$\hnup_i > \upexp^{\gn(\lcur)}$}{\label{line:move-up-condition}
                        $i$ \release $1$\label{line:release-noisy-up}.\\
                        $L_{\lcur + 1}[i] \leftarrow L_{\lcur}[i] + 1$.\label{line:update-list} \Comment{Curator moves $i$ up one level.}\\
                    }\Else{\label{line:stay-same}
                        $i$ \release $0$\label{line:release-noisy-bool}.\\
                    }
                }
            }
            
            Curator publishes $L_{\lcur + 1}$.\label{line:levels}\label{line:new-pub-list}\\
        }
        Curator calls $C \leftarrow \FnCoreNumber(L_{\numlevels - 1}, \lambda, \psi)$.\\
        Curator orders nodes in $D$ by $L_{\numlevels-1}$ (from smaller to larger) breaking ties by node index.\label{line:ordering}\\
        \return $(C, D)$.
        \caption{\label{alg:insert} $\eps$-LEDP Decomposition and Ordering}
    }
\end{algorithm}
First, the curator iterates through the levels starting from level $0$ to level $\numlevels - 1$ (\cref{line:iterate}). 
Let $\lcur$ be the current level.
The curator asks each node, $i$, in level $\lcur$
to compute its \emph{noisy} number of neighbors in level $\lcur$, denoted
$\hnup_i$ (\cref{line:level-cur-level}). 
To compute its $\hnup_i$, node
$i$ first computes $\nup_i$ using the most recently
published levels of its neighbors (\cref{line:compute-up}) where $\nup_i$ is the number of its
neighbors in level $\lcur$. If this is the
first round of the algorithm, then all neighbors of $i$ are on its level, level $0$. 
Then, $i$ computes $\hnup_i \leftarrow
\nup_i + X$, where $X \sim \geom(\nparam)$ denotes
a noise drawn i.i.d.\ from the symmetric geometric distribution with parameter
$\nparaminside$ (\cref{line:compute-noisy-up}).
The curator moves $i$ \emph{up a level} (to level $\lcur + 1$) if
and only if the released bit is $1$ (i.e., when $\hnup_i > \upexp^{\gn(\lcur)}$ in 
\cref{line:move-up-condition,line:release-noisy-up,line:update-list}).
(The curator performs this step for each node in level $\lcur$.)
Otherwise, $i$ releases $0$ (\cref{line:release-noisy-bool}) and it stays in the same level.
Then, the curator publicizes a new list, $L_{\lcur+1}$, that contains the
new levels of each node (\cref{line:new-pub-list}). 
If the node does not move up, then its old level is included in $L_{\lcur+1}$ (\cref{line:old-level}). 
This repeats in subsequent rounds until we reach the final level $\numlevels - 1$.

After processing the final level, $\numlevels - 1$, the curator estimates the
core numbers of nodes using their levels. This computation is shown in~\cref{alg:estimate}.
We use Definition 3.14 of~\cite{LSYDS22} to calculate this
estimate. Intuitively, the core number estimate for node $i$ is calculated to be
$(1+\lf)^{g}$, where $g$ is the maximum group index, where $L_{\numlevels - 1}[i]$ is the
topmost level in group $g$ or is higher than the topmost level in group~$g$.

We provide some brief intuition for the privacy of this algorithm. 
Each node moves up at most $\numlevels - 1$ levels. Thus, there will be at most $\numlevels - 1$ rounds of
communication.
We show that the \emph{sensitivity} of $\nup_i$ for any $0\leq \lcur \leq \numlevels - 1$ of 
any node $i$ is $1$; not only that, but the sensitivity of the vector of these values 
is $2$ for edge-neighboring graphs. Hence,
we add sufficient noise each round to maintain LEDP using the geometric \conffull{mechanism.}{(mechanism \cref{def:sgd-mech}).} We show that by the
adaptive composition \conffull{theorem }{theorem~(\cref{thm:composition})}over $O(\log^2 n)$ rounds that our
algorithm is a $\eps$-LEDP algorithm. This is a simplification of 
our privacy proofs; full details can be found in \conffull{the full version of our paper.}{\cref{sec:privacy}.}

\paragraph{Core Number Estimation}
Our algorithms here use the core number estimation
algorithm presented in~\cref{alg:estimate}, which takes $\lcur$ 
as input and computes the estimate of the core number of 
all nodes according to $\lcur$. This estimate is denoted $\kest(i)$ for each node $i$.

\paragraph{Low Out-Degree Ordering} The ordering of the nodes
is determined first by sorting its final level (contained in $L_{\numlevels -1}$), from smallest level
to largest, and then breaking ties using the nodes' indices.

\begin{algorithm}[!t]
\conffull{
\small}{}
    \Fn{\FnCoreNumber{$L, \lambda, \lf$}}{
        \For{$i = 1$ to $n$}{
            $\kest(i)\leftarrow\upexpold^{\max\left(
            \left\lfloor\frac{L(i) + 1}{4\ceil{\log_{1+\psi} n}}\right\rfloor
            -1, 0\right)}$.
        }
        \return $\{(i, \kest(i)) : i \in [n]\}$.
	}
    \caption{\label{alg:estimate} Estimate Core Number~\cite{LSYDS22}}
\end{algorithm}

\conffull{}{
\subsection{LEDP Guarantees}\label{sec:privacy}

In this section, we prove that our algorithm is $\eps$-LEDP (see~\cref{def:LEDP}). 
For clarity, we specifically prove the $\eps$-edge LEDP of our \kc decomposition and related algorithms here. 
Later, we show that the algorithms presented here can be generalized to our locally-adjustable algorithms framework, which may significantly simplify the privacy analysis of these algorithms and similar algorithms.

\begin{theorem}\label{lem:incremental-LEDP}
    \cref{alg:insert} is $\eps$-LEDP.
\end{theorem}

\begin{proof}
    We show that~\cref{alg:insert} can be implemented so that it accesses its input only via local randomizers. 
    In each round $\lcur$, each node in level $\lcur$ computes its \emph{induced degree} (the number of its neighbors that are also at level $\lcur$) with added geometric noise in~\cref{line:release-noisy-up}.
    We can implement this step with the local randomizer $R$ which takes as input an adjacency list, $\adj$, a level $\lcur$,
    and a list of levels of all nodes, $L_\lcur$ and computes $R(\adj, \lcur, L_\lcur) 
    = \sum_{j \in \adj}[L_\lcur[j] = \lcur] + X$ where $X \sim \geom(\eps/(8\log^2n))$.
    This randomizer is run independently by each party $i$ on its adjacency list $\adj_i$ and the public information 
    $\lcur$ and $L_\lcur$. Algorithm $R$ uses the geometric mechanism to obtain the number of its neighbors on $\lcur$. For any
    two edge-neighboring adjacency lists, $\adj$ and $\adj'$, the corresponding sums $\sum_{j \in \adj}[L_\lcur[j] = \lcur]$
    and $\sum_{j \in \adj'} [L_\lcur[j] = \lcur]$ differ by $1$. Thus, this sum function has sensitivity $1$ and by
    the privacy of the geometric mechanism is $\frac{\eps}{8\log^2 n}$-private.
    \cref{line:move-up-condition,line:release-noisy-up,line:stay-same,line:release-noisy-bool} 
    post-processes the output of the geometric mechanism and so, by~\cref{thm:post-processing}, $R$ is
    a $\frac{\eps}{8\log^2 n}$-local randomizer by~\cref{lem:sgd-private}. Using this released information, the curator
    then computes a new public set of levels $L_{\lcur+1}$.
    
    For each edge $\{i, j\}$, we apply two $\frac{\eps}{8\log^2 n}$-local randomizers to each endpoint over $4\log^2 n$ rounds.
    By the adaptive composition theorem (\cref{thm:composition}), group privacy (\cref{thm:group-dp}) and because 
    privacy is preserved after post-processing (\cref{thm:post-processing}), 
    \cref{alg:insert} is $\eps$-LEDP because it satisfies~\cref{def:LEDP}.
\end{proof}

\subsection{Approximation Guarantees}\label{sec:approx-runtime}

In this section, we calculate the approximation factors produced by~\cref{alg:insert}.
We first note that \emph{in expectation}, we return a $(2+\coren, 0)$-approximation 
(where the additive error is $0$) for the core numbers for any constant $\coren > 0$.

\begin{lemma}\label{lem:expectation}
\cref{alg:insert} returns a $(2+\coren, 0)$-approximation for the 
core numbers and low out-degree ordering, in expectation.
\end{lemma}

\begin{proof}
First, for any random variable 
$X\sim \noise$, the expectation $\expect[X] = 0$ by definition of the symmetric geometric distribution.
Thus, $\expect[\hnup_i] = \nup_i$ for all $i$ and the lemma follows
from a slightly modified version of the proof of Lemma 5.13 of~\cite{LSYDS22} given in~\cref{app:ledp-approx}.
\end{proof}

In the remainder of this section, we show
that we also obtain an additive error
of $O\left(\frac{\log^3 n}{\eps}\right)$ where the error is bounded by $O\left(\frac{\log^3 n}{\eps}\right)$ \whp{}. 

First, we show that~\cref{alg:insert} maintains 
the following two invariants (which we prove in \cref{lem:invariants-hold}).
These invariants are crucial in both our bounds on the approximation factor as
well as our message complexity analysis.
For the following two invariants, let $Z_{\lcur}$ be the
set of nodes in levels $\lcur$ and
above and $V_{\lcur}$ be the set of nodes in level $\lcur$. Recall as before
that the bottommost level is level $0$ and the topmost level is $\numlevels - 1$.
All node levels are taken from $L_{\numlevels-1}$ in the
final computation of approximate core numbers.

\begin{invariant}[Degree Upper Bound]\label{inv:degree-1}
    If node $i \in V_{\lcur}$ and level $\lcur < \numlevels - 1$, then $i$
    has at most $\upexp^{\gn(\lcur)} + \approxfac + 1$
    neighbors in $Z_{\lcur}$ \whp{}.
\end{invariant}

\begin{invariant}[Degree Lower Bound]\label{inv:degree-2}
    If node $i \in V_{\lcur}$ and level $\lcur > 0$, then
    $i$ has at least $(1 + \lf)^{\gn(\lcur-1)} - \approxfac - 1$ neighbors in $Z_{\lcur -
    1}$ \whp{}.
\end{invariant}
Before we prove our algorithm maintains
these invariants, we first prove the following lemma for each
node $i$ that moves up a level.

\begin{lemma}\label{lem:node-move}
    Let $\nup_i$ be the number of neighbors of node $i$ in level $\lcur$ where the levels of nodes are 
    taken from $L_\lcur$.
    Each node $i$ that moves \emph{up} one level (from level $\lcur$ to level
    $\lcur + 1$) in~\cref{alg:insert}
    has $\nup_i \geq \upexp^{\gn(\lcur)} - \cl - 1$ \whp.
\end{lemma}

\begin{proof}
    Suppose for contradiction a node $i$ moves up one level (from level $\lcur$
    to level $\lcur + 1$) and has $\nup_i < \upexp^{\gn(\lcur)} - \cl - 1$ with
    probability $> \probfactor$. %
    This means
    that it sampled a noise $X \sim \noise$, where
    $X > \cl$ with probability $> \probfactor$. This is because
    $\nup_i + X$ must exceed the $\upexp^{\gn(\lcur)}$ threshold
    in order to move up. The
    probability that $i$ sampled noise $X > \cl$ from
    $\noise$ is $< \frac{1}{n^c}$
    by~\cref{lem:noise-whp-bound}.
    This is a contradiction so $X >
    \cl$ could \emph{not} have been sampled with probability greater than
    $\frac{1}{n^c}$.%
\end{proof}

On the flip side, we can also bound the degree values of nodes that
\emph{do not move}.

\begin{lemma}\label{lem:not-moving}
    Let $\nup_i$ be as defined in~\cref{lem:node-move}.
    Each node $i$ that does not move up from a level $\lcur$ in~\cref{alg:insert} after
    computing $\nup_i$ (where the levels of nodes used
    to compute $\nup_i$ are 
    taken from $L_\lcur$) has
    $\nup_i \leq \upexp^{\gn(\lcur)} + \cl + 1$ \whp{}.
\end{lemma}

\begin{proof}
    As in the proof of~\cref{lem:node-move}, we prove this lemma via
    contradiction. First, suppose for contradiction
    that $i$ is currently in 
    level $\lcur$ and does not move up, and $\nup_i
    > \upexp^{\gn(\lcur)} + \cl + 1$
    with probability $> \probfactor$.
    Node $i$ does not move up when $\nup_i >
    \upexp^{\gn(\lcur)} + \cl + 1$ if it chooses a noise $X \sim \geom(\nparam)$
    with value $X < -\cl$. The same calculation as that
    performed in the proof of~\cref{lem:node-move} shows that with probability
    $< 1/n^c$, node $i$ does not move up if $\nup_i >
    \upexp^{\gn(\lcur)} + \cl + 1$. This contradicts our assumption
    that $\nup_i > \upexp^{\gn(\lcur)} + \cl + 1$ and does not move up with
    probability $> \probfactor$ since we showed that this occurs with
    probability at most $\probfactor$.
\end{proof}

We use \cref{lem:node-move,lem:not-moving} to prove the following lemma, which
shows that~\cref{inv:degree-1} and~\cref{inv:degree-2} are 
maintained by our algorithm. 

\begin{lemma}\label{lem:invariants-hold}
    Given an input graph $G = ([n], E)$, \cref{inv:degree-1} and~\cref{inv:degree-2} are satisfied at the end of~\cref{alg:insert} \whp.
\end{lemma}

\begin{proof}
    For each node $i$ in the graph that is in the current level $\lcur$, a noise $X_i$
    is independently sampled $X_i \sim \draw$ at the beginning of that round.
    We first make a few observations for when a node $i$ is in level $0$ or $\numlevels - 1$.
    If $i$ did not move up from level $0$, then~\cref{inv:degree-2} is trivially satisfied. 
    If $i$ moves up to level $\numlevels - 1$, then~\cref{inv:degree-1} is trivially satisfied.
    By~\cref{lem:not-moving}, any $i$ which does not move to a higher level than $\lcur \in [0, \numlevels-1)$
    has $\nup_i \leq \upexp^{\gn(L_\lcur[i]))} + \cl + 1$ \whp{} in round $\lcur$.
    Furthermore, $|\adj_i \cap Z_{L_{\numlevels-1}[i]}|$ cannot be greater than
    $\nup_i$ since only nodes in level $\lcur$ can move to a higher level than $\lcur$ in round $\lcur$. 
    Thus, $|\adj_i \cap Z_{L_{\numlevels-1}[i]}| \leq \upexp^{\gn(L_\lcur[i]))} + \cl + 1$ and
    \cref{inv:degree-1} is satisfied in this case. To prove that~\cref{inv:degree-2} is also satisfied for
    any node in level $\lcur \in (0, \numlevels-1]$ at the end of the algorithm,
    first observe that because nodes only move up in our algorithm, never down, $|\adj_i \cap Z_{\lcur'-1}|$
    can also never decrease for \emph{any node that stays in $r'$} as its neighbors move up for any level $\lcur'$. Thus, when $i$ 
    moved up from $\lcur - 1$ to $\lcur$, we have $\nup_i \geq
    \upexp^{\gn(\lcur - 1)} - \cl - 1$ with probability $\geq 1 -
    \probfactor$ by~\cref{lem:node-move}. Since $|\adj_i \cap Z_{\lcur-1}|$ cannot decrease if $i$ stays in level $r$ and we know that $i$ is in level $r$ at the end of the algorithm, it is lower bounded by
    $|\adj_i \cap Z_{L_{\numlevels-1}[i]}| = |\adj_i \cap Z_{\lcur-1}| = \nup_i \geq \upexp^{\gn(\lcur - 1)} - \cl - 1$ and~\cref{inv:degree-2} is also satisfied.
    This proves that both invariants hold for all nodes in all levels with probability at least $1- \frac{1}{n^d}$ 
    for any constant $d \geq 1$ by the union bound over 
    $\numlevels$ levels and $n$ nodes when $c\geq 1$ is a sufficiently large constant. 
\end{proof}

The theorem below is the main result of this section. The proof method is
inspired by the proof of the approximation factor given in the algorithm of Liu
\etal~\cite{LSYDS22}. However, our proof is more involved because of the
added noise and requires new insights in order to obtain our additive
approximation bound. First, our algorithm is randomized. Since we move each node at most $O(\log^2 n)$ times,
we can show via the union bound that our guarantees hold with high probability. 
Second, the additive noise from~\cref{inv:degree-1} and~\cref{inv:degree-2} 
requires us to re-perform the analysis taking this noise into account.

We are now ready to prove the approximation factor of our LEDP
algorithm in this section.
The proof of~\cref{lem:approximation-k-core} uses
the following folklore lemma, the proof of which follows immediately from the
traditional peeling method which provides the exact core number of each node in a
static graph. Note that although some previous literature
use the term \emph{peel} for the below process, we instead
use the phrase \emph{prune} to distinguish from the classic
peeling algorithm of Matula and Beck~\cite{Matula83}.

\begin{lemma}[Folklore]\label{lem:folklore}
    Suppose we perform the following \emph{peeling} procedure. Provided an input
    graph $G = (V, E)$, we iteratively remove (prune) nodes with degree $\leq d^*$ for
    some $d^* > 0$ until no nodes with degree $\leq d^*$ remain. Then, $\core(i) \leq d^*$ for each removed node $v$
    and all remaining nodes $w$ (not pruned) have core number $\core(w) > d^*$.
\end{lemma}

\begin{theorem}\label{lem:approximation-k-core}
    Our static LEDP algorithm (\cref{alg:insert}) on input
    $G=([n], E)$ gives
    $(2+\coren, O(\log^3 n/\eps))$-approximate
    core numbers \whp.
\end{theorem}

\begin{proof}
    In this proof, when we refer to the level of a node $i$, we mean the level of $i$
    in $L_{\numlevels-1}$. Using notation from previous work,
    let $\kest(i)$ be the core number estimate of $i$ and $\core(i)$ be the
    core number of $i$. First, we show that 
    \begin{align}
        \text{if } \kest(i) \leq (2+\lambda)(1+\lf)^{g'}, \text{ then }
        \core(i) \leq (1+\lf)^{g' + 1} + \frac{8c\log^3 n}{\eps} + 1\label{eq:kcore-bound-1}
    \end{align}
    with probability at least $1 - \frac{1}{n^{c-1}}$ for any group $g'
    \geq 0$ and for any constant $c \geq 2$.
    Let $T(g')$ be the topmost level of group $g'$. In order
    for $(2+\lambda)(1+\psi)^{g'}$ to be the estimate of $i$'s core number, 
    the level of $i$ is bounded by $T(g') \leq \level(i) \leq T(g' + 1) - 1$.
    Let $\lcur$ be $i$'s level (in $L_{\numlevels - 1}$).
    By~\cref{inv:degree-1}, if $\lcur < T(g' + 1)$,
    then $|\adj_i \cap Z_{\lcur}| \leq\upexp^{\gn(\lcur)} + \cl + 1 \leq
    \upexp^{g' + 1} + \cl + 1$ with probability at least $1- \frac{1}{n^c}$ for any constant
    $c \geq 1$. Furthermore, each node $w$ at the same or lower level $\lcur' \leq \lcur$ 
    has $|\adj_w \cap Z_{\lcur'}| \leq (1+\lf)^{g' + 1} + \cl + 1$, also by~\cref{inv:degree-1}.
    
    Suppose we perform the following iterative procedure:
    starting from level $\lcur = 0$,
    remove all nodes in level $\lcur$ during this turn and set $\lcur
    \leftarrow \lcur + 1$ for the next turn. Using this procedure, the nodes in
    level $0$ are removed in the first turn, the nodes in level $1$ are
    removed in the second turn, and so on until the graph is empty. 
    Let $d_{\lcur}(i)$ be the induced degree of
    any node $i$ after the removal in the $(\lcur-1)$-st turn and 
    prior to the removal in the $\lcur$-th turn. Since we showed above
    that $|\adj_i \cap Z_{\lcur}| \leq (1+\lf)^{g' + 1} + \cl + 1$ for any node
    $i$ at level $\lcur < T(g' + 1)$, node $i$ on level $\lcur < T(g' + 1)$
    during the $\lcur$-th turn has $d_{\lcur}(v) \leq \upexp^{g' + 1} + \cl + 1$.
    Thus, when $i$ is removed in the $\lcur$-th turn, it has degree $\leq
    \upexp^{g' + 1} + \cl + 1$. Since all nodes removed before $i$
    also had degree $\leq \upexp^{g' + 1} + \cl + 1$ when they were removed,
    by~\cref{lem:folklore}, node $i$
    has core number $\core(i) \leq \upexp^{g' + 1} + \cl + 1$ with probability $\geq 1
    - \probfactor$. By the union bound over all nodes,
    this proves that $\core(i) \leq (1+\lf)^{g' + 1} + \frac{8c\log^3 n}{\eps} + 1$ for all $i \in [n]$ 
    with $\kest(i) \leq (2+\lambda)(1+\psi)^{g'}$
    for all constants $c \geq 2$ with probability at least $1 - \frac{1}{n^{c-1}}$.
    
    Now we prove our lower bound on $\kest(i)$.
    We prove that for any $g' \geq 0$, either the approximation falls within our additive bounds or,
    \begin{align}
    \text{if } \kest(i) \geq (1+\lf)^{g'}, \text{ then }
    \core(i) \geq \frac{(1+\lf)^{g'} - \cl - 1}{\upexpold}\label{eq:kcore-bound-2}
    \end{align}
    for all nodes $i$
    in the graph with probability at least $1 -
    \probfactorminusone$ for any constant $c \geq 2$.
    We assume for contradiction 
    that with probability $> \frac{1}{n^{c-1}}$
    there exists a node $i$ where $\kest(i) \geq (1+\lf)^{g'}$ and $\core(i) <
    \frac{(1+\lf)^{g'} - \cl - 1}{\upexpold}$.
    This, in turn, implies that $\core(i) \geq \frac{\downexp^{g'} - \cl -
    1}{\upexpold}$ for all $i \in [n]$ and $\kest(i) \geq (1+\lf)^{g'}$ 
    with probability $< 1 - \probfactorminusone$.
    To consider this case, we use the \emph{pruning} process defined
    in~\cref{lem:folklore}. In the below proof, let $d_S(i)$ denote the 
    induced degree of node $i$ in the subgraph induced by nodes in $S$.
    For a given subgraph $S$, we \emph{prune} $S$ 
    by repeatedly removing all nodes $i \in S$ whose $d_S(i)
    < \frac{(1+\lf)^{g'} - \cl - 1}{\upexpold}$. As in the proof of Lemma 5.12
    of~\cite{LSYDS22}, we consider levels from the same group $g'$ since levels 
    in groups lower than $g'$ will also have smaller upper bound cutoffs, leading to an easier proof. 
    Let $j$ be the number of levels below level $T(g')$. We prove via induction that the
    number of nodes pruned from the subgraph induced by $Z_{T(g') - j}$ must
    be at least
    \begin{align}
        \left(\frac{\upexpold}{2}\right)^{j-1}\left(\lbexp\right)\label{eq:pruned}.
    \end{align}

    We first prove the base case when $j = 1$. In this case, we know that for any node $i$ in level $T(g')$,
    it holds that $d_{Z_{T(g') - 1}}(i) \geq (1+\lf)^{g'} - \cl - 1$ with probability
    $\geq 1 - \probfactor$ by~\cref{inv:degree-2}. Taking the union bound over all nodes in level $T(g')$
    shows that this bound holds for all nodes $i \in [n]$ with probability at least $1 - \frac{1}{n^{c-1}}$.
    All below expressions hold with probability at least $1 - \frac{1}{n^{c-1}}$; for simplicity, we omit
    this phrase when giving these expressions.
    In order to prune $i$ from the graph, we must prune at least
    \begin{align*}
        &\left((1+\lf)^{g'} - \cl - 1\right) -
        \frac{(1+\lf)^{g'} - \cl - 1}{\upexpold}\\
        &= \left(\downexp^{g'} - \cl - 1\right) \cdot \left(1 -
            \frac{1}{\upexpold}\right).
    \end{align*}
    neighbors of $i$ from $Z_{T(g') - 1}$. We must prune at least this many neighbors
    in order to reduce the degree of $i$ to below the cutoff for pruning a vertex (as we show more formally below).
    In the case when $\downexp^{g'} \leq 4(\cl - 1)$, then our original approximation statement in the lemma
    is trivially satisfied because the core number is always non-negative and so even if $\core(i) = 0$, 
    this is still within our additive approximation bounds.
    Hence, we only need to prove the case when $\downexp^{g'} > 4\left(\cl + 1\right)$.

    Then, if fewer than $\lbexp$
    neighbors of $i$ are pruned from the graph, then $i$ is not pruned from the
    graph. If $i$ is not pruned from the graph, then $i$ is part of a $\left(
    \frac{(1+\lf)^{g'} - \cl - 1}{\upexpold}\right)$-core (by~\cref{lem:folklore}), then by what we showed above using~\cref{inv:degree-2}, $\core(i) \geq
    \frac{(1+\lf)^{g'} - \cl - 1}{\upexpold}$ for all $i \in [n]$ and $\kest(i) \geq (1+\lf)^{g'}$
    with probability $\geq 1 - \frac{1}{n^{c-1}}$,
    a contradiction with our assumption. Thus, it must be
    the case that there exists at least one $i$ where at least $\lbexp$
    neighbors of $i$ are pruned in $Z_{T(g') - 1}$.
    For our induction hypothesis, we
    assume that at least the number of nodes as indicated in~\cref{eq:pruned}
    is pruned for $j$ and prove this for $j + 1$.

    Each node $w$ in levels $T(g') - j$ and above has
    $d_{Z_{T(g') - j - 1}}(w) \geq
    (1+\lf)^{g'} - \cl - 1$ by~\cref{inv:degree-2} (recall that all $j$ levels
    below $T(g')$ are in group $g'$). For simplicity of expression,
    we denote $\lbabrv \triangleq \lbexp$.
    Then, in order to prune the
    $\left(\frac{\upexpold}{2}\right)^{j-1}\lbabrv$
    nodes by our induction hypothesis, we must prune at least
    \begin{align}
        &\left(\frac{\upexpold}{2}\right)^{j-1}\lbabrv \cdot
        \left(\frac{(1+\lf)^{g'} - \cl - 1}{2}\right)
        \label{eq:pruned-edges}
    \end{align}
    edges where we ``charge'' the edge to the endpoint that is pruned last. We use the phrase \emph{charge} to mean
    that if an edge $(u, v)$ is charged to endpoint $v$, then
    the edge needs to be pruned in order for the 
    endpoint $v$ to be pruned.
    (Note that we actually need to prune at least 
    $\left(\downexp^{g'} - \cl - 1\right) \cdot \left(1 -
    \frac{1}{\upexpold}\right)$ edges per pruned node
    as in the base case but 
    $\frac{\left(\downexp^{g'} - \cl - 1\right)}{2}$ lower bounds this amount.)
    Each pruned node prunes less than
    $\frac{(1+\lf)^{g'} - \cl - 1}{\upexpold}$
    edges. Thus, using~\cref{eq:pruned-edges}, the number of nodes
    that must be pruned from $Z_{T(g') - j - 1}$ is
    \begin{align}
        \left(\frac{\upexpold}{2}\right)^{j-1} \lbabrv \cdot
        \frac{(1+\lf)^{g'} - \cl - 1}{2\left(\frac{(1+\lf)^{g'}-\cl-1}
        {\upexpold}\right)} = \left(\frac{\upexpold}{2}\right)^j \lbabrv.
        \label{eq:induction}
    \end{align}
    \cref{eq:induction} proves our induction step.
    Using~\cref{eq:pruned}, the number of nodes that must be pruned from
    $Z_{T(g') - 2\log_{\upexpold/2}\left(n\right)}$ is
    greater than $n$ since $J \geq 1$ by our assumption that $\downexp^{g'} > 4\left(\cl + 1\right)$:
    \begin{align}
        \left(\frac{\upexpold}{2}\right)^{2\log_{\upexpold/2}\left(n\right)} \cdot \lbabrv \geq n^2.
        \label{eq:final-eq}
    \end{align}
    Thus, at $j =
    2\log_{\upexpold/2}\left(n\right)$, we run out of
    nodes to prune. We have reached a contradiction as we require pruning greater than 
    $n$ nodes with probability at least $\frac{1}{n^{c-1}} \cdot \left(1- \probfactorminusone\right) > 0$ 
    via the union bound over all nodes where $\kest(i) \geq (1+\lf)^{g'}$ and using our assumption that
    with probability $> \frac{1}{n^{c-1}}$ there exists an $i \in [n]$ 
    where $\core(i) < \frac{(1+\lf)^{g'} - \cl - 1}{\upexpold}$. 
    This contradicts with the fact that
    more than $n$ nodes can be pruned with $0$ probability.
    
    From~\cref{eq:kcore-bound-1}, we can first obtain the inequality $\core(i) \leq \kest(i) + \cl + 1$ since
    this bounds the case when $\kest(i) = (2+\lambda)(1+\psi)^{g'}$; if $\kest(i) < (2+\lambda)(1+\psi)^{g'}$
    then the largest possible value for $\kest(i)$ is $(2+\lambda)(1+\psi)^{g'-1}$ by~\cref{alg:estimate} and
    we can obtain the tighter bound of $\core(i) \leq (1+\psi)^{g'} + \frac{8c\log^3 n}{\eps} + 1$. 
    We can substitute $\kest(i) = 
    (2+\lambda)(1+\psi)^{g'}$ since $(1+\psi)^{g' + 1} < (2+\lambda)(1+\psi)^{g'}$ for all $\psi \in (0, 1)$ and 
    $\lambda > 0$. Second, from~\cref{eq:kcore-bound-2},
    for any estimate $(2+\lambda)(1+\psi)^{g}$, the largest $g'$ for which this estimate
    has $(1+\psi)^{g'}$ as a lower bound is $g' = g + \floor{\log_{(1+\psi)}(2+\lambda)} \geq g 
    + \log_{(1+\psi)}(2+\lambda) -1$. Substituting this $g'$ into $\frac{(1+\lf)^{g'} - \cl - 1}{\upexpold}$ 
    results in 
    \begin{align*}
        \frac{(1+\lf)^{g + \log_{(1+\psi)}(2+\lambda) -1} - \cl - 1}{\upexpold} = \frac{\frac{(2+\lambda)(1+\lf)^{g}}{1+\lf} - \cl - 1}{\upexpold} = \frac{\frac{\kest(i)}{1+\lf} - \cl - 1}{\upexpold}.
    \end{align*}
    Thus,
    we can solve $\core(i) \geq \frac{\frac{\kest(i)}{1+\lf}-\cl - 1}{(2+\lambda)(1+\lf)}$
    and $\core(i) \leq \kest(i) + \cl + 1$ to obtain
    $\core(i) - \cl - 1 \leq \kest(i) \leq
    (\upexpold^2)\core(i) + \frac{8c(1+\psi)\log^3 n}{\eps} + (1+\psi)$. Simplifying, we obtain
    \begin{align*}
        \core(i) - &O(\onoise) \leq \kest(i) \leq 
        \upexpold^2 \core(i) + O(\onoise)
    \end{align*}
    which is consistent with the definition of
    a $(2+\const, O(\onoise))$-approximation algorithm
    for core number for any constant $\coren > 0$ and appropriately
    chosen constants $\lambda, \lf \in (0, 1)$ that depend on $\coren$.
\end{proof}

}

\conffull{}{
\subsection{Reducing the Number of Rounds}\label{sec:reducing-rounds}
In this section, we describe how to reduce the number of rounds of our procedure given in~\cref{alg:insert} to $O(\log n)$.
Recall that the previous algorithm given in~\cref{alg:insert} consists of $O(\log n)$ groups each with a set of $O(\log n)$
levels that are stacked on top of one another. 
The intuition behind our modified procedure lies in separating out the groups in our level data structure into separate 
duplicate copies of the structure, each with $O(\log n)$ levels.
Then, the procedure moves each node \emph{simultaneously} up the levels in all groups. This increases the 
node communication complexity by a factor of $O(\log n)$. A similar procedure to the above is given for 
our densest subgraph algorithm in~\cref{alg:ledp-densest}. We describe our modified algorithm in
\cref{alg:separate-groups}.

\SetKwFunction{FnNewCoreNumber}{EstimateSmallRoundsCoreNumbers}
\SetKwFunction{FnNewCoreDecomp}{LEDPSmallRoundsCoreDecomposition}
\begin{algorithm}[!t]
\conffull{
\small}{}
    \textbf{Input:} Adjacency lists $(\adj_1, \dots, \adj_n)$, constant $\eta \in (0, 1)$, and
    privacy parameter $\eps \in (0, 1)$.\\ 
    \textbf{Output:} $\eps$-LEDP $\left(2+\coren, O\left(\frac{\log^3 n}{\eps}\right)\right)$-approximate core numbers of each node in $G$.\\
    \Fn{\FnNewCoreDecomp{$(\adj_1, \dots, \adj_n), \eps, \eta$}}{
        Set $\psi = 0.5$ and $\lambda = \frac{2}{9} (2 \coren - 5)$.\\ %
        Curator initializes $L^0_0, L^0_1, \dots, L^0_{\numgrouplevels - 1}, \dots, L^{2\log n - 1}_0, \dots, 
        L^{2\log n - 1}_{\numgrouplevels -1}$ with $L^g_\lcur[i] \leftarrow 0$ for every $i \in [n]; \lcur, g \in [0, \dots, \numgrouplevels - 1]$.\label{lkcsr-line:initial}\\
        \For{$\lcur = 0$ to $\numgrouplevels - 2$}{\label{lkcsr-line:iterate}
            \For{$i = 1$ to $n$}{\label{lkcsr-line:node}
                Initialize $A_i[g] \leftarrow 0$ for every $g \in [0, \dots, \numgrouplevels - 1]$.\label{lkcsr-line:level-array}\\
                \For{$g = 0$ to $\numgrouplevels - 1$}{\label{lkcsr-line:group-iterate}
                    $L^g_{\lcur + 1}[i] \leftarrow L^g_{\lcur}[i]$.\label{lkcsr-line:old-level}\\
                    \If{$L^g_\lcur[i] = \lcur$}{\label{lkcsr-line:level-cur-level}
                        Let $\nup_{i, g}$ be the number of neighbors $j \in \adj_i$ where $L^g_\lcur[j] =
                        \lcur$.\label{lkcsr-line:compute-up}\\
                        Sample $X \leftarrow \geom(\nparam)$.\label{lkcsr-line:sample-noise}\\
                        Compute $\hnup_{i, g} \leftarrow \nup_{i, g} + X$.\label{lkcsr-line:compute-noisy-up}\\
                        \If{$\hnup_{i, g} > \upexp^{g}$}{\label{lkcsr-line:move-up-condition}
                            $A_i[g] \leftarrow 1$.\label{lkcsr-line:group-update}\\
                        }
                    }
                }
                $i$ \release $A_i$.\label{lkcsr-line:release-ai}\\
                $L^g_{\lcur + 1}[i] \leftarrow L^g_{\lcur}[i] + 1$ for every $i, g$ where $A_i[g] = 1$.\label{lkcsr-line:update-list} \Comment{Curator moves $i$ up in group $g$.}
            }
            
            Curator publishes $L^g_{\lcur + 1}$ for every $g \in [0, \dots, \numgrouplevels-1]$.\label{lkcsr-line:levels}\\
        }
        Curator calls $C \leftarrow \FnNewCoreNumber(L^0_{\numgrouplevels-1}, \dots, L^{2\log n - 1}_{\numgrouplevels -1}, \lambda, \psi)$ [\cref{alg:small-rounds-estimate}].\label{lkcsr-line:core-number-estimate}\\
        \Return $C$.\label{lkcsr-line:return-noisy-density}
        \caption{\label{alg:separate-groups} $\eps$-LEDP $k$-Core Decomposition (Smaller Rounds)}
    }
\end{algorithm}

\begin{algorithm}[!t]
\conffull{
\small}{}
    \Fn{\FnNewCoreNumber{$L^0_{\numgrouplevels -1}, \dots, L^{2\log n - 1}_{\numgrouplevels -1}, \lambda, \lf$}}{
        \For{$i = 1$ to $n$}{
            Let $g'$ be the largest group number where $L^{g'}_{\numgrouplevels-1}[i] = \numgrouplevels-1$ or $0$ if 
            no group satisfies this condition.\label{newcoren:max-g}\\
            $\kest(i)\leftarrow\upexpold^{g'}$.\label{newcoren:compute-estimate}
        }
        \return $\{(i, \kest(i)) : i \in [n]\}$.\label{newcoren:return-estimate}
	}
    \caption{\label{alg:small-rounds-estimate} Small Rounds Estimate Core Number}
\end{algorithm}

In \cref{alg:separate-groups}, the curator first initializes a list for each level and group (\cref{lkcsr-line:initial}) which contains
the levels of each node in each of the groups (\cref{lkcsr-line:level-cur-level}) with cutoffs determined by the group number. 
Then, we iterate level by level (\cref{lkcsr-line:iterate}) for each of the $O(\log n)$ groups simultaneously.
For each level $r$ in the iteration, each node $i$ (\cref{lkcsr-line:node}) initializes an array $A_i$ which keeps track of 
whether $i$ moves up a level in group $g$ (\cref{lkcsr-line:level-array}). As before, the curator first sets the next levels of 
each node in each of the groups to equal the current level of the node (\cref{lkcsr-line:old-level}).
If $i$ is in the current level $\lcur$ of the iteration in group $g$ (\cref{lkcsr-line:level-cur-level}), then $i$ computes
its induced degree among its neighbors also in level $r$ using the published levels 
from the previous iteration (\cref{lkcsr-line:compute-up}).
Then, it samples a noise from the symmetric geometric distribution (\cref{lkcsr-line:sample-noise}).
If the induced degree plus the sampled noise exceeds the threshold for the group (\cref{lkcsr-line:move-up-condition}), then
$i$ adds a $1$ for the corresponding group $g$ into $A_i$, 
indicating that $i$ moves up a level in group $g$ (\cref{lkcsr-line:group-update}).
After performing this check for all of the groups, $i$ then releases $A_i$ (\cref{lkcsr-line:group-update}).
The curator then uses the released $A_i$ to determine if $i$ moves up one level in each group (\cref{lkcsr-line:update-list}) 
and publishes the list of new levels for the next iteration (\cref{lkcsr-line:levels}). After all the rounds are completed,
the curator uses all of the published levels for the \emph{last iteration} for each of the groups to 
compute the core number estimates (\cref{lkcsr-line:core-number-estimate}) using~\cref{alg:small-rounds-estimate}.
In~\cref{alg:small-rounds-estimate}, each $i$ computes the largest $g'$ where $L_{\numgrouplevels-1}^{g'}[i] = \numgrouplevels-1$
(\cref{newcoren:max-g}); in other words, each $i$ finds the largest $g'$ where $i$ is on the topmost level of group $g'$
after finishing all the iterations of the algorithm. Using $g'$, node $i$ computes 
the estimate $\hat{k}(i)$ as before (\cref{newcoren:compute-estimate}). The core estimates for all nodes are then returned
by the procedure (\cref{newcoren:return-estimate}). 

The privacy of~\cref{alg:separate-groups} follows using the same
local randomizers as in the proof of~\cref{lem:incremental-LEDP}.
The proof of the approximation factor follows almost identically from the 
proof of~\cref{lem:approximation-k-core} using the modified invariants below (changes are underline). 
Since only minor changes to the approximation proofs
are necessary, we do not repeat the proofs again here.
Let $V^g_{\lcur}$ be the set of nodes $i$ where $L^g_{\numgrouplevels-1}[i] = \numgrouplevels-1$
and $Z^g_{\lcur} = \bigcup_{r' \geq \lcur}V^g_{\lcur'}$.

\begin{invariant}[Degree Upper Bound]\label{inv:degree-sr-1}
    \underline{For all groups $g \in \{0, \dots, \numgrouplevels-1\}$}, 
    if node $i \in \underline{V^g_{\lcur}}$ and level $\lcur < \underline{\numgrouplevels} - 1$, then $i$
    has at most $\underline{\upexp^{g}} + \approxfac + 1$
    neighbors in \underline{$Z^g_{\lcur}$} \whp{}.
\end{invariant}

\begin{invariant}[Degree Lower Bound]\label{inv:degree-sr-2}
    \underline{For all groups $g \in \{0, \dots, \numgrouplevels-1\}$}, if node $i \in \underline{V^g_{\lcur}}$ 
    and level $\lcur > 0$, then
    $i$ has at least $\underline{(1 + \lf)^{g}} - \approxfac - 1$ neighbors in \underline{$Z^g_{\lcur -
    1}$} \whp{}.
\end{invariant}

\subsection{Low Out-Degree Ordering and Densest Subgraph}\label{app:other-graph-quants}
In this section, using~\cref{alg:insert},
we provide results for low out-degree ordering of the nodes in the graph.
The procedure that outputs the low out-degree ordering is provided in~\cref{line:ordering} of~\cref{alg:insert}. 
The ordering that we study in this paper is a private, approximate degeneracy ordering of the nodes.
Recall in~\cref{alg:insert}, we calculate the approximate low out-degree ordering, 
$D$, by ordering the nodes starting from level $0$ and going up. 
For nodes in the same level, we order the nodes according to their index. 
In other words, node $i$ is earlier in the ordering
than $j$ if and only if $L_{\numlevels-1}(i) < L_{\numlevels-1}(j)$ or $L_{\numlevels-1}(i) = L_{\numlevels-1}(j)$
and $i < j$.

We now prove that our algorithm is also a
$(2 + \const, O(\onoise))$-approximation algorithm for low
out-degree ordering.}

\conffull{}{
\begin{lemma}\label{lem:approximation-outdeg}
    Our static LEDP algorithm (\cref{alg:insert}) gives
    an $(2+\coren, O(\onoise))$-approximate
    low out-degree ordering with probability at least
    $1 - \probfactor$ for any constants $\const > 0, c \geq 1$.
\end{lemma}

\begin{proof}
    By~\cref{inv:degree-1}, if we order the nodes as described
    in~\cref{alg:insert}, then each node $i$ has degree at most $\upexp^{\gn(L_{\numlevels-1}[i])} + O(\onoise)$.
    By~\cref{lem:approximation-k-core}, the approximation on $\degen = \max_{i \in [n]} \{\core(i)\}$
    is upper bounded by $(2+\lambda)(1+\psi)^2 d + O\left(\onoise\right)$
    Then, the highest group that any node $j$ can be
    in has index $$g = \log_{\downexp}((2+\lambda)(1+\psi)^2 \degen + O(\onoise)) + 1,$$
    where the additional $1$ at the end comes from the fact that we take the maximum
    group number $g$ where the topmost level of $g$ is below or at $L_{\numlevels-1}[j]$.
    Within group $g$, the out-degree of any node in that group is upper
    bounded by $\upexp^{g} + O(\onoise)$ with high probability by~\cref{inv:degree-1}. 
    
    Simplifying by substituting the value for $g$,
    \begin{align*}
        &\upexp^{\log_{\downexp}\left((2+\lambda)\upexp^2 \degen + O(\onoise)\right) + 1} + O(\onoise)\\ 
        &= (1+\psi)\left((2 + \lambda)\upexp^2 \degen + O(\onoise)\right) + O(\onoise)\\
        &= (2+\lambda)(1+\psi)^3 \degen + O(\onoise).
    \end{align*}
    Finally, the last expression gives a $(2+\coren)\degen + O(\onoise)$ approximation for any constant $\coren > 0$
    for appropriately chosen $\lambda$ and $\psi$.
    This gives a $(2+\const, O(\log^3
    n/\eps))$-approximation for the orientation produced
    by the low out-degree ordering.
\end{proof}

}

\paragraph{Densest Subgraph}
A straightforward extension of~\cref{alg:insert} where nodes move up levels in all groups simultaneously
also leads to a $\eps$-LEDP approximate densest subgraph algorithm which is derived from the non-private
algorithm of~\cite{BHNT15} that finds
an approximate densest subgraph by peeling the layers of the level data structure 
one by one for each group $g$ and taking the subgraph with largest density.
\conffull{}{
The privacy for the densest subgraph algorithm holds by~\cref{lem:incremental-LEDP} and because
privacy is maintained after post-processing (\cref{thm:post-processing}).} Our algorithm is 
given in~\cref{alg:ledp-densest} and the proof of our approximation bound follows from the proofs
of Theorem 2.6 and Corollary 2.7 of Bhattacharya \etal~\cite{BHNT15} with minor modifications.

First, we describe a few key points of~\cref{alg:ledp-densest}. The algorithm operates over 
$O(\log n)$ rounds (\cref{lds-line:iterate}) where in each round, each node (\cref{lds-line:node})
computes its noisy degrees for each of the $O(\log n)$ groups (\cref{lds-line:group-iterate}).
For each group, each node $i$ computes a noisy degree (if it is in the current level $\lcur$)
and releases this noisy degree (\cref{lds-line:release-noisy-up}); this is in contrast to~\cref{alg:insert},
where $i$ releases either $1$ or $0$. The noisy degree in this setting is used to compute the 
approximate density. Finally,~\cref{lds-line:compute-subgraph} is performed by the curator who 
has all of the released degrees of every node $i \in [n]$ in each of the $O(\log n)$ levels and $O(\log n)$
groups. Thus, the curator can successively peel levels from the smallest to largest level while
computing the sum of the noisy degrees of all nodes that remain after the most recently peeled level.
Using this sum of noisy degrees as well as the public set of levels of each node, the curator can
produce a subset of nodes whose induced subgraph is an approximate densest subgraph as follows.

The curator finds the group $g$ with the largest index that has a non-empty last level (using $L^g_{\numgrouplevels-1}$
for every $g \in \{0, \dots, \numgrouplevels-1\}$).
Taking the noisy degrees of all nodes released for group $g$ and using $L_{\numgrouplevels-1}^g$, the curator
peels the levels one by one from the smallest level to the largest level while maintaining the sum of the 
degrees divided by the number of nodes that remain after each round of peeling. 
The curator uses the released noisy degrees corresponding with $L^g_{\numgrouplevels-1}[i]$ for each $i \in [n]$ and 
$g \in \{0, \dots, \numgrouplevels-1\}$.
The curator keeps and returns as $S$ the subset of nodes whose ratio of the sum of the noisy degrees
over the number of nodes in the subset is the largest across all iterations of peeling.
\cref{lds-line:return-noisy-density} returns the noisy density of the set $S$ (scaled by an appropriate factor)
using the sum of the released noisy degrees.

\conffull{}{
The privacy of this algorithm follows from~\cref{lem:incremental-LEDP} and the fact that privacy is maintained
after post-processing (\cref{thm:post-processing}). We now prove the approximation guarantees of this algorithm.
}

\begin{algorithm}[!t]
\conffull{
\small}{}
    \textbf{Input:} Adjacency lists $(\adj_1, \dots, \adj_n)$, constants $\psi \in (0, 1)$, and
    privacy parameter $\eps \in (0, 1)$.\\ 
    \textbf{Output:} A private set of nodes whose induced subgraph is a
    $\left(4+\coren, O\left(\frac{\log^3 n}{\eps}\right)\right)$-approximate densest
    subgraph in $G$.\\
    \Fn{\FnLEDPDensestSubgraph{$(\adj_1, \dots, \adj_n), \eps, \psi$}}{
        Curator initializes $L^0_0, L^0_1, \dots, L^0_{\numgrouplevels - 1}, \dots, L^{2\log n - 1}_0, \dots, 
        L^{2\log n - 1}_{\numgrouplevels -1}$ with $L^g_\lcur[i] \leftarrow 0$ for every $i \in [n]; \lcur, g \in [0, \dots, \numgrouplevels - 1]$.\label{lds-line:initial}\\
        \For{$\lcur = 0$ to $\numgrouplevels - 2$}{\label{lds-line:iterate}
            \For{$i = 1$ to $n$}{\label{lds-line:node}
                \For{$g = 0$ to $\numgrouplevels - 1$}{\label{lds-line:group-iterate}
                    $L^g_{\lcur + 1}[i] \leftarrow L^g_{\lcur}[i]$.\label{lds-line:old-level}\\
                    \If{$L^g_\lcur[i] = \lcur$}{\label{lds-line:level-cur-level}
                        Let $\nup_{i, g}$ be the number of neighbors $j \in \adj_i$ where $L^g_\lcur[j] =
                        \lcur$.\label{lds-line:compute-up} \\
                        Sample $X \leftarrow \geom(\nparam)$.\\
                        Compute $\hnup_{i, g} \leftarrow \nup_{i, g} + X$.\label{lds-line:compute-noisy-up}\\
                        $i$ \release $\hnup_{i, g}$\label{lds-line:release-noisy-up}.\\
                        \If{$\hnup_{i, g} > \upexp^{g}$}{\label{lds-line:move-up-condition}
                            $L^g_{\lcur + 1}[i] \leftarrow L^g_{\lcur}[i] + 1$.\label{lds-line:update-list} \Comment{Curator moves $i$ up one level in group $g$.}\\
                        }
                    }
                }
            }
            
            Curator publishes $L^g_{\lcur + 1}$ for every $g \in [0, \dots, \numgrouplevels-1]$.\label{lds-line:levels}\label{lds-line:new-pub-list}\\
        }
        Curator uses released noisy degrees of all nodes
        and $L^g_{\numgrouplevels-1}$ for all $g \in [0, \dots, \numgrouplevels-1]$ to peel the levels
        one by one and determine and return the set of nodes $S$ whose induced subgraph is an
        approximate densest subgraph.\label{lds-line:compute-subgraph}\\
        Let $\hat{W}$ be the sum of the noisy degrees of $S$. \Return $(S, \frac{\hat{W}}{2|S|} - \frac{c\log^3 n}{\eps})$ 
        for sufficiently large $c\geq 1$.\label{lds-line:return-noisy-density}
        \caption{\label{alg:ledp-densest} $\eps$-LEDP Densest Subgraph}
    }
\end{algorithm}

\conffull{}{
\begin{lemma}\label{lem:densest-subgraph}
    There exists an $\eps$-edge LEDP algorithm on input
    $G=(V, E)$ that gives a set of nodes whose induced subgraph is a 
    $(4+\coren, O(\log^3 n/\eps))$-approximate
    densest subgraph in $O(\log n)$ rounds.
\end{lemma}

\begin{proof}
    First, we show that if $(1+\lf)^g > 2(1+\lf)D^* + 
    \frac{c\log^3 n}{\eps}$, then for all $i \in [n]$, it holds that
    $L^g_{\numgrouplevels-1}[i] < \numgrouplevels - 1$ \whp
    for appropriately large constant $c \geq 1$;
    in other words, the last level of group $g$ is empty
    \whp. As in the proof of Theorem 2.6 in~\cite{BHNT15}, we can show that
    for any level $\lvl \in \{0, \dots, \numgrouplevels - 1\}$ in group
    $g$ it holds that (let $\delta(G[Z_\lvl])$ be the sum of the degrees of 
    every node in the induced subgraph of $Z_\lcur$ divided by $|Z_\lcur|$):
    \begin{align}
    \delta(G[Z_\lcur]) = 2 \cdot \dense(G[Z_\lcur]) \leq 2 \cdot
    \max_{S \subseteq V} \dense(G[S]) = 2\cdot D^* <
    \frac{(1+\lf)^g - \frac{c\log^3 n}{\eps}}{1+\lf}.\label{eq:degree-bound}
    \end{align}
    From this, we know that the number of nodes in 
    $Z_\lcur$ with induced degree in $Z_\lcur$ at least 
    $(1+\lf)^g - \frac{c\log^3 n}{\eps}$ is at most 
    $\frac{|Z_\lvl|}{1+\lf}$ since otherwise, it violates
    \cref{eq:degree-bound}. Let $C_\lvl$ be the set of 
    nodes in $Z_\lcur$ with degree less than $(1+\lf)^g - 
    \frac{c\log^3 n}{\eps}$. Then, it follows that
    $|Z_\lvl \setminus C_\lvl| \leq \frac{|Z_\lvl|}{1+\lf}$.
    By~\cref{line:move-up-condition} and~\cref{lem:noise-whp-bound}, 
    we have $Z_{\lvl+1}\cap C_\lvl = \emptyset$ since the noise does not exceed $\frac{c\log^3 n}{\eps}$ \whp. 
    Thus, $|Z_{\lvl+1}| \leq |Z_\lvl \setminus C_\lvl| \leq \frac{|Z_\lvl|}{1+\lf}$. 
    Then, for all $\lvl \in \{0, \dots, \numgrouplevels-2\}$, we 
    have that $|Z_{\lvl+1}| \leq \frac{|Z_\lvl|}{1+\lf}$.
    Multiplying these inequalities gives us $|Z_{\numgrouplevels-1}| \leq \frac{|Z_1|}{(1+\lf)^{\numgrouplevels-1}}$. Since
    $|Z_1| \leq n$, we get
    $|Z_{\numgrouplevels-1}| \leq \frac{n}{(1+\lf)^{\numgrouplevels-1}} = \frac{(1+\lf)n}{n^2} < 1$ which means that 
    $Z_{\numgrouplevels - 1} = \emptyset$ \whp.
    
    Now, suppose that $(1+\lf)^g < \frac{D^*}{1+\lf} - \frac{c\log^3 n}{\eps}$, and let $S^* \subseteq V$
    be a subset of nodes with highest density, i.e.\ $\dense(G[S^*]) = D^*$. We will show that 
    $S^* \in Z_\lvl$ for all $\lvl \in \{0, \dots, \numgrouplevels-1\}$. 
    This means that $Z_{\numgrouplevels-1} \neq \emptyset$. We prove this via induction.
    Clearly, in the base case, since $S^* \subseteq V$, then $S^* \subseteq Z_0$. 
    For the induction hypothesis, we assume that $S^* \subseteq Z_\lvl$. Then, we prove that 
    $S^* \subseteq Z_{\lvl+1}$. By Lemma 2.4 of~\cite{BHNT15},
    for every node $i \in S^*$, we have that $\deg_{Z_\lvl}(i)
    \geq \deg_{S^*}(i) \geq \dense(G[S^*]) = D^* > (1+\lf) \cdot \left((1+\lf)^g + \frac{c\log^3 n}{\eps}\right)$. By~\cref{line:move-up-condition} and~\cref{lem:noise-whp-bound},
    the minimum value of the noise variable added to the degree of $i$ is $-\frac{c\log^3 n}{\eps}$ \whp
    and $(1+\lf)^{g+1}$ exceeds the cutoff of $(1+\lf)^g$; so $i$ is in level $\lvl + 1$.
    Thus, for all $i \in S^*$, node $i \in Z_{\lvl+1}$.
    This shows that the topmost level of group $g$ contains $S^*$. 
    
    Now, we show that 
    if $(1+\lf)^g < \frac{D^*}{1+\lf} - \frac{c\log^3 n}{\eps}$, then there is a level $\lvl \in \{0, \dots, \numgrouplevels-1\}$
    where $\dense(G[Z_\lvl]) \geq \frac{(1+\lf)^g}{2(1+\lf)} - \frac{c\log^3 n}{\eps}$. For the sake of contradiction,
    suppose this is not the case. Then, we have $\delta(G[Z_\lvl]) = 2 \cdot \dense(G[Z_\lvl]) < \frac{(1+\lf)^g}{1+\lf}
    - \frac{2c\log^3n}{\eps}$ for every $\lvl \in \{0, \dots, \numgrouplevels-1\}$. Then, applying a similar argument as
    for the first case (since this equation follows~\cref{eq:degree-bound} except for the constant on the additive noise), 
    we conclude that $|Z_{\lvl + 1}| \leq \frac{|Z_\lvl|}{1+\lf}$ for every $\lvl \in 
    \{0, \dots, \numgrouplevels-1\}$ which implies that $Z_{\numgrouplevels-1} = \emptyset$. We have arrived at a contradiction.
    
    The rest of the proof follows from a modification of the proof of Corollary 2.7 of~\cite{BHNT15}. We consider
    all groups $g \in \{0, \dots, \numgrouplevels-1\}$. We find the group with the largest index
    that has a non-empty last level and perform the peeling procedure. 
    Suppose $g$ is the group. By what we showed above, we know that 
    $\frac{D^*}{1+\lf} - \frac{c\log^3 n}{\eps} \leq (1+\lf)^g \leq 2(1+\lf)D^* + \frac{c\log^3 n}{\eps}$.
    By the above, if the last level is not empty in group $g$, then, 
    there exists an induced subgraph $Z_\lvl$ where $\lvl \in \{0, \dots, \numgrouplevels-1\}$, obtained
    from peeling the levels, that has density at least $\frac{(1+\lf)^g}{2(1+\lf)} - \frac{c\log^3 n}{\eps}$. 
    Finally, the returned noisy degrees for each node is at least its actual degree minus $\frac{c\log^3n}{\eps}$ \whp.
    Let $\lvl \in \{0, \dots, \numgrouplevels - 1\}$ be the level where $Z_{\lvl}$ has the highest
    density in  $g$. Let $\hat{W}_r$ be the sum of the noisy degrees of nodes in $Z_\lvl$.
    Then, we can bound $\hat{W}_r$ by the following expression in terms of $\delta(G[Z_\lvl])$ and
    $\dense(G[Z_\lvl])$,
    \begin{align*}
        \dense(G[Z_\lvl]) - \frac{c\log^3 n}{\eps}
        = \frac{\delta(G[Z_\lvl])}{2} - \frac{c\log^3 n}{\eps} \leq \hat{W}_r 
        \leq \delta(G[Z_\lvl]) + \frac{c\log^3 n}{\eps} = 2 
        \cdot \dense(G[Z_\lvl]) + \frac{c\log^3 n}{\eps}. 
    \end{align*}
    By the above expression, we are guaranteed to select an $\lvl'$ \whp where 
    \begin{align}
        \frac{\dense(G[Z_\lvl])}{2} - \frac{c\log^3 n}{\eps} = \frac{\delta(G[Z_{\lvl}])}{4} - \frac{c\log^3 n}{\eps} \leq
        \frac{\hat{W}_{r'}}{2} \leq \frac{\delta(G[Z_{\lvl'}])}{2} + \frac{2c\log^3 n}{\eps} = \dense(G[Z_{\lvl'}]) + \frac{2c\log^3 n}{\eps}
    \end{align}
    
    Combining all of the above observations, we get that the returned set of nodes in the approximate densest subgraph
    has density at least 
    \begin{align*}
        \dense(G[Z_{\lvl'}]) &\geq \frac{\dense(G[Z_\lvl])}{2} - \frac{3c\log^3 n}{\eps} \\
        &\geq \frac{1}{2} \cdot \left(\frac{(1+\lf)^g}{2(1+\lf)} - \frac{c\log^3 n}{\eps}\right) - \frac{3c\log^3 n}{\eps}\\
        &\geq \frac{\frac{D^*}{1+\lf} 
        - \frac{c\log^3 n}{\eps}}{4(1+\lf)}- \frac{7c\log^3 n}{2\eps} = \frac{D^*}{4(1+\lf)^2} 
        - O\left(\frac{\log^3n}{\eps}\right).
    \end{align*}
\end{proof}

\subsection{Complexity Measures}

In this section, we analyze the \emph{round and node communication complexity} of our $\eps$-edge 
LEDP algorithm; the latter is measured in terms of the size of messages sent by the nodes. 
Each node in~\cref{alg:insert} releases $O(1)$
bits of information in each round. Then, we show a simple modification to the algorithm 
that allows us to obtain $O(\log n)$ rounds but with an additional node communication complexity of $O(\log n)$ bits. 

\begin{theorem}\label{thm:comm}
    \cref{alg:insert} has round complexity
    $O(\log^2 n)$. Released messages contain $O(1)$ bits, with 
    high probability.
\end{theorem}

\begin{proof}
    Our algorithm requires $O(\log^2 n)$ rounds since we move vertices up at most $O(\log^2 n)$ levels. 
    During each round, to obtain the communication complexity of $O(1)$ bits, 
    each vertex outputs $O(1)$ bits indicating whether it is moving up a level or staying in the same level. 
\end{proof}

\begin{corollary}\label{cor:rounds}
    There exists a $\left(2+\coren, O\left(\frac{\log^3 n}{\eps}\right)\right)$-approximate 
    LEDP algorithm for $k$-core decomposition that takes $O(\log n)$ rounds and outputs messages with $O(\log n)$ bits,
    with high probability. There exists a $\left(4+\coren, O\left(\frac{\log^3 n}{\eps}\right)\right)$-approximate 
    LEDP algorithm for densest subgraph that takes $O(\log n)$ rounds and outputs messages with $O(\log^2 n)$ bits,
    with high probability.
\end{corollary}

\begin{proof}
    \cref{alg:separate-groups} modifies~\cref{alg:insert} 
    in the following way. Instead of computing the levels sequentially, we compute multiple groups in parallel. 
    Specifically, each nodes outputs in each round, a message of $2\log_{(1+\eps)}{n} - 1$ length where the $i$-th bit is $1$ if node $i$
    moves up a level in group $i$. Thus, since each group has $O(\log n)$ levels, we can compute the movements of all vertices in 
    all groups simultaneously in $O(\log n)$ rounds. 
    
    In~\cref{alg:ledp-densest}, each node releases $O(\log n)$ bits for each level in each group for their noisy degree.
    Thus, given $O(\log n)$ groups, the total node communication complexity is $O(\log^2 n)$ 
    per level $\lvl \in \{0, \dots, \numgrouplevels-1\}$.
    There are $O(\log n)$ rounds since~\cref{line:iterate} uses $O(\log n)$ rounds. (Everything inside the for loop can be
    done simultaneously.)
\end{proof}

}
\section{$\eps$-Edge DP Densest Subgraph}\label{sec:densest}

In this section, we present an $\eps$-edge DP densest subgraph algorithm that returns a subset of vertices $V' \subseteq 
V$ that induces
a $\left(1+\coren, O\left(\frac{\log^4 n}{\eps}\right)\right)$-approximate 
densest subgraph for any constant $\coren > 0$ in $\densetime$ time. Specifically, we prove the following theorem.

\begin{theorem}\label{thm:densest}
    There exists an $\eps$-edge DP densest subgraph algorithm that runs in $\densetime$ time
    and returns a subset of vertices $V' \subseteq V$
    that induces a $\left(1+\coren, O\left(\frac{\log^4 n}{\eps}\right)\right)$-approximate 
    densest subgraph for any constant $\coren > 0$.
\end{theorem}

We build upon the multiplicative weight update algorithm from~\cite{BGM14} along with modifications made
by~\cite{SuVu20} in their distributed algorithm. Although a number of private 
algorithms~\cite{blum2013learning,GLM10,gupta2012iterative,ganesh2020privately,hardt2012simple,hardt2010multiplicative,SU17,ullman2011pcps,vadhan2017complexity} 
exist for the multiplicative weight updates (MWU) method
of Arora \etal~\cite{AHK12}, such private algorithms focus on private multiplicative
weights for linear and non-linear
queries into databases and data release.
Such techniques in the database query setting do not immediately transfer to our densest subgraph setting;
namely, these algorithms often use the exponential mechanism to select queries, which is unnecessary
in our setting, since all nodes are queried during each update step and every
node provides a value to update the state of the algorithm.
Conversely, our $\eps$-edge DP
densest subgraph algorithm also does not have implications for private multiplicative weight update algorithms
for database queries.

In this section, to be consistent with the previous non-private works~\cite{BGM14,CQT22,SuVu20}, we give
our multiplicative approximation factors as $\left(1-a \cdot \coren'\right)$ for fixed constant $a$
and $\coren' \in (0, 1/a)$. Specifically, they define a $(1-\coren')$-approximate
densest subgraph to be one with density at least $(1-\coren') \cdot D^*$, where $D^*$ is the density of the densest subgraph. In~\cref{sec:prelims},
we define our multiplicative approximation factor as $(1+\coren)$ to be consistent with the other private
densest subgraph works~\cite{AHS21,NV21}. It is easy to convert between the two approximation factors
since a $(1-a \cdot \coren')$ guarantee for any constant $\coren' \in (0, 1/a)$ implies
a $(1+\coren)$ multiplicative guarantee for any $\coren > 0$ since $(1-a \cdot \coren') \leq \frac{1}{1+\coren}$
when $0 < \coren \leq \frac{a\coren'}{1-a\coren'}$, and we can choose a sufficiently small constant $\coren$. 
For simplicity, from now on
we fix a constant $\coren \in (0, \densestmultconst)$ to be the input parameter for our algorithms.

We present our algorithm in two parts. The main part, given in~\cref{alg:dp-densest},
calls the subroutine given in~\cref{alg:dp-densest-z} on various values of $z$.
\cref{alg:dp-densest} iterates through powers of $(1+\coren)^i$ for every $i \in [\floor{\log_{(1+\coren)}n}]$.
For each $(1+\coren)^i$, the algorithm passes the value as the input parameter $z$ into~\cref{alg:dp-densest-z}.
\cref{alg:dp-densest-z} goes through $O\left(\frac{\log n}{\coren^3}\right)$ phases where loads
are added to the edges in each phase. Then, the algorithm
returns a set of nodes whose induced subgraph has density at least
$(\kcoremultfactor)z - \frac{c\log^4 n}{\eps}$ for sufficiently large constant $c \geq 1$ \whp. 
\cref{alg:dp-densest} then returns the set of nodes returned for the largest power of $(1+\coren)^i$
or all of the nodes in the input graph if~\cref{alg:dp-densest-z} did not return a subset of the nodes
for any of the inputs for the parameter $z$.

The crux of our approximate densest subgraph algorithm lies in~\cref{alg:dp-densest-z}.
In order to ensure $\eps$-edge DP, our algorithm creates \defn{dummy}
edges that take a portion of the load. 
\conffull{}{
This technique also naturally leads to a general privacy framework for a certain class of \emph{locally-adjustable algorithms},
which we present in~\cref{sec:framework}.} 
Such dummy edges are responsible for both 
accumulating load and for determining whether the stopping 
conditions are satisfied. We describe our algorithm in more detail and prove its 
privacy, approximation, and runtime guarantees in the following sections.

\conffull{}{
\begin{figure}[!t]
\noindent
{\centering
\fbox{
\begin{minipage}[t][4.5cm]{.45\textwidth}
~
  (\textsc{Primal}) Maximum Density Subgraph
  ~
\begin{align*}
\textup{maximize} & \qquad \sum_{e \in E} x_{e} & \\
\textup{subject to} & \qquad \sum_{v \in V} y_v = 1 &\\
&\qquad x_{e} \leq y_u, x_{e} \leq y_v &\enspace \forall e=\{u, v\} \in E \\
&\qquad y_{v}, x_{e} \geq 0 &\enspace \forall e \in E, v \in V
\end{align*}
\end{minipage}}
\fbox{
\begin{minipage}[t][4.5cm]{.45\textwidth}
 
 (\textsc{Dual}) Lowest Out-Degree Orientation
~
\begin{align*}
\textup{minimize} &\qquad B \\
\textup{subject to} &\qquad \alpha_{eu} + \alpha_{ev} \geq 1 &\quad \forall e=\{u,v\} \in E \\
&\qquad \sum_{e \ni u} \alpha_{eu} \leq B &\quad \forall u \in V \text{ where } e \in E\\
&\qquad \alpha_{eu},\alpha_{ev}\geq 0 &\quad \forall e=\{u, v\} \in E
\end{align*}
\end{minipage}
}
}
\caption{Fig. 1 of~\cite{SuVu20}: Linear programs for densest subgraph (primal) and fractional lowest out-degree orientation (dual).}\label{fig:densest-lp}
\end{figure}
}

The densest subgraph algorithm performs multiplicative weight update 
on the dual of the densest 
\conffull{subgraph LP~\cite{Charikar00}.}{subgraph LP~\cite{Charikar00}
(shown in~\cref{fig:densest-lp}).}
Intuitively, this algorithm works by distributing a given load $z$ on a node (corresponding to the loads given by edges
to the nodes in the dual LP) to its adjacent edges.
Nodes with a large number of adjacent edges will be able to spread out its load among its many adjacent edges.
Edges with small cumulative load will be adjacent to two
high-degree nodes. Hence, they should be included in the densest subgraph. We can find such a subgraph by 
iterating from small to large load and keeping nodes that are adjacent to many edges with small loads.
Bahmani \etal~\cite{BGM14} were the first to apply the MWU framework to densest subgraphs and 
Su and Vu~\cite{SuVu20} give an explicit analysis for this algorithm. 
Although the non-private algorithms of~\cite{BGM14,SuVu20} adapt the MWU framework,
the analyses of~\cite{SuVu20} are self-contained. Thus, we present our private algorithm
in its entirety without the need to define the MWU framework.
We modify the analysis of~\cite{SuVu20} to show the approximation factor of our $\eps$-DP algorithm.

\conffull{}{
\subsection{Detailed Algorithm}

The pseudocode for our algorithm is given in~\cref{alg:dp-densest}. 
Recall from \cref{def:geom} that $\geom(b)$ denotes
the symmetric geometric distribution.
Specifically, $\geom(b)$  gives an integer output $X \sim 
\geom(b)$ with probability $\frac{e^b - 1}{e^b+1} \cdot e^{-|X| \cdot b}$.
In~\cref{alg:dp-densest}, we iterate through $i \in [\floor{\log_{(1+\coren)}(n)}]$ (\cref{line:i-iteration})
to obtain $z = (1+\coren)^i$ (\cref{line:compute-z}). After obtaining $z$, 
we pass the value into~\cref{alg:dp-densest-z} as the parameter $z$ (\cref{line:call-subroutine}). 
\cref{alg:dp-densest-z} returns a subset of nodes with density at least $z - \frac{c\log^4 n}{\eps}$ 
(for sufficiently large constant $c \geq 1$) \whp or all the nodes in the graph. 
We then check whether a set of nodes for the current parameter $z$ is returned (\cref{line:check-returned}); if so,
then the subgraph $S \subseteq V$ obtained for the largest $z$ parameter that is passed into the algorithm (\cref{line:initialize-max-v,line:return-max-v}) is returned. If no set of nodes is returned for any $z$, then
all of the nodes in the graph are returned.
We now explain the crux of our algorithm:~\cref{alg:dp-densest} calls~\cref{alg:dp-densest-z} in~\cref{line:call-subroutine}
with input parameter $z$ and~\cref{alg:dp-densest-z} returns
a subset of nodes with density bounded by $z$ minus
additive $\poly(\log n)$ error.
}

\SetKwFunction{FnDPDensestZ}{EdgeDPDensestSubgraphZ}
\SetKwFunction{FnDPDensest}{EdgeDPDensestSubgraph}
\SetKwFunction{FnCQT}{ApproximateDensestSubgraphCQT}
\begin{algorithm}[!t]
\conffull{
\small}{}
    \caption{\label{alg:dp-densest} $\eps$-Edge DP Densest Subgraph}
    \textbf{Input:} Graph $G = (V, E)$ with $n = |V|$ and $m = |E|$, a constant $\coren \in (0, \densestmultconst)$, and privacy parameter $\epsilon \in (0, 1)$.\\ 
    \textbf{Output:} A subset of nodes whose induced subgraph is a $\left(\kcoremultfactor, O\left(\frac{\log^4 n}{\eps}\right)\right)$-approximate densest subgraph.\\
    \Fn{\FnDPDensest{$G = (V, E), \coren, \eps$}}{
        $V_{\max} \leftarrow V$.\label{line:initialize-max-v}\\
        \For{$i \in [\floor{\log_{(1+\coren)} n}]$}{\label{line:i-iteration}
            Set $z = (1+\coren)^i$.\label{line:compute-z}\\
            $(S, x) \leftarrow$ \FnDPDensestZ{$G = (V, E), z, \coren, \epsilon$} (\cref{alg:dp-densest-z}).\label{line:call-subroutine}\\
            \If{$x \not= 0$}{\label{line:check-returned}
                $V_{\max} \leftarrow S$.\label{line:return-max-v}\\
            }
        }
        \Return $V_{\max}$.
    }
\end{algorithm}

\begin{algorithm}[!t]
\conffull{
\small}{}
    \caption{\label{alg:dp-densest-z} Edge DP Densest Subgraph using Density Parameter $z$}
    \textbf{Input:} Graph $G = (V, E)$ with $n = |V|$ and $m = |E|$, density $z \geq 0$, constant $\coren \in (0, \densestmultconst)$, privacy parameter $\epsilon \in (0, 1)$, 
    and sufficiently large constants $c_0, c_1, c_2, c > 0$.\\ 
    \textbf{Output:} A pair $(V', z')$ where $V'$ is a set of nodes $V' \subseteq V$ where $G[V']$ has density at least $z' - O\left(\frac{\log^4 n}{\eps}\right)$.\\
    \Fn{\FnDPDensestZ{$G = (V, E), z, \coren, \eps$}}{
        Let $T \leftarrow \densestphases$.\label{densest:phase}\\
        Initialize load $\ell(e) \leftarrow 0$ for all $e \in E$.\label{densest:initial-load}\\
        \If{$z = 0$}{
            \Return $(V, 0)$.
        }
        \For{phase $t = 1$ to $T$}{\label{densest:phases}
            \For{each node $v \in V$}{\label{densest:loads-loop}
                Let $[e_1, e_2, \dots, e_{\deg(v)}]$ be an ordered list of edges adjacent to $v$ sorted in non-decreasing order
                by $\ell(e_1) \leq \ell(e_2) \leq \cdots \leq \ell(e_{\deg(v)})$ (breaking ties by the ID of the other endpoint).\label{densest:sort-load}\\
                Sample $\dedge_v \sim \geom\left(\phasesnoise\right)$.\label{densest:dummy-edges}\\
                Initialize each $\halpha_{e_iv}^t \leftarrow 0$.\label{densest:initialize-new-load}\\
                Set $\halpha_{e_iv}^t\leftarrow 2$ for $i = 1, \dots, \left[\ceil{z/2} - 1 + \dedge_v\right]_0^{\deg(v)}$. \label{densest:assign-new-load}\\
            }
        \For{each integer $\ell \in [0, 4T]$}{\label{densest:search}
            Set $V_{\ell}' \leftarrow \emptyset$.\\
            \For{each node $v \in V$}{\label{densest:iterate-node}
                Sample $Z \sim \geom\left(\cutoffnoise\right)$.\label{densest:induced-degree-noise}\\
                \If{number of incident edges $e$ to $v$ with $\ell(e) \leq \ell$ is at least $\ceil{z/2} + Z - \frac{c_1 \log^4 n}{\eps}$}{\label{densest:check-induced}
                    $V_{\ell}' \leftarrow V_{\ell}' \cup \{v\}$.\label{densest:induced-subgraph}\\
                }
            }
            
            Sample $Y \sim \geom\left(\densecutoffnoise\right)$.\\
            \If{$G[V_{\ell}']$ is a non-empty graph with density $\geq z + Y - \cutoffsubfactor$}{\label{densest:density}
                \Return $(V_{\ell}', z)$\label{densest:subgraph}\\
            }
        }
        \For{each edge $e = (u, v) \in E$}{
            Set $\ell(e) \leftarrow \ell(e) + \halpha_{eu}^t + \halpha_{ev}^t$.\label{densest:update-load}
        }
        }

        \Return $(V, 0)$.\label{densest:noisy-density}
	}
\end{algorithm}

\conffull{}{
\cref{alg:dp-densest-z} operates over $T \leftarrow \densestphases$ phases %
for some sufficiently large constant $c_0 > 0$ (\cref{densest:phase}). 
Each phase updates the \emph{load} of each edge $e$, where the initial loads $\ell(e)$ of all edges are set to $0$ (\cref{densest:initial-load}).
Then, for each of the $T$ phases, in parallel, each node $v$ sorts its adjacent edges in non-decreasing order
according to their load (\cref{densest:sort-load}). The loads of the edges essentially correspond to variable values in the well-known 
dual of the densest subgraph LP (shown in~\cref{fig:densest-lp}).
Then, we sample from the symmetric geometric distribution
to create the dummy edges (\cref{densest:dummy-edges}). Let $X_v$ be the number of these dummy edges created for a node $v$.
There are two ways that these dummy edges affect the loads distributed on the real edges. First, if $X_v$ is negative, then we
create dummy edges that we assume have the \emph{smallest} loads among all adjacent edges
to $v$. In this case, these dummy edges then take up part of 
the load given by $v$ and hence, the real edges adjacent to $v$
altogether receive less total load than $z$. %
However, if $X_v$ is positive, then we imagine we create more open slots for additional
real edges to take on load from $v$. In this case, the sum of the loads received by all real edges adjacent to 
$v$ may be greater than $z$.
This intuition translates to the procedure for assigning loads to edges in~\cref{densest:assign-new-load}: each of the 
$\left[\ceil{z/2} - 1 + \dedge_v\right]_0^{\deg(v)}$ 
edges with the smallest loads 
is assigned a load of $2$.
After assigning new loads to the edges,
the algorithm attempts to find an approximate densest subgraph with sufficient density.

We iterate through each possible integer load in $[0, 4T]$
(\cref{densest:search}), which are the minimum and maximum possible load on any edge due to the fact that 
each endpoint of an edge adds at most a load of $2$ during each of the $T$ phases. 
For each possible integer load in $[0, 4T]$, we
check whether there exists a densest subgraph of sufficient density, specifically, density at least 
$z + Y - \cutoffsubfactor$ (\cref{densest:density}).
If so, then the algorithm returns a set of nodes $V'_\ell$, where the density of the induced subgraph consisting of the nodes in $V'_\ell$ approximates
the density of the densest subgraph in $G$ (\cref{densest:subgraph}).
To check for whether a sufficiently dense subgraph exists, first, each node $v \in V$ samples a noise $Z$ from the symmetric geometric 
distribution (\cref{densest:induced-degree-noise}). Then, node $v$ checks the number of incident edges $e$ to $v$ that
have load $\ell(e) \leq \ell$ that is at most the load of the current iteration, $\ell$. If the number of such 
incident edges is at least $\ceil{z/2} + Z -\frac{c_1 \log^4 n}{\eps}$ (\cref{densest:check-induced}), we add the node to the set $V'_{\ell}$. We then check the density of the subgraph induced by all nodes that pass the check (\cref{densest:induced-subgraph,densest:density}). 
If the density of this
subgraph is sufficiently large, we return the nodes in $V_{\ell}'$ as the approximate densest subgraph (\cref{densest:subgraph}).
Otherwise, if the density is not sufficiently large, we update the loads of each of the edges and 
proceed with the next iteration of the algorithm (\cref{densest:update-load}). We show in~\cref{sec:approx} that we
return a subgraph of sufficiently large density \whp if $z$ is sufficiently small (but not too small)
compared to the optimum density $D^*$.
}

\conffull{}{
We first prove that~\cref{alg:dp-densest} is $\eps$-edge DP in~\cref{sec:densest-privacy}. As the main intermediate step,
we prove that~\cref{alg:dp-densest-z}
is $\left(\frac{\eps}{\log_{(1+\coren)}n}\right)$-edge DP in~\cref{lem:densest-z-dp}. Then, we prove that 
we return a set of nodes where the induced subgraph on this 
set gives our desired approximation factor in~\cref{sec:approx}. Finally, we analyze the 
runtime of our algorithm in~\cref{sec:densest-runtime}.

\subsection{$\eps$-Edge DP Privacy Guarantees}\label{sec:densest-privacy}

We prove~\cref{alg:dp-densest} is $\eps$-edge DP in this section. We first 
prove our main lemma that~\cref{alg:dp-densest-z} is
$\frac{\eps}{\log_{(1+\coren)} n}$-edge DP in \cref{lem:densest-z-dp}.
Then, we prove \cref{thm:densest-dp} via composition (\cref{thm:composition})
over $\log_{(1+\coren)} n$ calls to the algorithm. 
The key challenge in our analysis is that although we only use symmetric geometric noise in our algorithm, 
we sometimes cannot use the proof of the privacy of the geometric mechanism directly.
and additional analysis is necessary. Nevertheless, we show that our analysis is quite intuitive and 
can also be applied to our framework in~\cref{sec:framework}.
In the following proofs, we denote edge-neighboring graphs by $G = ([n], E)$ and $G' = ([n], E \cup \{e'\})$. We 
also use the terminology \textbf{$\eps$-indistinguishable distributions} to denote two distributions 
where the probabilities of any event 
differs by at most a factor of $e^{\eps}$ between the distributions. It follows by definition that two symmetric geometric
distributions with the same parameter and whose means differ by at most $1$ are $\eps$-indistinguishable.
Furthermore, suppose that we have the pairs of distributions $(A_1, B_1)$ and $(A_2, B_2)$ where each pair is 
$\eps$-indistinguishable. If all four distributions are independent, then, the joint distributions for $(A_1, B_1)$
and $(A_2, B_2)$ are $2\eps$-indistinguishable.

\begin{lemma}\label{lem:densest-z-dp}
    \cref{alg:dp-densest-z} is $\frac{\eps}{\log_{(1+\coren)} n}$-edge DP. %
\end{lemma}

\begin{proof}
    \cref{alg:dp-densest-z} iterates through $T$ phases (\cref{densest:phases}) where in each phase,
    it first assigns loads to edges 
    (\cref{densest:sort-load,densest:dummy-edges,densest:initialize-new-load,densest:assign-new-load}). The
    assigned loads are added to a cumulative load at the end of the procedure (\cref{densest:update-load}).
    After assigning loads to edges, 
    it iterates through $4T$ integer loads (\cref{densest:search}). 
    For each load, the algorithm iterates through 
    each node (\cref{densest:iterate-node}) and determines whether to add the node to $V_{\ell}'$ 
    (\cref{densest:induced-degree-noise,densest:check-induced,densest:induced-subgraph}). 
    After constructing $V_{\ell}'$, it determines whether to return $V_{\ell}'$ if the subgraph induced on
    $V_{\ell}'$ has sufficiently large density (\cref{densest:density,densest:subgraph}). 
    Finally, if no node is returned at the end of all phases and all iterations, all nodes in the graph
    are returned (\cref{densest:noisy-density}).
    
    We introduce notation for each of the following random variables:
    the loads assigned to each edge during each phase, indicator variables for whether each node is added to $V_{\ell}'$
    during each of the $4T$ iterations of each phase, and indicator variables for whether $V_{\ell}'$ is returned.
    We prove the privacy of our algorithm by conditioning on the states of random variables from previous iterations and phases
    and showing that with approximately equal probabilities the random variables in the current iteration take the same 
    values in edge-neighboring graphs. 
    
    We now define each of the random variables that we use in our proof. Given a graph $G = ([n], E)$, 
    we first define the random variables
    representing the loads assigned to each edge during each phase. 
    In other words, 
    Let $L^t$ be a set of triples $L^t = \{(e, \halpha^t_{ei}, \halpha^t_{ej}) \mid e \in E\}$ where each triple
    $(e, \halpha^t_{ei}, \halpha^t_{ej})$ contains an edge $e = (i, j) \in E$ 
    and the loads $\halpha^t_{ei}, \halpha^t_{ej}$ assigned to edge $e = (i, j)$ in phase $t$. 
    Then, let $L = \left(L^1, \dots, L^T\right)$ 
    be a sequence of load assignments to edges for each of the $T$ phases. The interesting case is when $z > 0$ as otherwise
    all nodes of the graph are returned when $z = 0$.
    Let $F(L^t, E^*)$ be a function that, given $L^t$ and a set of edges $E^*$, returns a set of loads
    $J^t$ where $J^t = \{(e, \halpha^t_{ei}, \halpha^t_{ej}) \in L^t \mid e \in E^*\}$
    is a set of tuples $(e, \halpha^t_{ei}, \halpha^t_{ej})$ from $L^t$ for each edge $e \in E^*$.
    
    Then, let $B^{\ell}_{v, t}(G)$ be the indicator random variable for node $v$, phase $t$, and load cutoff $\ell$ that
    is $1$ if node $i$ has at least $\ceil{z/2} + Z^{\ell}_{v, t} -\frac{c_1 \log^4 n}{\eps}$ 
    neighbors in $G$ with load at most $\ell$ (where
    $Z^{\ell}_{v, t} \sim \geom\left(\cutoffnoise\right)$) and $0$ otherwise. Finally, 
    let $C^{\ell}_{t}(G)$ be the indicator random variable for 
    phase $t$ and load cutoff $\ell$ in $G$ that is $1$ if the density of $V_{\ell}'$ satisfies the cutoff in
    \cref{densest:density} and 
    $0$ otherwise. If $C^{\ell}_{t}(G) = 1$ for some phase $t$ and load $\ell$, then we set the values 
    of \emph{all random variables} to $0$ for all iterations $\ell' > \ell$ in phase $t$ and for all phases 
    $t' > t$.

    We introduce several claims (\cref{claim:noisy-loads,claim:cutoff-load,claim:density}) 
    and their proofs before proving this lemma.    
    We first show the following claim about the loads assigned to edges in phase $t$ conditioned
    on the values of all random variables for phases less than $t$. 
    Let $\mathcal{J}$ be the set of all possible sets of triples of load assignments to edges
    in $E$ and let $S^t \in \mathcal{J}$.
    Let $(b_1, \dots, b_n)_{\ell, t} \in \{0, 1\}^n$ and $d^{\ell}_t \in \{0, 1\}$ for
    all $t \in [T]$ and $\ell \in \{0, \dots, 4T\}$.
    To set up~\cref{claim:noisy-loads}, %
    let $R^t$ be the sequence of events where for all phases less than a fixed $t \in [T]$,
    \begin{align*}
        &R^t(G) = (C^{4T}_{t-1}(G) = d^{4T}_{t-1}, \dots, C^0_{1}(G) = d^0_1,\\
        &\left(B^{4T}_{1, t-1}(G), \dots, B^{4T}_{n, t-1}(G)\right) = \left(b_1, \dots, b_n\right)_{4T, t-1}, \dots,
        \left(B^{0}_{1, 1}(G), \dots, B^{0}_{n, 1}(G)\right) = \left(b_1, \dots, b_n\right)_{0, 1},\\
        &J^{t-1}(G) = S^{t-1}, \dots, J^1(G) = S^1).
    \end{align*}
    Define $\hat{R}^{t}$ the same way as $R^t$, substituting $G'$ for $G$.
    
    Note that since we are in the discrete setting, computing the probability expression
    in~\cref{claim:noisy-loads} is sufficient for privacy as opposed to the common $\frac{1}{e^{\eps/(4T)}} \leq
    \frac{\prob\left[J^t(G') \in S \mid \hat{R}^t(G')\right]}{\prob\left[J^t(G) \in S \mid R^t(G)\right]} 
    \leq e^{\eps/(4T)}$ for any $S \subseteq \mathcal{J}$, which is necessary for continuous settings. 
    
    \begin{claim}\label{claim:noisy-loads}
    For all edge-neighboring graphs $G = ([n], E)$ and $G' = ([n], E \cup \{e'\})$ and all $S^t \in \mathcal{J}$,
        $$\frac{1}{e^{\eps/(3T\log_{(1+\coren)}n)}} \leq \frac{\prob\left[J^t(G') = S^t \mid R^t(G')\right]}{\prob\left[J^t(G) = S^t \mid R^t(G)\right]} \leq e^{\eps/(3T\log_{(1+\coren)}n)}.$$ 
    \end{claim}
    
    \begin{proof}
    First, note that only \cref{densest:loads-loop,densest:sort-load,densest:dummy-edges,densest:assign-new-load,densest:initialize-new-load,densest:update-load}
    assign loads to edges in~\cref{alg:dp-densest-z}. 
    Let $J^t(G)$ be the random variable representing the set of loads returned by $F(L^t, E)$ for phase $t$. 
    Let $S = (S^1, \dots, S^{t-1})$ be a sequence of sets of load values on edges in $E$. 
    We prove, for all $t \in [T]$, that
    \begin{align}
        \frac{1}{e^{\eps/(3T\log_{(1+\coren)}n)}} \leq \frac{\prob\left[J^t(G') = S^t \mid R^t(G')\right]}{\prob\left[J^t(G) = S^t \mid R^t(G)\right]} \leq e^{\eps/(3T\log_{(1+\coren)}n)}.
    \end{align}
    First, conditioned on the previous sets of loads being the same for edges
    in $J^t(G)$ and $J^t(G')$ 
    for phases in $[t-1]$,
    the load at the beginning of phase $t$ on $e \in E$ is given by $\ell_t(e) = \ell'_t(e)$ (where $\ell_t(e)$ is the load
    on $e$ in phase $t$ in $G$ and $\ell'_t(e)$ is the load on $e$ in phase $t$ in $G'$)
    since each $\ell_t(e) = \sum_{r = 1}^{t-1} \left(\halpha^r_{ei} + \halpha^r_{ej}\right)$ for all $e \in E$. 
    Then, the distribution on the loads assigned to
    $\halpha^t_{ei}$%
    and $\halpha^t_{ej}$ are the same in $G$ and $G'$
    for edges $e = (i, j) \in E$ where \emph{neither endpoint} is $u$ or $v$. 
    It remains to consider loads assigned to edges in $E$ where exactly one of the endpoints is $u$ or $v$.
    
    Let $(e_1, \dots, e_{|N(v)|})$ denote the order of the edges determined by loads in $\{L^1(G), \dots, L^{t-1}(G)\}$ and
    let $(\hat{e}_1 \dots, \hat{e}_{|N(v)| + 1})$ denote the order of the edges $N(v) \cup \{e'\}$ 
    determined by loads in $\{L^1(G'), \dots, L^{t-1}(G')\}$ (\cref{densest:sort-load}). 
    Let edge $e' = (u, v)$ be the $k$-th edge in $(\hat{e}_1 \dots, \hat{e}_{|N(v)| + 1})$.
    Since we conditioned on $J^{t-1}(G') = J^{t-1}(G) = S^{t-1}, \dots, J^{1}(G') = J^1(G) = S^1$, 
    it follows that $e_p = \hat{e}_p$ for 
    all $p \in [k - 1]$ and $e_{p} = \hat{e}_{p + 1}$ for all $p \in \{k, \dots, |N(v)|\}$.
    Since $G'$ has an additional edge, $\deg(v) = |N(v)|$ and $\deg'(v) = |N(v)| + 1$
    for all nodes $v$. For the purposes of this proof, 
    we pad each of the two sequences $OrderSeq(G) = (e_1 \dots, e_{|N(v)|})$ and
    $OrderSeq(G') = (\hat{e}_1 \dots, \hat{e}_{|N(v)| + 1})$ with extra \emph{filler} edges on the right ends to produce
    $(e_1, \dots, e_{|N(v)|}, f_1, f_2, f_3, \dots, f_{\ceil{z/2} - 1 - |N(v)| + 1})$ and 
    $(\hat{e}_1 \dots, \hat{e}_{|N(v)| + 1}, f_1, f_2, 
    f_3, \dots, f_{\ceil{z/2}-1 - |N(v)| - 1})$. These filler edges are unique and given in the same order in both 
    sequences. This means that the position $\ceil{z/2} - 1$ falls on an edge in both sequences. We now 
    define the \emph{ideal cutoff} for the sequence $OrderSeq(G)$ to be $\ceil{z/2} - 1$ and for the sequence $OrderSeq(G')$ to be the position of the
    $\left(\ceil{z/2} - 1\right)$-st edge in $OrderSeq(G)$.
    To state in another 
    way, if we let $\hat{e}$ be the $(\ceil{z/2} - 1)$-st edge in $OrderSeq(G')$, we find the position of $\hat{e}$ in $OrderSeq(G)$
    (if $\hat{e} = e'$, then we take the edge to its right). The position of this edge is either at position 
    $\ceil{z/2} - 1$ in $OrderSeq(G)$ or at position $\ceil{z/2} - 2$ in $OrderSeq(G)$. Trivially, the ideal cutoff 
    for the $(\ceil{z/2} - 1)$-st edge in $OrderSeq(G)$ is $\ceil{z/2} - 1$. We now define the following variables
    with respect to these ideal cutoffs
    which we denote by $IDEAL_v(G')$ and $IDEAL_v(G)$. Conditioned on $R^t(G)$ and $R^t(G')$, these
    ideal cutoffs are fixed.
    
    Let $CUT^t_v(G) = IDEAL_v(G) + X_v$ be the symmetric geometric variable 
    with mean $IDEAL_v(G)$ for $G$ in phase $t$ and $CUT^t_v(G') = IDEAL_v(G') + \hat{X}_v$
    be the variable for $G'$. Then, the conditional distributions (conditioned on $R^t(G)$ and $R^t(G')$, respectively)
    for $CUT^t_v(G)$ and $CUT^t_v(G')$ are symmetric geometric distributions with parameter given in~\cref{densest:dummy-edges}
    and with means $IDEAL_v(G)$ and $IDEAL_v(G')$, respectively.
    These random variables are drawn from symmetric geometric distributions where the mean differs by at most $1$. Hence, 
    these distributions are $\frac{\eps}{6T\log_{(1+\coren)}n}$-indistinguishable by the parameter given 
    in~\cref{densest:dummy-edges}. 
    The same argument above holds for $u$.
    
    Now, we consider the joint distributions for the variables $CUT^t_v(G)$ and 
    $CUT^t_u(G)$ (and respectively, $CUT^t_v(G')$ and $CUT^t_u(G')$). 
    Conditioned on $R^t(G)$ and $R^t(G')$, the orderings are fixed for $u$ and $v$ which means that $IDEAL_u(G)$ and $IDEAL_v(G)$ are also fixed; 
    thus, the variables 
    $CUT^t_u(G)$ and $CUT^t_v(G)$ are independently drawn, the events for each endpoint are independent
    and so the joint conditional probability distributions are $\frac{\eps}{3T\log_{(1+\coren)}n}$-indistinguishable, 
    proving our claim.
    \end{proof}
    
    Intuitively, \cref{claim:noisy-loads} shows that for the set of edges $E$ in both $G$ and $G'$, our algorithm assigns
    the same loads to edges in $E$ in $G$ and $G'$ with roughly equal probability, conditioned on the values of the random
    variables from previous iterations and phases.
    
    We show the following claim regarding the values of $B^{\ell}_{v, t}(G)$.
    Let $W$ be the sequences of events for iterations less than $\ell$ in phase $t$ and all phases less than $t$
    where 
    \begin{align*}
        &W = (C^{\ell - 1}_{t}(G) = d^{\ell-1}_{t}, \dots, C^0_{1}(G) = d^0_1,\\
        &\left(B^{\ell-1}_{1, t}(G), \dots, B^{\ell-1}_{n, t}(G)\right) = \left(b_1, \dots, b_n\right)_{\ell-1, t}, \dots,
        \left(B^{0}_{1, 1}(G), \dots, B^{0}_{n, 1}(G)\right) = \left(b_1, \dots, b_n\right)_{0, 1},\\
        &J^{t}(G) = S^{t}, \dots, J^1(G) = S^1).
    \end{align*}
    Define $W'$ the same way as $W$, substituting $G'$ instead of $G.$
    
    \begin{claim}\label{claim:cutoff-load}
        Given two edge-neighboring graphs $G$ and $G'$, let $(B^{\ell}_{1, t}(G), \dots, B^{\ell}_{n, t}(G))$ be the 
        sequence of random variables and $(b_1, \dots, b_n)_{\ell, t} \in \{0, 1\}^n$ for all nodes $i \in [n]$. 
        Then, for any phase $t \in [T]$ and load $\ell \in \{0, \dots, 4T\}$, 
        \begin{align*}
            \frac{1}{e^{\eps/(3T(4T+1)\log_{(1+\coren)}n)}} \leq \frac{\prob[(B^{\ell}_{1, t}(G'), \dots, B^{\ell}_{n, t}(G')) = (b_1, \dots, b_n)_{\ell, t} \mid W']}{\prob[(B^{\ell}_{1, t}(G), \dots, B^{\ell}_{n, t}(G)) = (b_1, \dots, b_n)_{\ell, t}\mid W]} \leq e^{\eps/(3T(4T+1)\log_{(1+\coren)}n)}.
        \end{align*}
    \end{claim}
    
    \begin{proof}
        Let $Z^{\ell}_{v, t}$ and $\hat{Z}^\ell_{v, t}$ be the noise drawn in phase $t$ and iteration $\ell$ for $i$ 
        in~\cref{densest:induced-degree-noise} of~\cref{alg:dp-densest-z} for $G$ and $G'$, respectively.
        Conditioned on $J^t(G) = S^t$ and $J^t(G') = S^t$, edge $e' = (u, v) \in E' \setminus E$ has some arbitrary load
        from $L^t(G')$ and all other edges have the same load in phase $t$ in $G$ and $G'$. For all nodes $i \neq u, v \in [n]$,
        the conditional distributions of the random variables $B^{\ell}_{i, t}(G)$ and $B^{\ell}_{i, t}(G')$ are the same given
        $W'$ and $W$. Now we consider the conditional distributions of $B^{\ell}_{v, t}(G)$ and $B^{\ell}_{v, t}(G')$.
        Let $load^{\ell}_{v, t}(G)$ and $load^{\ell}_{v, t}(G')$ be the number of adjacent edges to $v$ in $G$ and $G'$, respectively,
        with load at most $\ell$ in phase $t$; conditioned on $W$ and $W'$, respectively, $load^{\ell}_{v, t}(G)$ and $load^{\ell}_{v, t}(G')$ are fixed and $load^{\ell}_{v, t}(G) \leq load^{\ell}_{v, t}(G') \leq load^{\ell}_{v, t}(G) + 1$ since
        $v$ is adjacent to one additional edge in $G'$ which can have load at most $\ell$. Then, $B^{\ell}_{v, t}(G) = 1$
        when $LOAD^{\ell}_{v, t}(G) = Z^{\ell}_{v, t} - load^{\ell}_{v, t}(G) 
        + \ceil{z/2} - \frac{c_1 \log^4 n}{\eps} \leq 0$ (and symmetrically for $B^{\ell}_{v, t}(G) = 0$). We define
        the random variable $LOAD^{\ell}_{v, t}(G')$ the same way except substituting $\hat{Z}^{\ell}_{v, t}$
        and $load^{\ell}_{v, t}(G')$. The random variables $LOAD^{\ell}_{v, t}(G)$ and $LOAD^{\ell}_{v, t}(G')$ are
        symmetric geometric variables drawn from distributions with the same parameter and where the mean differs by 
        at most $1$. Thus, the conditional distributions for these two random variables are $\left(\cutoffnoise\right)$-indistinguishable 
        (using the parameter from~\cref{densest:induced-degree-noise}).
        Finally, the conditional distributions for $LOAD^{\ell}_{v, t}(G)$ and $LOAD^{\ell}_{u, t}(G)$
        are independent since $load^{\ell}_{v, t}(G) 
        + \ceil{z/2} - \frac{c_1 \log^4 n}{\eps}$ are fixed and $Z^{\ell}_{v, t}$ and $\hat{Z}^{\ell}_{v, t}$ are drawn
        independently. (The same holds for the distributions of the corresponding random variables for $G'$.) Hence,
        the conditional joint distributions for the pairs in $G$ and the pairs in $G'$ are $\left(\frac{\eps}{3(4T+1)T\log_{(1+\coren)}n}\right)$-indistinguishable.
    \end{proof}
    
    Finally, we return the set of nodes in $V_{\ell}'$ in~\cref{densest:density,densest:noisy-density}
    if the density of the induced subgraph $G[V_{\ell}']$ is at least 
    $z + Y - \cutoffsubfactor$ and
    if no induced density satisfies the cutoff, respectively.
    Similarly to the previous claims, we can prove the following claim for $C^{\ell}_{t}(G)$
    and $C^{\ell}_{t}(G')$ conditioned on the sequence of indicator random variables $(B^{\ell}_{1, t}(G), \dots, 
    B^{\ell}_{n, t}(G))$ and $(B^{\ell}_{1, t}(G'), \dots, B^{\ell}_{n, t}(G'))$. 
    
    Let $Q$  be the 
    sequence of events where 
    \begin{align*}
        &Q = (C^{\ell - 1}_{t}(G) = d^{\ell-1}_{t}, \dots, C^0_{1}(G) = d^0_1,\\
        &\left(B^{\ell}_{1, t}(G), \dots, B^{\ell}_{n, t}(G)\right) = \left(b_1, \dots, b_n\right)_{\ell, t}, \dots,
        \left(B^{0}_{1, 1}(G), \dots, B^{0}_{n, 1}(G)\right) = \left(b_1, \dots, b_n\right)_{0, 1},\\
        &J^{t}(G) = S^{t}, \dots, J^1(G) = S^1).
    \end{align*}
    We define $Q'$ the same way as $Q$, substituting $G'$ instead of $G.$
    
    \begin{claim}\label{claim:density}
        Given two edge-neighboring graphs $G$ and $G'$, let $Q$ and $Q'$ be the sequences of 
        events as defined above. Let $C^{\ell}_{t}(G)$  and $C^{\ell}_t(G')$ be the 
        indicator random variables for $G$ and $G'$, respectively, as defined above, 
        and let $d^\ell_t \in \{0, 1\}$. 
        Then, for any phase $t \in [T]$ and integer load $\ell \in \{0, \dots, 4T\}$, 
        \begin{align*}
            \frac{1}{e^{\eps/(3T(4T+1)\log_{(1+\coren)}n)}} \leq \frac{\prob[C^{\ell}_t(G') = d^\ell_t \mid Q']}{\prob[C^{\ell}_t(G) = d^\ell_t \mid Q]} \leq e^{\eps/(3T(4T+1)\log_{(1+\coren)}n)}.
        \end{align*}
    \end{claim}
    
    \begin{proof}
        This proof follows the proof of~\cref{claim:cutoff-load} almost identically except for the use of
        different random variables.
        Given the conditioning on events $Q$ and $Q'$, the set of nodes in $V_{\ell}'$ are fixed
        and the same in phase $t$
        in both $G$ and $G'$. Then, the sensitivity of a function that returns the 
        density $\dense(G[V'])$ of any subset $V' \subseteq V$ 
        of nodes is $1$. It follows that $\dense(G[V_{\ell}']) \leq \dense(G'[V_{\ell}']) \leq
        \dense(G[V_{\ell}']) + 1$. This means that we can again argue that the random variables 
        $DENSE^{\ell}_t(G) = Y^\ell_t + z - \dense(G[V_{\ell}']) - \cutoffsubfactor$ and 
        $DENSE^{\ell}_t(G')$ (defined the same way except substituting $\hat{Y}^\ell_t$ and $\dense(G'[V_{\ell}'])$)
        are symmetric geometric variables drawn from distributions with the same parameter and where the mean of the 
        distributions differ by at most $1$. Thus, these two distributions are $\left(\densecutoffnoise\right)$-indistinguishable
        and the claim follows.
    \end{proof}
    
    We now apply the chain rule 
    on~\cref{claim:noisy-loads,claim:cutoff-load,claim:density} to prove our 
    lemma. During each phase $t \in [T]$,
    we first compute the new loads assigned
    to $\halpha_{eu}$
    for all $e \in E$ and $u \in e$. Then, we iterate
    through all loads $\ell \in \{0, \dots, 4T\}$
    to obtain the random variables $B^{\ell}_{v, t}(G)$
    and $C^{\ell}_t$. We apply the
    chain rule as follows:
    \begin{align*}
         A &= \frac{\prob[C^{4T}_T(G') = d^{4T}_T \mid Q']}{\prob[C^{4T}_T(G) = d^{4T}_T \mid Q]}
         \times \frac{\prob[(B^{4T}_{1, T}(G'), \dots, B^{4T}_{n, T}(G')) = (b_1, \dots, b_n)_{4T,T} \mid W']}{\prob[(B^{4T}_{1, T}(G), \dots, B^{4T}_{n, T}(G)) = (b_1, \dots, b_n)_{4T,T}\mid W]}\\ 
         &\times \frac{\prob\left[J^{T}(G') = S^{T} \mid R'\right]}{\prob\left[J^{T}(G) = S^{T} \mid R\right]} 
         \times \cdots \times
         \frac{\prob[C^{0}_1(G') = d^{0}_1 \mid (B^{0}_{1, 1}(G'), \dots, B^{0}_{n, 1}(G')) = (b_1, \dots, b_n)_{0, 1}, J^{1}(G') = S^{1}]}{\prob[C^{0}_1(G) = d^{0}_1 \mid (B^{0}_{1, 1}(G), \dots, B^{0}_{n, 1}(G)) = (b_1, \dots, b_n)_{0, 1}, J^{1}(G) = S^{1}]}\\
         &\times \frac{\prob[(B^{0}_{1, 1}(G'), \dots, B^{0}_{n, 1}(G')) = (b_1, \dots, b_n)_{0, 1} \mid J^{1}(G') = S^{1}]}{\prob[(B^{0}_{1, 1}(G), \dots, B^{0}_{n, 1}(G)) = (b_1, \dots, b_n)_{0, 1}\mid J^{1}(G) = S^{1}]}
         \times \frac{\prob\left[J^{1}(G') = S^{1}\right]}{\prob\left[J^{1}(G) = S^{1}\right]}\\
         &= \frac{\prob[C^{4T}_T(G') = d^{4T}_T \cap (B^{4T}_{1, T}(G'), \dots, B^{4T}_{n, T}(G')) = (b_1, \dots, b_n)_{4T,T} \cap
         C^{0}_1(G') = d^{0}_1 \cap J^{T}(G') = S^{T} \cap \cdots]}{\prob[C^{4T}_T(G) = d^{4T}_T \cap (B^{4T}_{1, T}(G), \dots, B^{4T}_{n, T}(G)) = (b_1, \dots, b_n)_{4T,T} \cap
         C^{0}_1(G) = d^{0}_1 \cap J^{T}(G) = S^{T} \cap \cdots]}
    \end{align*}
    
    By~\cref{claim:cutoff-load,claim:density,claim:noisy-loads}, we can upper and lower bound $A$ by 
    $\left(\frac{1}{e^{\eps/(3T\log_{(1+\coren)}n)}}\right)^T \cdot \left(\frac{1}{e^{\eps/(3T(4T+1)\log_{(1+\coren)}n)}}\right)^{T(4T+1)} \cdot \left(\frac{1}{e^{\eps/(3T(4T+1)\log_{(1+\coren)}n)}}\right)^{T(4T+1)} \leq 
    \frac{1}{e^{\eps/\log_{(1+\coren)}n}} \leq A \leq e^{\eps/\log_{(1+\coren)}n} \leq \left(e^{\eps/(3T\log_{(1+\coren)}n)}\right)^T \cdot \left(e^{\eps/(3T(4T+1)\log_{(1+\coren)}n)}\right)^{T(4T+1)}
    \cdot \left(e^{\eps/(3T(4T+1)\log_{(1+\coren)}n)}\right)^{(4T+1)T}$, proving \cref{lem:densest-z-dp}.
\end{proof}

\begin{lemma}\label{thm:densest-dp}
    \cref{alg:dp-densest} is $\eps$-edge DP.
\end{lemma}

\begin{proof}
    \cref{alg:dp-densest} calls~\cref{alg:dp-densest-z} a total of $\log_{(1+\coren)}n$ times.
    By~\cref{lem:densest-z-dp} and by the composition theorem (\cref{thm:composition},
    \cref{alg:dp-densest-z} is $\frac{\eps}{\log_{(1+\coren)} n}$-edge DP and by applying composition $\log_{(1+\coren)}n$
    times, \cref{alg:dp-densest-z} is $\eps$-edge DP.
\end{proof}

\subsection{Approximation Analysis}\label{sec:approx}
The proof of our approximation guarantee relies on the LP formulation of the densest subgraph problem and 
its dual introduced by Charikar~\cite{Charikar00}. We modify
the proofs of Su and Vu~\cite{SuVu20} to prove the approximation factor
of our algorithm in the $\eps$-edge DP setting; the application of the 
multiplicative weights framework to densest subgraph was 
first given in~\cite{BGM14} and the analysis was made explicit for the densest 
subgraph problem by~\cite{SuVu20}. Due to the added noise, 
all of the proofs for the original algorithm must be modified to account
for the additional error. The primal LP solves the densest subgraph 
problem while its dual is traditionally used for the (fractional) lowest out-degree orientation problem.
The primal (left) and its dual (right) are shown in~\cref{fig:densest-lp}. We use the 
formulation of the LP and its dual that is given in \cite{SuVu20} and is also present in~\cite{BGM14}. 
In our analysis, we let $\optdensity$ be the density of the densest subgraph in the input graph. 

Charikar~\cite{Charikar00} showed, via a rounding algorithm, that the optimal value of the LP is exactly the density of the densest subgraph in the graph.
We use the notation from~\cite{SuVu20} when describing solutions to the primal and dual. Specifically, we say that a solution to
the primal satisfies $\primal(A)$ if it is a feasible solution whose objective value is at least $A$. 
A dual solution satisfies $\dual(B)$ if it is a feasible solution with objective value at most $B$. 
Let $D^*$ be the density of the densest subgraph in the input graph.
We reiterate that the primal is feasible if $A \leq \optdensity$
and the dual is feasible if $B \geq \optdensity$.
There are two main aspects of the analysis. As in Su and Vu~\cite{SuVu20}, for each value of $z$ and appropriately large constant $c \geq 1$,
we show that at least one of the following conditions is satisfied:

\begin{itemize}
    \item \cref{alg:dp-densest-z} returns a subgraph which satisfies $\primal\left((\kcoremultfactor)z - \cutoffsubfactor\right)$ with high probability.
    \item The loads on the edges given by~\cref{alg:dp-densest-z} form a solution to $\dual\left(\dualfactor\right)$ with high probability.
\end{itemize}

We use the above to show that if $z < \frac{D^* - \densestsubfactor}{1+10\coren}$, 
then our algorithm returns a subgraph with density at least 
$(\kcoremultfactor)z - \cutoffsubfactor$ \whp for sufficiently large constant $c \geq 1$ (that is also dependent on the 
constant parameter $\coren \in (0, \densestmultconst)$). 
To show this, we show that if $z < \frac{D^* - \densestsubfactor}{1+10\coren}$, then no feasible solution to 
$\dual(\dualfactor)$ exists, \whp. 
This means that our algorithm will obtain a $(\kcoremultfactor)z - \cutoffsubfactor$ approximate densest subgraph since in order for at least one of the two conditions above to be satisfied, the first condition must
be satisfied.

We use the following notation in our proofs. We define $w_e^t =
(1-\coren)^{\ell^t(e)}$ to be the \emph{weight} of edge $e$ in phase $t$ where
$\ell^t(e)$ is the load on edge $e$ in the \emph{beginning} of phase $t$. Specifically, $\ell^t(e)$ is the load on edge $e$ in phase
$t$ \emph{before} the update in~\cref{densest:update-load}.
Note that $\ell(e)$ is updated with a value in $\{0, 2\}$
for all phases $1 \leq t \leq T$, unlike the original algorithm given by~\cite{SuVu20} which updates the load with a potentially
non-integer value. We, in fact, only need to modify the proofs of Lemma 3.1 and Lemma 3.2 (and not the proof of Lemma 3.3) from
Su and Vu~\cite{SuVu20} since 
the loads on our edges are guaranteed to be integers throughout all phases of~\cref{alg:dp-densest-z}.
Our additive error allows us to update the loads with only integer values.
Thus, our modification to only update loads using integers actually allows us to simplify our accuracy proof in some aspects.

We use the bound on the symmetric geometric noise (\cref{lem:noise-whp-bound}) in our main lemmas. 
Now, we present the lemmas that we will use to prove our approximation factors. The structure of this section is 
as follows. First, in~\cref{lem:feasible-lp}, we show that if $z < \frac{D^* - \densestsubfactor}{1+10\coren}$, then there must exist a phase $t$ 
where $\sum_{e \in E} \left(w_e^t \left(\alpha_{eu}^t + \alpha_{ev}^t\right)\right) < \sum_{e \in E} w_e^t$. %
Then, in~\cref{lem:approx-densest-subgraph}, we 
show that if $\sum_{e \in E} \left(\hat{w}_e^t \left(\alpha_{eu}^t + \alpha_{ev}^t\right)\right) 
< \sum_{e \in E} \hat{w}_e^t$ for a
given set of weights $\hat{w}_e^t$, then
there exists a $\lambda$ where taking 
all nodes which are incident to at least $\ceil{z/2} + Z$ edges with weight at least $\lambda$
gives a subgraph with density at least $z - \cutoffsubfactor$, \whp, for a large enough constant $c \geq 1$. %
We show in~\cref{lem:integer-solution} that because the loads on all edges are
integers, the set of weights $w'_e$ obtained from~\cref{alg:dp-densest-z} gives a value for $\lambda = (1-\coren)^{\ell'}$ where $\ell'$ is
an integer in $[0, 4T]$ that satisfies the conditions of~\cref{lem:approx-densest-subgraph};
this directly leads to our final approximation factor given in~\cref{lem:densest-subgraph-approx}.
Unlike Lemma 3.2 in~\cite{SuVu20}, we do not need to prove~\cref{lem:approx-densest-subgraph} in terms of a 
factor $F \in (1/2, 1)$ since our additive error allows us to assign \emph{integer} loads to edges
upfront. The reason that Lemma 3.2 in~\cite{SuVu20} requires the factor of $F \in (1/2, 1)$ is that
they allow no additive error; thus, each node must assign all of its $z$ load to its adjacent edges (if it 
has degree at least $\ceil{z/2}$) and the load assigned to the $\ceil{z/2}$-th edge may not be an integer. In our
case, because we have an additional additive error of $O\left(\frac{\log^4 n}{\eps}\right)$, each node can afford to assign
$z \pm \frac{c\log^4 n}{\eps}$ load for constant $c \geq 1$.

\begin{lemma}\label{lem:feasible-lp}
    If $z < \frac{D^* - \densestsubfactor}{1+10\coren}$ for a large enough 
    constant $c \geq 1$ (depending on constant parameter $\coren > 0$), then whp there exists a phase $1 \leq t \leq T$ where $\sum_{e \in E} \left(w_e^t \left(\alpha_{eu}^t + \alpha_{ev}^t\right)\right)
    < \sum_{e \in E} w_e^t$.
\end{lemma}

\begin{proof}
    We modify the proof of \cite[Lemma 3.1]{SuVu20}, which is a standard analysis of the multiplicative weights update
    method of~\cite{AHK12,You95}. As in \cite[Lemma 3.1]{SuVu20}, we prove this theorem via the following statements. 
    This proof makes use of a feasible solution to the dual (\cref{fig:densest-lp}).
    We make a small modification to our algorithm to obtain a potential set of variable values $\alpha_{eu}$ for every $e \in E, u \in e$ 
    for the dual. Namely, we obtain a value for each $\alpha_{eu}$ by computing $\alpha_{eu} = (1+10\coren) \cdot \frac{\sum_{t = 1}^T \alpha^{t}_{eu}}{T}$.
    Note that we do not return the solution to the dual in any part of our algorithm but only make use of it in the proof of this lemma.
    First, as noted before, $\dual(B)$ is feasible if and only if $B \geq \optdensity$. This means that if $B = \dualfactor < \optdensity$,
    then $\dual\left(\dualfactor\right)$ is not feasible. The quantity satisfies $\dualfactor < \optdensity$ 
    when $z < \zbound$. If $\dual\left(\dualfactor\right)$ is not feasible, then naturally our algorithm cannot return a feasible solution. 
    The final statement that we need to show is then: if our algorithm does not return a feasible solution to $\dual\left(\dualfactor\right)$, then there must 
    exist a phase $1 \leq t \leq T$ where $
    \sum_{e \in E} w_e^t \left(\alpha_{eu}^t + \alpha_{ev}^t\right) < \sum_{e \in E} w_e^t$. This
    set of statements directly proves our desired lemma statement. The first two statements are trivial to prove. The 
    crux of this proof is then focused on proving the last statement. To prove the last statement, it suffices to prove
    the contrapositive of the statement. Specifically, we spend the remaining part of this proof proving the statement:
    if $\sum_{e \in E} w_e^t \left(\alpha_{eu}^t + \alpha_{ev}^t\right) \geq \sum_{e \in E} w_e^t$ for all phases $1 \leq t \leq T$,
    then our algorithm returns a feasible solution to $\dual\left(\dualfactor\right)$. Note that returning a feasible solution
    given parameter $z$ necessarily implies that $\dualfactor \geq \optdensity$ and $z \geq \zbound$.
    
    In order for the returned solution to be a feasible solution to the dual, it must satisfy all constraints in the dual.
    We first prove that if $\sum_{e \in E} w_e^t \left(\alpha_{eu}^t + \alpha_{ev}^t\right) \geq \sum_{e \in E} w_e^t$ for all $1 \leq t \leq T$, 
    then $\sum_{e \ni u} \alpha_{eu}$ is upper bounded by $\dualfactor$. This satisfies one of the constraints in the dual. 
    Let $X^t_{u}$ be the noise chosen by node $u$ in phase $t$.
    We must upper bound the sum of the loads so cases when $X^t_u < 0$ are upper bounded by the cases
    when $X^t_u \geq 0$. If $X^t_{u} \geq 0$, then we potentially give weights of $2$ to $X^t_{u}$ additional adjacent edges to $u$. This means that the following holds:
    \begin{align}
        \sum_{e \ni u} \frac{\sum_{t = 1}^T \alpha_{eu}^t}{T} \cdot (1 + 10\coren) &= \frac{(1+10\coren)}{T} \cdot \sum_{t = 1}^T \sum_{e \ni u} \alpha_{eu}^t \leq \frac{(1+10\coren)}{T} \cdot \sum_{t = 1}^T \left(z + \frac{2d\log^3 n}{\eps \coren^3}\right)\label{eq:noise-upper}\\ 
        &= (1+10\coren) \cdot \left(z + \frac{2d\log^3 n}{\eps \coren^3}\right) \leq (1+10\coren)z + \dualadditivefactor\label{eq:simplify-t}\\
        \sum_{e \ni u} \alpha_{eu} &\leq (1+10\coren)z + \dualadditivefactor\label{eq:final-densest-bound}
    \end{align}
    for appropriately large constant values $c, d \geq 1$. The last expression in~\cref{eq:noise-upper} holds because $X_u^t$ is upper bounded by 
    $O\left(\frac{\log^3 n}{\eps}\right)$ by~\cref{lem:noise-whp-bound}. The first
    part of~\cref{eq:simplify-t} follows by canceling $T$. Then, the second expression follows by
    substituting a large enough constant $c\geq 1$ for $2(1+10\coren)d$. 
    Finally,~\cref{eq:final-densest-bound} follows from
    substituting $\alpha_{eu} = (1+10\coren) \cdot \sum_{t =1}^T \frac{\alpha_{eu}^t}{T}$.
    
    To prove that the returned solution satisfies the other constraints, 
    we use a similar potential function as used in the proof of Lemma 3.1 in~\cite{SuVu20}.\footnote{We use the
    version of the proof of Lemma 3.1
    in v4 of~\cite{SuVu20arxiv} where Appendix B had a typo in the last steps which necessitated this
    change in multiplicative factor in the lower bound on the weights and the change in the potential function. 
    Namely in the last steps, 
    they wrote $\Phi(0) = \sum_{e \in E}\frac{(1-\coren/2)^T}{(1-\coren/2)^{(1-\coren)T}}$ which is 
    a typo and the denominator should instead be $(1-\coren)^{(1-\coren)T}$. Thus, our analysis differs somewhat from
    the original analysis of Lemma 3.1 to correct for this typo.}
    The potential function is defined as follows
    \begin{align*}
        \Phi(t) = \sum_{e \in E} \frac{(1-\coren)^{\ell^t(e)} \cdot \left(1 - \coren + 8\coren^2\right)^{T-t}}{\left(1 - \coren + 8\coren^2\right)^{(1-\coren) \cdot T}}.
    \end{align*}
    We show that $\Phi(T) = \sum_{e \in E} \frac{(1-\coren)^{\ell^T (e)}}{\left(1 - \coren + 8\coren^2\right)^{(1-\coren) \cdot T}} < 1$ which implies for every edge
    $e = \{u, v\}$,
    \begin{align}
        (1-\coren)^{\ell^T (e)} &< \left(1 - \coren + 8\coren^2\right)^{(1-\coren) \cdot T}\label{eq:final-weight-greater}\\
        \ell^T(e) \log_{(1 - \coren + 8\coren^2)}(1-\coren) &> (1-\coren)T\label{eq:simplify-log}\\
        \sum_{t = 1}^T \left(\alpha^t_{eu} + \alpha_{ev}^t\right) &> \frac{(1-\coren)T}{\log_{(1 - \coren + 8\coren^2)}(1-\coren)} \label{eq:alpha-weight-greater}\\
        \frac{\sum_{t = 1}^T\left(\alpha^t_{eu} + \alpha^t_{ev}\right)}{T} &> \frac{(1-\coren)}{\log_{(1 - \coren + 8\coren^2)}(1-\coren)} \label{eq:simplify}\\
        \frac{\sum_{t = 1}^T \left(\alpha^t_{eu} + \alpha^t_{ev}\right)}{T} \cdot (1+10\coren) &> \frac{(1-\coren)(1+10\coren)}{\log_{(1 - \coren + 8\coren^2)}(1-\coren)}\label{eq:mult}\\
        \alpha_{eu} + \alpha_{ev} &> 1 \text{ when } \coren \in (0, \densestmultconst).\label{eq:final}
    \end{align}
    The first inequality (\cref{eq:final-weight-greater}) 
    follows since if $\sum_{e \in E} \frac{(1-\coren)^{\ell^T (e)}}{(1 - \coren + 8\coren^2)^{(1-\coren) \cdot T}} < 1$, this means
    that $\frac{(1-\coren)^{\ell^T (e)}}{(1 - \coren + 8\coren^2)^{(1-\coren) \cdot T}} < 1$ for each $e \in E$ 
    since $\frac{(1-\coren)^{\ell^T (e)}}{(1 - \coren + 8\coren^2)^{(1-\coren) \cdot T}} \geq 0$.
    This in turn implies that $(1-\coren)^{\ell^T (e)} < (1 - \coren + 8\coren^2)^{(1-\coren) \cdot T}$. 
    \cref{eq:simplify-log} follows by taking $\log_{(1 - \coren + 8\coren^2)}$ of both sides (where $(1 - \coren + 8\coren^2) \in (0, 1)$).
    Then, the load on edge $e$ is defined to be the sum of the loads assigned to it 
    from each phase $t$. Thus, the total load on edge $\ell^T(e) = 
    \sum_{t = 1}^T \left(\alpha_{eu}^t + \alpha_{ev}^t\right)$. Substituting this expression into~\cref{eq:simplify-log} 
    results in~\cref{eq:alpha-weight-greater}.
    We divide both sides by $T$ in~\cref{eq:simplify}. Then, 
    we multiply both sides by $(1+10\coren)$ in~\cref{eq:mult}. 
    Finally, by definition, $\alpha_{eu} = (1+10\coren) \cdot \sum_{t = 1}^T 
    \frac{\alpha_{eu}^t}{T}$. So, because $\frac{(1-\coren)(1+10\coren)}{\log_{(1 - \coren + 8\coren^2)}(1-\coren)} > 1$
    for all $\coren \in (0, \densestmultconst)$,
    we obtain~\cref{eq:final}, and the solutions satisfy the constraint.
    
    We show that $\Phi$ is non-increasing as $t$ increases; then, if initially, $\Phi(0) < 1$, then we directly prove our desired statement.
    The rest of the proof follows almost identically (except for the fix for the typo)
    from the proof of Lemma 3.1 in~\cite{SuVu20} but we repeat the proof for completeness.
    For any $0 \leq t < T$, we have the following:
    \begin{align}
        \Phi(t + 1) &= \sum_{e \in E} \frac{(1-\coren)^{\ell^{t + 1}(e)} \cdot (1 - \coren + 8\coren^2)^{T - t - 1}}{(1 - \coren + 8\coren^2)^{(1-\coren) T}}\label{eq:t+1}\\
        &= \sum_{e \in E} \frac{(1-\coren)^{\ell^t(e)} (1 - \coren)^{\alpha_{eu}^t + \alpha_{ev}^t} (1 - \coren + 8\coren^2)^{T - t - 1}}{(1 - \coren + 8\coren^2)^{(1-\coren)T}}\label{eq:next-core-nums}\\
        &\leq \sum_{e \in E} \frac{(1-\coren)^{\ell^t(e)} (1 - \coren(\alpha_{eu}^t + \alpha_{ev}^t) + \coren^2 (\alpha_{eu}^t + \alpha_{ev}^t)^2/2) (1 - \coren + 8\coren^2)^{T-t-1}}{(1 - \coren + 8\coren^2)^{(1-\coren)T}}\label{eq:upper-bound}\\
        &\leq \left(\sum_{e \in E} \frac{(1-\coren)^{\ell^t(e)} (1 - \coren + 8\coren^2)^{T-t-1}}{(1 - \coren + 8\coren^2)^{(1-\coren)T}}\right) - (\coren - 8\coren^2)\left(
        \sum_{e\in E} \frac{(1-\coren)^{\ell^t(e)} (1 - \coren + 8\coren^2)^{T-t-1}}{(1 - \coren + 8\coren^2)^{(1-\coren)T}}\right)\label{eq:assumption-upper}\\
        &= (1 - \coren + 8\coren^2) \cdot \left(\sum_{e \in E} \frac{(1-\coren)^{\ell^t(e)} (1 - \coren + 8\coren^2)^{T-t-1}}{(1 - \coren + 8\coren^2)^{(1-\coren)T}}\right)\label{eq:cancel-2}\\
        &=  \left(\sum_{e \in E} \frac{(1-\coren)^{\ell^t(e)} (1 - \coren + 8\coren^2)^{T-t}}{(1 - \coren + 8\coren^2)^{(1-\coren)T}}\right) = \Phi(t).\label{eq:final-upper}
    \end{align}
    \cref{eq:t+1} is obtained by definition of the potential function $\Phi(t)$.
    \cref{eq:next-core-nums} follows because $\ell^{t+1}(e) = \ell^t(e) + \alpha_{eu}^t + \alpha_{ev}^t$. 
    Using the well-known fact that $(1 - a)^b \leq e^{-ab} \leq 
    1 - ab + (ab)^2/2$ for $0\leq ab \leq 1$, we let $a = \eta$ and $b = \alpha^{t}_{eu} + \alpha_{ev}^t$; assuming $\eta \leq 1/4$
    ensures $0 \leq \eta \left(\alpha^{t}_{eu} + \alpha_{ev}^t\right) \leq 1$ since $0 \leq \left(\alpha^{t}_{eu} + \alpha_{ev}^t\right) \leq 4$. 
    We obtain~\cref{eq:upper-bound} by using the aforementioned inequalities. 
    Assuming $\sum_{e \in E} w^t_e \left(\alpha_{eu}^t + \alpha_{ev}^t\right) \geq \sum_{e \in E} w_e^t$, we obtain~\cref{eq:assumption-upper} by lower bounding
    $\sum_{e \in E} (1 - \coren)^{\ell^t(e)}(\alpha_{eu}^t + \alpha_{ev}^t) = \sum_{e \in E} w_e^t\left(\alpha^t_{eu} + \alpha_{ev}^t\right) \geq \sum_{e \in E} w_e^t = \sum_{e \in E} (1 - \coren)^{\ell^t(e)}$ 
    since by definition $w_e^t = (1 - \coren)^{\ell^t(e)}$ and upper bounding
    $(\alpha_{eu}^t + \alpha_{ev}^t)^2 \leq 16$. \cref{eq:cancel-2} simplifies the previous expression.
    \cref{eq:final-upper} shows that $\Phi(t + 1) \leq \Phi(t)$ for all $0 \leq t < T$.
    Hence,
    \begin{align*}
        \Phi(T) \leq \Phi(T - 1) \leq \dots &\leq \Phi(0) \\
        &= \sum_{e \in E} \frac{(1 - \coren + 8\coren^2)^{T}}{(1 - \coren + 8\coren^2)^{(1-\coren)T}}\\
        &\leq \sum_{e \in E} \exp\left(-\left(\coren^2 - 8\coren^3\right) T\right)\\
        &\leq m \cdot \exp(-\left(\coren^2 - 8\coren^3\right) \cdot (c_0 \ln n)/\left(\coren^2 - 8\coren^3\right)) = O\left(\frac{1}{\poly(n)}\right).
    \end{align*}
    We obtain the last expression since $T = \frac{c_0\ln n}{\coren^3} \geq \frac{c_0\ln n}{\coren^2 - 8\coren^3}$ for large
    enough constant $c_0 > 0$ and $\coren \in (0, \densestmultconst)$.
    The above shows that $\Phi(T) < 1$ and we conclude our proof of the lemma.
\end{proof}

Knowing the condition that whp at least one phase $t$ exists where $\sum_{e \in E} w_e^t \left(\alpha_{eu}^t + 
\alpha_{ev}^t \right) < \sum_{e \in E}
w_e^t$ allows us to prove~\cref{lem:approx-densest-subgraph}.

Let $S(z)$ be the set of all possible \emph{integer} assignments of weights to $\alpha_{eu}^t$
for all $e = \{u, v\}$ and endpoints $u, v$. 

\begin{lemma}\label{lem:approx-densest-subgraph}
Suppose there exists a set of (non-negative) weights $\{\hat{w}^t_e\}$ 
obtained by our algorithm in phase $t$ such that
\begin{align}
    \sum_{e \in E} \hat{w}^t_e > \sum_{e \in E} \hat{w}^t_e \left(\halpha_{eu}^t + \halpha_{ev}^t\right),\label{eq:weight-lower}
\end{align}
where $\hat{\alpha}^t_{eu}, \hat{\alpha}^t_{ev} \in S(z)$ is determined by our algorithm during phase $t$.
Let $V_{\lambda}''$ denote the set of nodes, determined by our algorithm,
that are incident to at least $\ceil{z/2} + Z$ real
edges $e$ with $\hw^t_e \geq \lambda$ where $Z \sim \geom\left(\cutoffnoise\right)$. Then, there exists $\lambda$ such that 
$\delta\left(G[V''_{\lambda}]\right) > z - \cutoffsubfactor$ whp for sufficiently large constant $c \geq 1$ (which depends
on constant parameter $\coren \in (0, \densestmultconst)$).
\end{lemma}

\begin{proof}
This proof closely follows the 
proof of Lemma 3.2 in~\cite{SuVu20} (and, hence, also the proof in
\cite{BGM14}). 
The main difference in our proof is to show that the statement still holds when 
we lose some part of the loads due to the dummy edges and the random noise $Z$; 
specifically, we care more about the instances when $Z < 0$, which is the more difficult setting to prove. 
First, note that $\halpha_{eu}^t \in \{0, 2\}$
by definition of $S(z)$ and how we assign load to edges within a phase of the algorithm. Then, let $E_{\lambda} := \{e \in E: \hw^t_e \geq \lambda\}$ be the
set of edges where $\hw^t_e \geq \lambda$ and let $\deg_{\lambda}(v)$ be the number of 
edges in $E_{\lambda}$ that are incident to $v$. To provide some intuition, $\hat{w}^t_e$ \emph{decreases} when
the load on the edges \emph{increases}. Thus, edges with higher weights correspond with edges with lower loads
and new loads within a phase are assigned to edges in order from largest weights (lowest load) to smallest weights (highest load).

Now, we define a variable $d_{\lambda}(v)$ for each node $v$ which lower bounds 
the the sum of current loads $\alpha^t_{eu}$ assigned to all \emph{real} edges adjacent to $v$
from the current phase $t$. Recall from our algorithm that out of the real edges with smallest accumulated load a noise $X_v$ is
picked that can cause more or less of these real edges to obtain new loads. Edges with weight larger than $\lambda$ (equivalent
to number of edges with load smaller than $\log_{(1-\coren)}(\lambda)$) receive loads starting from the
largest weight to smallest weight; however, not all of the $\min(\ceil{z/2}-1, \deg(v))$ edges with the largest weight may
receive loads due to the presence of dummy edges. This difference due to dummy edges is what we capture below. We \emph{lower bound} the sum of 
the loads given to the highest weight edges adjacent to $v$ \whp in terms of $\min(\ceil{z/2}-1, \deg(v))$.
We thus define $d_{\lambda}(v)$ to be
\begin{align*}
    d_{\lambda}(v)=\left\{
                \begin{array}{ll}
                  2\deg_{\lambda}(v), \text{ if } \deg_{\lambda}(v) < \ceil{z/2} - \frac{cT\log^2 n}{\eps}\\
                  \max(0, z - \frac{2cT \log^2 n}{\eps}), \text{ otherwise}.
                \end{array}
              \right.
\end{align*}
This is the key difference between our procedure and the proof of Lemma 3.3 
in~\cite{SuVu20}. We require $\frac{2cT\log n}{\eps}$ to be subtracted from the sum of the loads
of adjacent edges since \whp, at most 
$\frac{cT \log n}{\eps}$ dummy adjacent edges resulting from $X_v$  for each $v$ are added to each node from~\cref{densest:dummy-edges}. This means that part of the total load of $z$ could be spent on these dummy edges, decreasing the 
total load on the real edges. %
Thus, the subtracted value
of $\frac{2cT\log n}{\eps}$ from $\deg_{\lambda}(v)$ reflects the load that is lost to these dummy edges.
We first note that $d_{\lambda}(v) \leq 
\sum_{e \ni v: \hat{w}_e \geq \lambda} \halpha_{ev}^t$ with high probability.
If $\deg_{\lambda}(v) + \frac{cT\log n}{\eps} < \ceil{z/2}$, then
a number of real edges plus at most $\frac{cT\log n}{\eps}$ dummy edges are each given $2$ load. 
If $\deg_{\lambda}(v) \geq \ceil{z/2} - \frac{cT\log n}{\eps}$, %
then $2\deg_{\lambda}(v) \geq z - \frac{2cT\log n}{\eps}$. %
When $X_v > 0$, additional edges receive loads, maintaining our lower bound. 
Thus,

\begin{align}
\sum_{v \in V}d_{\lambda}(v) \leq \sum_{v \in V} \sum_{e \ni v: \hw^t_e \geq \lambda} \halpha_{ev}^t
= \sum_{e \in E: \hw^t_e \geq \lambda} 
(\halpha_{ev}^t + \halpha_{eu}^t).\label{eq:sum-degs}
\end{align}

Suppose the weights are upper bounded by $b$: $a \leq \hw^t_e 
\leq b$ for all $e \in E$. Then, 

\begin{align*}
    \int_{a}^b |E_{\lambda}| d\lambda &= \sum_{e \in E} \hw^t_e \\
    &> \sum_{e \in E} \hw^t_e \left(\halpha_{eu}^t + \halpha_{ev}^t\right)
    && \text{ by~\cref{eq:weight-lower} }\\
    &= \int_{a}^b \sum_{e \in E: \hw_e^t \geq \lambda} (\halpha^t_{eu} + \halpha^t_{ev})
    d\lambda\\
    &\geq \int_{a}^b \sum_{v \in V} d_{\lambda}(v) d\lambda && \text{ 
    by~\cref{eq:sum-degs} }
\end{align*}
Thus, there exists a $a \leq \lambda \leq b$ where $|E_{\lambda}| > 
\sum_{v \in V}d_{\lambda}(v)$. 
Recall $V_{\lambda}''$ is the set of nodes with 
$\deg_{\lambda}(v) \geq \ceil{z/2} + \zfactor - \frac{c_1\log^4 n}{\eps}$ 
where $Z \sim \geom\left(\cutoffnoise\right)$. Let $c_2 \geq 1$ be an appropriately
large constant such that $\frac{c_1\log^4 n}{\eps} \leq \frac{c_2T^2 \log^2n}{\eps}$. 
Then, \whp,

\begin{align}
    |E(V_{\lambda}'')| &\geq |E_{\lambda}| - \sum_{v \in V - V_{\lambda}''} 
    \deg_{\lambda}(v)\label{eq:lambda-density} \\
    &> \sum_{v \in V} d_{\lambda}(v)
    - \sum_{v \in V - V_{\lambda}''} \deg_{\lambda}(v)\label{eq:lambda-substitute}\\
    &\geq \left(|V_{\lambda}''| \cdot \max\left(0, z - \frac{2(c+c_2)T^2\log^2 n}{\eps}\right) 
    + \sum_{v \in V- V_{\lambda}''} \left(2 \deg_{\lambda}(v)\right)\right) - \sum_{v \in V - V_{\lambda}''} \deg_{\lambda}(v)\label{eq:lambda-dv}\\
    &\geq |V_{\lambda}''| \cdot \max\left(0, z - \frac{2(c+c_2)T^2\log^2 n}{\eps}\right).\label{eq:lambda-final}
\end{align}
\cref{eq:lambda-density} follows since the number of edges in the induced subgraph of $E(V_{\lambda}'')$ is lower bounded
by the number of edges in $|E_{\lambda}|$ minus the set of those edges that are adjacent to nodes not in $V_{\lambda}''$; 
this set of edges can be each be counted at most twice.
\cref{eq:lambda-substitute} follows from substituting $|E_{\lambda}| > \sum_{v \in V} d_{\lambda}(v)$. 
\cref{eq:lambda-dv} follows from the definition of $d_{\lambda}(v)$ and the fact that $Z \leq \frac{c T^2\log^2 n}{\eps}$ 
\whp for constant $c \geq 1$; 
thus, every $v \in V_{\lambda}''$ has $\deg_{\lambda}(v) \geq \ceil{z/2} - \frac{(c+c_2) T^2\log^2 n}{\eps}$ \whp
which means that the load for each of these nodes is at least $2\deg_{\lambda}(v) \geq z - \frac{2(c+c_2) T^2\log^2 n}{\eps}$.
Furthermore, we can guarantee that each node in $v \in V - V_{\lambda}''$ has $\deg_{\lambda}(v) \leq \ceil{z/2} -\frac{2c_2 T^2\log^2 n}{\eps} + \frac{c T^2\log^2 n}{\eps} < \ceil{z/2} 
- \frac{cT\log^2 n}{\eps}$ for appropriately large constants $c_1, c_2 \geq 1$
which means that the load given to them is $2\deg_{\lambda}(v)$.
Finally,~\cref{eq:lambda-final} follows because the last two terms is lower bounded by $0$.

The last step is simple; since the right hand side of~\cref{eq:lambda-final} is lower bounded by
$0$, then $|E(V_{\lambda}'')| > 0, |V_{\lambda}''| > 0$ and so 
$\frac{|E(V_{\lambda}'')|}{|V_{\lambda}''|} = \dense\left(G[V_{\lambda}'']\right) 
> z - \frac{2(c+c_2) T^2\log^2 n}{\eps}$. 
\end{proof}

\begin{lemma}\label{lem:integer-solution}
Given a $z < \frac{D^* - \densestsubfactor}{1+10\coren}$ for appropriately large constant $c > 0$, then, for 
$\coren \in (0, \densestmultconst)$,~\cref{alg:dp-densest-z} will output a subgraph of density at least 
$z - O\left(\frac{\log^4 n}{\eps}\right)$, \whp.
\end{lemma}

\begin{proof}
    \cref{lem:approx-densest-subgraph} proves there exists a $\lambda$ such that $\dense(G[V_{\lambda}'']) > z - O\left(\frac{\log^4 n}{\eps}\right)$.
    Since all loads given by~\cref{alg:dp-densest-z} are integers, the load on any edge is an integer in $[0, 4T]$.
    This means that the weights on the edges $(1-\coren)^{\ell^t(e)}$ in the 
    proof of~\cref{lem:approx-densest-subgraph} are also obtained from these integer loads. Thus, it is sufficient 
    to look at all possible integer loads $\ell \in [0, 4T]$ which give the appropriate weight cutoffs
    $\lambda = (1-\coren)^\ell$ and we prove the lemma by~\cref{lem:approx-densest-subgraph,lem:feasible-lp} 
    which ensure such a $\lambda$ exists.
    
    Finally, to pass the last check just requires setting the constant $c_3 \geq 1$ to a large enough value since $|Y| = O\left(\frac{\log^4 n}{\eps}\right)$ \whp. By the proof above, 
    for small enough $z$, we are guaranteed 
    a set of nodes of sufficiently large density.
\end{proof}

\begin{lemma}\label{lem:densest-subgraph-approx}
    \cref{alg:dp-densest} returns a $\left(\kcoremultfactor, O\left(\frac{\log^4 n}{\eps}\right)\right)$-approximate densest
    subgraph for any constant $\coren \in (0, \densestmultconst)$ with high probability.
\end{lemma}

\begin{proof}
    First, as long as $z \leq \frac{\optdensity - \densestsubfactor}{1+10\coren}$, we obtain our desired approximation.
    \cref{alg:dp-densest} searches for the largest value of $z$ that satisfies this constraint by searching all 
    $(1+\coren)^i$ for all $i \in [\floor{\log_{(1+\coren)}n}]$. In the worst case, $(1+\coren)^i = \frac{\optdensity - \densestsubfactor}{1+10\coren} + y$ by some very small $y \in (0, 1)$. This means that the closest 
    $(1+\coren)^{i-1}$ is the largest value of $z$ that is smaller than the cutoff; 
    in other words, $z = (1+\coren)^{i-1} \geq \frac{\optdensity - \densestsubfactor}{(1+\coren)(1+10\coren)}$. 
    Thus, by~\cref{lem:integer-solution}, we obtain a set of nodes whose induced subgraph has density at least
    $z - \frac{c\log^4 n}{\eps}$ for an appropriately large constant $c \geq 1$. Together, the returned 
    set of nodes has density at least $z - \frac{c\log^4 n}{\eps} \geq \frac{\optdensity - \densestsubfactor}{(1+\coren)(1+10\coren)} - \frac{c\log^4 n}{\eps} \geq (\kcoremultfactor)D^* 
    - O\left(\frac{\log^4 n}{\eps}\right)$. In the case when no value of $z$ returns a set of nodes, all the nodes in the 
    graph are returned. This case occurs when all values of $z$ passed into~\cref{alg:dp-densest-z} are too large
    which means that $1 \geq \frac{\optdensity - \densestsubfactor}{(1+10\coren)}$ and $D^* \leq  \densestsubfactor + 1 + 10\coren$ and returning all the nodes in the graph will result in an induced subgraph with density 
    $\dense(G) \leq D^* \leq \densestsubfactor + 1 + 10\coren = O\left(\frac{\log^4 n}{\eps}\right)$ which falls within 
    our additive error bound.
\end{proof}

\subsection{Efficiency Analysis}\label{sec:densest-runtime}

We analyze the runtime of our $\eps$-DP densest subgraph algorithm in this section.
Specifically, we show that we obtain the same runtime as the original non-private algorithm.
As before, we assume that $\coren \in (0, \densestmultconst)$ is constant. Throughout our runtime
analysis, we use the standard assumption
that obtaining a sample from the symmetric geometric
distribution requires $O(1)$ time. Such an assumption is reasonable for the following reason. 
One can generate a random variable from the geometric distribution in $O(1)$ time by generating 
a uniformly random real number in binary (up to a certain precision) and outputting the position of the first $1$. 
A recent paper of Balcer and Vadhan~\cite{BV18} shows how to generate such random variables on finite memory
machines. Once such a random variable has been generated, we can generate a \emph{symmetric geometric} random variable
by first deciding whether the random variable is $0$, positive, or negative with appropriate probability determined
by the parameter $b$. If the chosen symmetric random variable is either positive or negative, then we 
use the random variable generated by the geometric distribution with the appropriate sign as the output.

\begin{theorem}\label{thm:runtime-densest}
    The runtime of~\cref{alg:dp-densest} is $O\left((n + m)\log^3 n\right)$.
\end{theorem}

\begin{proof}   
    This analysis assumes the input parameter $\coren \in (0, \densestmultconst)$ is constant.
    First, the search in~\cref{line:i-iteration} of~\cref{alg:dp-densest} 
    requires $O(\log n)$ iterations of~\cref{alg:dp-densest-z}, which means that the runtime of
    \cref{alg:dp-densest-z} is multiplied by $O(\log n)$. %
    Next, we determine the runtime of~\cref{alg:dp-densest-z}.
    For a given $z$, our algorithm runs through $T = O\left(\log n\right)$ phases. Each phase 
    requires $O(m)$ time to sort the loads (using radix sort) 
    and to determine the new loads to apply to each edge. 
    Then, we run $O\left(\log n\right)$ trials of finding a subgraph of sufficient density. 
    Each of these trials requires $O(m)$ time. Altogether, the runtime for a fixed parameter $z$ of this algorithm is 
    $O\left((m + n)\log^2 n\right)$. %
    Thus, in total, our running for~\cref{alg:dp-densest} is dominated by $O(\log n)$ iterations of~\cref{alg:dp-densest-z}
    which take a total of $O\left((n + m)\log^3 n\right)$ time.
\end{proof}

} %
\section{Differential Privacy from \LA Algorithms}\label{sec:framework}

The graph algorithms studied in this paper all have a generalizable structure that is conducive to privacy. 
Our framework applies to a number of parallel, distributed, and dynamic
algorithms~\cite{LSYDS22,BHNT15,SuVu20,HNW20,BGM14,chan2021},
described in detail in~\cref{sec:alg,sec:densest}.
In this section, we describe a general class of algorithms and 
show that we can transform them to be $\eps$-edge DP (\cref{sec:framework-privacy}).
\conffull{}{
Some even become $\eps$-LEDP (\cref{sec:framework-ledp}).}
We call the algorithms that have these characteristics \emph{\localadjust}
algorithms. Beyond the various algorithms we study in this paper, we believe
that our generalization and the privacy framework can be applied to a 
broader set of non-private graph algorithms, most naturally, in the parallel and distributed settings.
Furthermore, interestingly, the techniques that we use to obtain our privacy guarantees result
in only a small \emph{additive} polylogarithmic error while maintaining the \emph{same multiplicative} approximation
factor of each of the original non-private algorithms given in~\cref{sec:alg,sec:densest}.
However, we do not have a general statement bounding the utility (or error) of private algorithms obtained via our framework
and leave this as an interesting open question. %

\subsection{\LA Graph Algorithms}

We call a graph algorithm $\alg$ on input graph $G = (V, E)$ \defn{\localadjust} if it
has the following characteristics. The algorithm proceeds in at most $K$ total phases.
Two or more phases may occur in parallel if the pairs of phases $k_1, k_2 \leq K$ do not depend on each other. Each node $v$ maintains an 
internal state $I_{v, p}$ parameterized by the phase number $p$; similarly, each edge $e$ maintains an internal state $I_{e, p}$
also parameterized by $p$. Initially, in phase $0$, all states
are set to default identical values.
Since the phases may be processed in parallel, for each phase $p$, 
we refer to the previous phase that phase $p$ depends on as $\prev$, where 
$\prev < p$, but $\prev$ is not necessarily equal to $p - 1$.

During each of the at most $K$ phases, 
each node $v \in V$ computes a function using \emph{only} 
the previous states $I_{w, \prev}, I_{e, \prev}$
of its immediate one-hop neighborhood, where $w \in N(v)$ and $e = \{v, a\}$ for 
any $a \in N(v)$.
Notably, these functions determine for a node $v$ whether its neighbors and/or its 
incident edges satisfy a condition. 
Then, $v$ uses another function with the \emph{number} of neighbors 
or incident edges that satisfy the condition as input
to compute a new state. 
Each edge $e \in E$ also computes a function using \emph{only} 
information \emph{received} from its two endpoints.
Formally, these functions are defined in the following paragraphs.
Importantly, all of the functions satisfy a ``local'' property 
where any edge insertion or deletion, $e' = (u, v)$, in the graph 
affects the count of the number of neighbors that satisfy
the condition of only $u$ or $v$ and no other nodes. 
The output of the function is not changed for any other node $w \not\in \{u, v\}$
or any other edge $e = (i, j)$ where neither $i$ or $j$
is $u$ or $v$. This is the \emph{locally adjustable} specification of the type of algorithms that
we are considering. 

\paragraph{Node functions} 

Let $B$ be a predicate that can be satisfied (or not) by a neighbor 
$w \in N(v)$ of $v$. Node $v$ has a deterministic function that is evaluated in each phase $p$:
\begin{align*}
\neighbf_v(w) = 
\left\lbrace
\begin{array}{r@{}l}
    1, & \text{ if $w\in N(v)$ and $I_{w, \prev}$ }\\
       & \text{ satisfies condition $B$ } \\
    0, & \text{ otherwise }
\end{array}
\right.
\end{align*}
that takes as input a neighbor, $w \in N(v)$, of $v$
and outputs a  $0$ or $1$ bit for the neighbor indicating whether the neighbor satisfies $B$ using
the previous state of the neighbor $I_{w, \prev}$. 
Whether $w$ satisfies $B$ is determined by the state $I_{w, \prev}$ of node $w$, 
parameterized by the last phase $\prev$ that $p$ depends
on. 

For clarity, we emphasize a few crucial observations regarding $B$.
Suppose, without loss of generality, that an edge $e = \{u, v\}$ is inserted in the beginning of phase $p$. 
This means that only the
count of the number of neighbors of $u$ or $v$ that satisfy $B$ is affected compared
to the case when $e$ is not inserted. 
Specifically, this count can increase by at most $1$. Since the previous 
states $I_{v, p'}$ for all $p' < p$ are fixed prior to the insertion, the insertion
cannot affect the output $\neighbf_v(w)$ for any other $w \neq u \in N(v)$. Hence, 
the count for the number of neighbors of $v$ that satisfy $B$ increases by $1$ (compared to the case
when $e$ is not inserted)
when $\neighbf_v(u) = 1$. Symmetrically, for an edge deletion, the number of neighbors that satisfy
$B$ can decrease by at most $1$ compared to the case when $e$ is not deleted.
If the update is not incident to a node $w$, then whether $B$ is satisfied or not
\emph{does not change} for any neighbor of $w$. This means that $\neighbf_v$ 
is a ``local'' function where
edge updates in the graph can only affect their incident endpoints.

Similarly, let $C$ be a condition that can be satisfied (or not) by an 
incident edge to $v$. Again, we specify the trivial ``local'' requirement for $C$
that the addition or deletion of any edge $e = \{u, v\}$ (with an arbitrary state $I_{e, \prev}$)
at the beginning of phase $p$ changes 
whether $C$ is satisfied only for edge $e$; this is a trivial requirement
since edge $e$ did not exist prior to the insertion of $e$
(and the edge $e$ no longer exists after the deletion of $e$).
Then, $v$ %
has another deterministic function that is also evaluated in phase $p$,
\begin{align*}
\adjf_v(\{v, w\}) = 
\left\lbrace
\begin{array}{r@{}l}
1, & \text{ if $\{v, w\} \in E$ and $I_{\{v, w\}, \prev}$} \\ 
   & \text{ satisfies condition $C$ }\\
0, & \text{ otherwise } 
\end{array}
\right.
\end{align*}
that takes as input an incident edge to $v$ and outputs a  $0$ or $1$ bit depending on whether the edge satisfies
$C$. As before, whether $\{v, w\}$ satisfies $C$ depends on the previous state $I_{\{v, w\}, \prev}$ of the edge $\{v, w\}$. 

Let $\nodef_v$ be a deterministic function that 
updates the state of $v$ using only the \emph{number} of neighbors, $n_{v, p}$ and/or edges, $e_{v, p}$, 
that satisfy the conditions $B$ and $C$, respectively. Notably, the function does not require knowledge about the state 
of the neighbors or edges that satisfy the condition. 
Specifically, let $n_{v, p} = \left|\left\{w \in N(v) : \neighbf_v(w) = 1\right\}\right|$ and $e_{v, p} = \left|\left\{\{v, w\} \in E :
\adjf_v\left(\{v, w\}\right) = 1\right\}\right|$. In phase $p$, the function $\nodef_v(I_{v, \prev}, n_{v, p}, e_{v, p}) \rightarrow I_{v, p}$
outputs the next state of $v$ provided the state $I_{v, \prev}$ of $v$ from the most recent phase $\prev$ that $p$ depends on
and the computed values $n_{v, p}$ and $e_{v, p}$.

Finally, on the set of incident edges to $v$ that satisfy condition $C$, 
$\{\{v, w\} \in E : \adjf_v(\{v, w\}) = 1\}$,
$\alg$ uses each edge's function, $\edgef_{\{v, w\}}$, to update the state of the edge, the details of which are given next.

\paragraph{Edge function}

Each edge $e = \{u, v\}$ has a function $\edgef_e$ that takes its previous edge state and real-valued
inputs from its adjacent nodes and outputs its current state. Namely, in phase $p$, edge $e$ computes $\edgef_{e}(I_{e, \prev}, i_u, i_v) \rightarrow I_{e, p}$
to determine its next state $I_{e, p}$ using its previous state $I_{e, \prev}$ and inputs from its endpoints, $i_u$ and $i_v$.
\\\\
The algorithm proceeds with the next phases until a \emph{stopping condition} is satisfied for the entire graph.
The stopping function of the algorithm is based on a \emph{threshold function}, which
takes as input the states of the nodes and edges computed in the current phase $p$. 
If the number of nodes/edges that satisfy the condition
exceeds a threshold, then the algorithm terminates. 
Specifically, these stopping functions are defined as follows.

\paragraph{Stopping functions} The stopping functions determine whether the algorithm stops running or continues running
with the next phase. There are stopping functions for each individual node and also a global stopping function that
determines whether a certain number of nodes and edges
that satisfy a condition $F$ is at least some threshold $T$. The individual stopping function prevents a particular node from 
participating in the next phases while a global stopping function stops the algorithm.
The individual stopping function for each node, $\stopf_v$, relies on how many 
neighbors' states or neighboring adjacent edges' states satisfy a 
condition $F$. We denote the number of neighbors and adjacent edges that satisfy $F$ by 
$s_{v, p} = \left|\left\{w \in N(v) : I_{w, p} \text{ satisfies } F\right\}\right|$
and $t_{v, p} = \left|\left\{w \in N(v): I_{\{v, w\}, p} \text{ satisfies } F\right\}\right|$, respectively.
Then, we define $\stopf_v$ as follows, for some fixed constants $c_1, c_2 \geq 0$ and fixed threshold $T \geq 0$:

\begin{align*}
    \stopf_v\left(s_{v, p}, t_{v, p}\right) = 
    \left\lbrace
    \begin{array}{r@{}l}
        1, & \text{ if } c_1 \cdot s_{v, p} + c_2 \cdot t_{v, p} \geq T;\\
        0, & \text{ otherwise. } 
    \end{array}
    \right.
\end{align*}

We define the global stopping function, in each phase $p$, using $\gcount_p = \left|\left\{v \in V : I_{v, p} \text{ satisfies } F \right\}\right|$ and $t_p = \left|\left\{e \in E : I_{e, p} \text{ satisfies } F \right\}\right|$, for some fixed constants $c_3, c_4 \geq 0$ and fixed threshold $T \geq 0$:

\begin{align*}
    \globalstopf\left(s_p, t_p\right) = 
    \left\lbrace
    \begin{array}{r@{}l}
        1, & \text{ if $c_3 \cdot s_p + c_4 \cdot t_p \geq T$ };\\
        0, & \text{ otherwise. } 
    \end{array}
    \right.
\end{align*}

\paragraph{Output function} 
Once the algorithm terminates, each node outputs an answer to the problem using output functions $\outf_v(I_{v, p}) 
\rightarrow \mathbb{Z}$, where $\mathbb{Z}$ is the set of integers. There may exist a global function
$\globaloutf(\{I_{c, p} : c \in V \cup E\}) \rightarrow \mathbb{Z}$, which takes the internal states of the nodes and edges and outputs an integer answer.
\\\\
A non-trivial number of parallel and distributed graph algorithms 
are \localadjust,
including the non-private algorithms from prior sections 
for \kc decomposition, densest subgraphs, and low out-degree orderings. 
In the next section, we discuss how to obtain $\eps$-edge DP
graph algorithms from \localadjust graph algorithms.

\subsection{Edge Differential Privacy from Local Adjustability}\label{sec:framework-privacy}

We first show how to obtain $\eps$-edge DP \localadjust algorithms. Then, a slight modification to the \localadjust 
conditions also allows us to obtain $\eps$-\emph{LEDP} algorithms. We leave as an interesting open question
proving general utility bounds for our framework.

The main idea here is to show that, on edge-neighboring
graphs $G = (V, E)$ and $G'=(V, E')$, the probability 
that the \emph{same} states are maintained
over the at most $K$ phases satisfy~\cref{def:dp}. We prove this by conditioning
on the states from previous phases.
Then, we show via the chain rule that this implies that our algorithm is $\eps$-edge DP. To do 
this, we make several modifications to the node, edge, stopping, and output functions. 
Our privacy framework is given as follows:

\paragraph{Privacy Framework}
Suppose we are provided a \localadjust algorithm $\alg$. Then
we formulate the following mechanism $\mech(G, \interns_{\prev}, \alg, p) \rightarrow \left(\interns_p, \outputs\right)$ 
performed by the curator which takes as input
the graph $G$, the node and edge states of $G$ from phase $\prev$, the \localadjust algorithm $\alg$, 
and the phase number $p$. Mechanism
$\mech$ outputs the next set of states $\interns_p$ for phase $p$
if the algorithm is still running or $\emptyset$ if the algorithm has stopped. 
It also outputs the set of outputs $\outputs$ if the algorithm has stopped
or $\emptyset$ if the algorithm is still running.
The mechanism modifies $\alg$ in the following ways:

\begin{itemize}
    \item For each node $v$, the node computes $\hatn_{v, p}$ and $\he_{v, p}$ by first sampling $X_{v, p}, Y_{v, p} \sim \geom(\eps/20K)$ and 
    calculates $\hatn_{v, p} = n_{v, p} + X_{v, p}$ and $\he_{v, p} = e_{v, p} + Y_{v, p}$. Node $v$ uses $\hatn_{v, p}$ and $\he_{v, p}$ instead of $n_{v, p}$
    and $e_{v, p}$ in
    $\nodef_v(I_{v, \prev}, \hatn_{v,p}, \he_{v, p})$. 
    \item When determining the stopping condition, node $v$ samples $S_{v, p}, T_{v, p} \sim \geom(\eps/(20K))$
    and computes $\hs_{v, p} = s_{v, p} + S_{v, p}$ and $\hatt_{v, p} = t_{v, p} + T_{v, p}$. Node $v$ uses $\hs_{v, p}$ and $\hatt_{v, p}$ 
    instead of $s_{v, p}$ and $t_{v, p}$ in $\stopf_{v}(\hs_{v, p}, \hatt_{v, p})$.
    The curator also samples $S_{p}, T_p \sim \geom(\eps/(5K))$ and computes $\hs_p = s_p + S_p$ and $\hatt_p = t_p + T_p$. 
    The curator uses $\hs_p$ and $\hatt_p$ as input into $\globalstopf(\hs_p, \hatt_p)$.
    \item Let $\gs_{\globaloutf}$ and $\gs_{\outf_v}$ be the global sensitivities of $\globaloutf$ and $\outf_v$, respectively. To compute the output after satisfying the stopping condition, each node $v$ samples $Q_{v, p} \sim \geom(\eps/(10 \cdot \gs_{\outf_{v}} \cdot K))$
    and outputs $\outf_{v}(I_{v, p}) + Q_{v, p}$. Then, the curator samples $W_{p} \sim \geom(\eps/(5 \cdot \gs_{\globaloutf} \cdot K))$ and outputs 
    $\globaloutf(\interns_p) + W_p$. 
\end{itemize}

The main intuition behind our privacy framework is derived from our $\eps$-edge DP densest subgraph
algorithm (\cref{sec:densest}) regarding the creation of 
\defn{dummy} edges that satisfy the conditions of the various functions.
These dummy edges account for the case when 
the extra edge $e' \in E' \setminus E$ where $e' \in E'$ 
satisfies any of the functions.
The number of these dummy edges that are added for each node
is drawn from a symmetric geometric distribution using the sensitivities of the
appropriate functions 
\conffull{that we analyze in the full version of our paper.}{in~\cref{cor:sensitivity-count,lem:sensitivity-adjacents,lem:sensitivity-st,cor:sensitivity-threshold}.}

\conffull{}{
We now prove some useful lemmas that will help us in proving the privacy of our framework.

\paragraph{Sensitivity}
We first make a few intuitive observations before proving the privacy of our DP framework. 
First, instead of using the global sensitivity of function $\nodef_v$, for node $v$, 
we instead use the sensitivity of a function $F$ which takes a graph $G = (V, E)$, 
a phase $p$, and a set of valid states, $\mathcal{I}_V$, for all nodes $v \in V$
and a set of valid states, $\mathcal{I}_{V \times V}$, for every pair of two 
nodes. Function $F\left(G, p, \mathcal{I}_V, \mathcal{I}_{V \times V}\right)$ then outputs
$\sum_{v \in V} n_{v, p} + \sum_{e \in E} e_{v, p}$ 
for phase $p$ using the states given in $I_V$ and $I_{V \times V}$ as the previous
states (assuming $p$, $I_V$ and $I_{V \times V}$ are public). 
By definition of $n_{v, p}$ and $e_{v, p}$, we can in fact
bound the sensitivity of $F$ (on edge-neighboring inputs $G= (V, E)$ and $G' = (V, E \cup \{u, v\})$).

\begin{lemma}\label{lem:sensitivity-adjacents}
    The sensitivity of function $F\left(\adj_v, p, \mathcal{I}_V, \mathcal{I}_{V \times V}\right)$ 
    that computes the sum 
    $\sum_{v \in V} n_{v, p} + \sum_{e \in E} e_{v, p}$ for a given phase $p$ using the states of neighbors in
    $I_{V}$ and adjacent edges in $\mathcal{I}_{V \times V}$
    is $4$.
\end{lemma}

\begin{proof}
    The count $n_{v, p}$ for each node $v$ is the \emph{number} of neighbors $w$ of $v$ for which $\neighbf_v(w)$ outputs $1$. 
    Suppose $\{u, v\} \in E' \setminus E$. Then, $\neighbf_v(u)$ may equal $1$ (and similarly $\neighbf_u(v)$ may equal $1$). 
    The edge would not affect $n_{w, p}$ of any $w \notin \{u, v\}$. 
    Thus, the additional edge in 
    $E'$ can make at most two endpoints increase their $n_{v, p}$, each by $1$, and $\sum_{v \in V} n_{v, p}$ increases by at most $2$.
    By the same argument, $\sum_{e \in E} e_{v, p}$ increases by at most $2$. 
\end{proof}

Suppose, we have another function $f_n(\adj_v, p, \mathcal{I}_V, \mathcal{I}_{V \times V})$ that takes 
as input the adjacency list of a node $\adj_v$ and outputs
$n_{v, p}$ using the states of neighbors and adjacent edges in $\mathcal{I}_V$ and $\mathcal{I}_{V \times V}$.
(Define the function $f_e$ similarly except it outputs $e_{v, p}$.)
An immediate corollary of the above proof bounds the sensitivity of $n_{v, p}$ and $e_{v, p}$.

\begin{corollary}\label{cor:sensitivity-count}
    The sensitivity of $f_n$ for phase $p$ is $1$. The sensitivity of $f_e$ is also $1$.
\end{corollary}

We define similar functions to the above for $s_{v, p}, t_{v, p}, s_p, t_p$. Namely, 
we pass in the states $I_{V}$ and $I_{V \times V}$ into the functions (so they are public information)
and compute the sensitivity for edge-neighboring inputs whose states come from $I_V$ and $I_{V \times V}$.
For simplicity, we do not define these functions explicitly and instead 
give the sensitivity for $s_{v, p}, t_{v, p}, s_p, t_p$ in lieu of these functions.
Our framework also relies on the sensitivity of $s_{v, p}$ and $t_{v, p}$, whose proof is identical to~\cref{lem:sensitivity-adjacents}.

\begin{lemma}\label{lem:sensitivity-st}
The sensitivity of $s_{v, p}$ for phase $p$ is $1$. The same holds for $t_{v, p}$.
\end{lemma}

We bound the sensitivity of $\gcount_p + t_p= \left|\left\{c \in (V \cup E): I_{c, p} \text{ satisfies } F \right\}\right|$.

\begin{corollary}\label{cor:sensitivity-threshold}
    The sensitivity of $s_p = \left|\left\{v \in V: I_{v, p} \text{ satisfies } F \right\}\right|$ for any phase $p$ is $0$. 
    The sensitivity of $t_p$ is $1$.
\end{corollary}

\begin{proof}
At most one additional edge, $I_{e', p}$ in $G'$ can satisfy $F$. Thus, $t_p$ increases by at most 
$1$. All nodes remain the same and hence, $s_p$ does not change.
Thus, the sensitivity is $1$.
\end{proof}

We first show the following for $\mech$.

\begin{lemma}\label{lem:mech-dp}
For any pair of edge-neighboring graphs $G = (V, E)$ and $G' = (V, E \cup \{e'\})$,
let $e'$ be the edge that is present in $E'$ but not in $E$.
Let $R$ and $R'$ be the events where given
valid state values $state_{v, \prev}$ and $state_{e, \prev}$ 
(for all $v \in V, e \in E$), for all $I_{v, \prev}, I_{e, \prev} \in \interns_{\prev}$ and 
$I'_{v, \prev}, I'_{e, \prev} \in \interns'_{\prev}$, it holds that
$I_{v, \prev} = I'_{v, \prev} = state_{v, \prev}$ for all $v \in V$ and
$I_{e, \prev} = I'_{e, \prev} = state_{e, \prev}$ for all $e \in E$. Conditioned on $R$ and $R'$, the
following hold for all $S \in \range(\mech)$ on inputs $A$, and $p$:\footnote{We abuse notation slightly 
to indicate $(\interns_{\prev} \setminus \{I_{e, p}\}, \outputs)$ by $\mech(G, 
\interns_{\prev}, \alg, p) \setminus \{I_{e, p}\}$.}

\begin{align*}
    \frac{1}{\exp(\eps/K)} \leq \frac{\prob\left[\mech(G', \interns'_{\prev}, A, p) \setminus \{I_{e, p}\} = S \mid R'\right]}{\prob\left[\mech(G, \interns_{\prev}, A, p) = S \mid R\right]} \leq
    \exp(\eps/K).
\end{align*}
\end{lemma}

\begin{proof}
    In each phase $p$, mechanism $\mech$ uses the states of the edges and nodes from 
    the most recent phase $\prev$ it depends on and 
    evaluates the node, edge, stopping and output functions. 
    Each function $\neighbf_{v}$ is deterministic.
    By~\cref{cor:sensitivity-count}, $n_{v, p}$ increases by at most $1$ in 
    $G'$ compared to $G$
    conditioned on the events $R$ and $R'$. In the below proofs, we remove the 
    conditioning on the previous states for simplicity in notation, but
    all probabilities are conditioned on the events $R$ and $R'$. Similarly, $\adjf_v$ is also 
    deterministic and $e_{v, p}$ increases by at most $1$, 
    by~\cref{cor:sensitivity-count}, in $G'$ compared to $G$. Since
    $\nodef_v$ is a deterministic function on $I_{v, \prev}$, $\hatn_{v, p}$, 
    and $\he_{v, p}$, 
    the outputs of the functions are equal in $G$ and $G'$ 
    if $\hatn_{v, p}$ and $\he_{v, p}$ are 
    equal in $G$ and $G'$. What remains here is to bound the probabilities 
    that $\hatn_{v, p}$ and $\he_{v, p}$ equal particular values in $\mathbb{Z}$
    for both $G$ and $G'$. 
    Let $\hatn_{v, p}$ and $\he_{v, p}$ represent the values in $G$ and 
    $\hatn_{v, p}'$ and $\he_{v, p}'$ represent the values in $G'$. 
    
    Conditioned on $R$ and $R'$, $\hatn_{v, p}$ is drawn from a conditional 
    symmetric geometric distribution (SGD) with
    mean $n_{v, p}$ and $\he_{v, p}$ is drawn from a conditional SGD with mean $e_{v, p}$.
    By~\cref{cor:sensitivity-count}, $\hatn_{v, p}'$ is drawn from a conditional SGD with mean 
    $n_{v, p} \leq n_{v, p}' \leq n_{v, p} + 1$ (and similarly for $\he_{v, p}'$). 
    All the distributions have the same parameter, $\eps/20K$, by definition of our mechanism. 
    Thus, the conditional distributions from which $\hatn_{v, p}$ and $\hatn_{v, p}'$ are 
    drawn are $\eps$-indistinguishable (see the beginning of~\cref{sec:densest-privacy} for the definition). 
    The same holds for all other pairs of random variables for nodes $v$ and $u$ in $G$ and $G'$.
    There are four pairs of $\eps/(20K)$-indistinguishable conditional distributions. We now show that
    they are independent. Conditioned on $R$ and $R'$, the values $n_{v, p}, e_{v, p}, n_{u, p}$ and $e_{u, p}$
    are fixed and the noises drawn for each of these values are independent. Thus, the distributions are 
    independent and the joint distributions for $(\hatn_{v, p}, \he_{v, p}, \hatn_{u, p}, \he_{u, p})$
    and $(\hatn_{v, p}', \he_{v, p}', \hatn_{u, p}', \he_{u, p}')$ are $4\cdot \eps/(20K) = \eps/(5K)$-indistinguishable.

    We just showed that the new states of the nodes are the same for $G$ and $G'$ with similar probabilities.
    Now, we condition on the outputs of $\nodef_v$ being identical 
    in $G$ and $G'$ for every $v \in V$. This means that the inputs sent to the adjacent edges will be identical. 
    Since the function
    $\edgef_e$ for each $e \in E\cap E'$ is deterministic, all $I_{e, p} = I'_{e, p}$ for $e \in E \cap E'$. 
    
    Finally, we must determine the outputs of the stopping functions 
    and the final outputs of the mechanism are identical in $G$ and $G'$. Since each $\stopf_v$ is a 
    deterministic function, it will output the same output in $G$ and $G'$ as long as 
    $\hs_{v, p} = \hs'_{v, p}$ and $\hatt_{v, p} = \hatt'_{v, p}$. By~\cref{lem:sensitivity-st},
    the sensitivity of $s_{v, p}$ is $1$ (same for $t_{v, p}$). 
    We now condition on the events $W$ and $W'$ where $I_{v, p} = I'_{v, p} = state_{v, p}, 
    I_{e, p} = I'_{e, p} = state_{e, p}$.
    Hence, by  the same argument as above using~\cref{lem:sensitivity-st}, 
    the fact that $s_{v, p}, t_{v, p}, s'_{v, p}, t'_{v, p}$ are fixed conditioned on $W$ and $W'$, 
    and the independence of the noises, $\hs_{v, p}$ and $\hs'_{v, p}$ (resp.\ $\hatt_{v, p}$ and 
    $\hatt_{v, p}'$) are drawn from 
    $\eps/(20K)$-indistinguishable conditional distributions.

    The same argument holds for the global stopping function.
    Given that the stopping functions are met,
    the $(\eps/5K)$-edge DP of the outputs hold by~\cref{lem:mech-dp} provided
    the sensitivities of $\outf_v$ and $\globaloutf$ are 
    $\gs_{\outf_v}$ and $\gs_{\globaloutf}$, respectively. 
    
    Altogether, by the chain rule over the probabilities that the outputs of all functions, $\nodef_v$, $\edgef_e$, $\stopf_v$, $\globalstopf$,
    $\outf_v$, and $\globaloutf$ are identical, we obtain the conditions of the lemma.
\end{proof}

Let $\hat{\alg}$ be the algorithm obtained from \localadjust algorithm $A$ that 
runs $\mech$ during each phase of the algorithm. $\hat{\alg}$ outputs $\outputs$ obtained from $\mech\left(G, \interns_{\prev}, \alg, p\right)$ 
We now show our main theorem.

\begin{theorem}\label{thm:la-dp}
$\hat{\alg}$ is $\eps$-edge DP provided \localadjust algorithm $A$.
\end{theorem}

\begin{proof}
    Given $K$ as the upper bound on the number of phases used by the 
    algorithm, we show this theorem via induction on the phase number $1 \leq p \leq K$ 
    using~\cref{lem:mech-dp}. We show that in the $p$-th phase, the ratio of the probabilities that the outputs in $G$ and $G'$ for 
    phase $p$ equals $Z_p \in \range(\mech)$ is 
    upper bounded by $\exp(\eps \cdot p/K)$, conditioned on the
    outputs of the previous phases being equal to $Z_1, \dots, Z_{p-1}$. 
    Our base case consists of the first phase of the 
    algorithm when the states of nodes and edges are identical as they are set
    to default identical values. By~\cref{lem:mech-dp}, since $\interns_{0} = \interns_{0}'$,
    the ratio of the probabilities that the outputs are both equal to $Z_1$ is upper bounded by
    $\exp(\eps/K)$. 
    
    We assume for our induction step that in the $p$-th phase, the ratio of the probabilities is 
    upper bounded by $\exp(\eps \cdot p/K)$, conditioned on the outputs
    of the previous phases, and prove this for phase $p+1$. Let the previous phase that phase $p + 1$ depends on be
    $\prev$. By our induction 
    hypothesis, we have that the probability ratio is upper bounded by $\exp(\eps \cdot p/K)$ and that $\interns_{p} = \interns_{p}' \setminus \{I_{e, p}'\}$
    where $e \in E' \setminus E$. Let $W$ be the event that the outputs of the previous phases equal $Z_1, \dots, Z_{p-1}$.
    
    Then, since $G$, $A$, and $p+1$ are fixed, by~\cref{lem:mech-dp}, $\frac{\prob\left[\mech(G, \interns_{\prev}, A, p+1) = Z_{p+1} | W \right]}{\prob\left[\mech(G', \interns'_{\prev}, A, p+1) \setminus \{I_{e, p+1}\} = Z_{p+1} | W\right]} \leq 
    \exp(\eps/K)$ and $\frac{\prob\left[\mech(G', \interns'_{\prev}, A, p+1) \setminus \{I_{e, p+1}\} = Z_{p+1} | W\right]}{\prob\left[\mech(G, \interns_{\prev}, A, p+1) = Z_{p+1} | W\right]} \leq 
    \exp(\eps/K)$. Then, by the chain rule over all the conditional probabilities,
    the ratio of the probabilities that the outputs equal $Z_1, \dots, Z_{p+1}$ is upper bounded by
    $\exp(\eps \cdot p/K) \cdot \exp(\eps/K) = \exp(\eps \cdot (p+1)/K)$. 
    
    Since $K$ is an upper bound on the total number of phases, the ratio of the probabilities for all phases up to when the algorithm terminates 
    is upper bounded by $\exp(\eps \cdot K/K) = \exp(\eps)$. 
    After the algorithm $\hat{\alg}$ terminates, the internal states
    $\interns_{K} = \interns_{K}' = \emptyset$ (i.e., the curator outputs empty sets for the states of the 
    nodes and edges). 
    
    Finally, it is sufficient to show that $\frac{\prob[\hat{A}(G) = Z]}{\prob[\hat{A}(G') = Z]} \leq \exp(\eps)$ and 
    $\frac{\prob[\hat{A}(G') = Z]}{\prob[\hat{A}(G) = Z]} \leq \exp(\eps)$ for all $Z \in \range(\hat{A})$ for \emph{discrete} probability distributions
    because for any $S \subseteq \range(\hat{A})$, $\frac{\prob[\hat{A}(G) \in S]}{\prob[\hat{A}(G') \in S]}$ is given by $\frac{\prob[\hat{A}(G)= S_1] + \cdots + \prob[\hat{A}(G) = S_j]}{\prob[\hat{A}(G') = S_1] + \cdots + \prob[\hat{A}(G') = S_j]}$ (for all $S_1, \dots, S_j \in S$). This expression can be simplified to 
    
    \begin{align*}
        \frac{\prob[\hat{A}(G)= S_1] + \cdots + \prob[\hat{A}(G) = S_j]}{\prob[\hat{A}(G') = S_1] + \cdots + \prob[\hat{A}(G') = S_j]} \leq \frac{\exp(\eps) \cdot \prob[\hat{A}(G') = S_1] + \cdots + \exp(\eps) \cdot \prob[\hat{A}(G') = S_j]}{\prob[\hat{A}(G') = S_1] + \cdots + \prob[\hat{A}(G') = S_j]} = \exp(\eps).
    \end{align*}
    
    Hence, $\hat{\alg}$ is $\eps$-edge DP since the algorithm outputs identical outputs when it terminates, with probability ratio bounded
    by $\exp(\eps)$.
\end{proof}

\subsection{$\eps$-LEDP from Locally Adjustable Algorithms}\label{sec:framework-ledp}
Our framework can also be extended to the $\eps$-LEDP setting. In fact,
our $\eps$-LEDP \kc decomposition result is based on this framework. 
Our modified LEDP framework modifies our DP framework in 
the following ways. We remove the $\adjf_v(\{v, w\})$ and 
$\edgef_{\{v, w\}}$ functions for edges. We also remove all global functions
$\globalstopf$ and $\globaloutf$ since
the curator no longer has access to the private graph. We modify all 
$\nodef_v$ so they no longer have $e_{v, p}$ as input.
Then, we modify all $\stopf_v\left(s_{v, p}\right)$ functions \
so they only take as input 
$s_{v, p}$ for each $v \in V$ and no longer take
the states of adjacent edges. We prove that the modified
functions give a $\eps$-LEDP algorithm for any 
locally adjustable algorithm $\alg$ using the 
modified framework we give below. 

\paragraph{Modified Privacy Framework}
Suppose we are provided a \localadjust algorithm $\alg$ by our
modified definition above. Then,
we formulate the following distributed mechanism 
$\mech(\adj_i, \interns_{\prev}, \alg, p) \rightarrow \left(I_{i, p}, \outputs_i\right)$ 
performed by the nodes in the input
the graph $G$; the mechanism takes as input adjacency 
list $\adj_i$ for node $i$,
a set of public (previous) node states $\interns_{\prev}$ 
of $G$, the \localadjust algorithm $\alg$, 
and the phase number $p$. Mechanism
$\mech$ outputs the next state $I_{i, p}$ of node $i$ for phase $p$
if the algorithm is still running or $\emptyset$ if the algorithm has stopped. 
It also outputs the set of outputs $\outputs_i$ for node
$i$ if the algorithm has stopped
or $\emptyset$ if the algorithm is still running.
The mechanism modifies $\alg$ in the following ways:

\begin{itemize}
    \item For each node $i$, the node computes $\hatn_{i, p}$ by using the public node states, 
    sampling $X_{i, p} \sim \geom(\eps/6K)$ and 
    calculating $\hatn_{i, p} = n_{i, p} + X_{i, p}$. 
    Node $i$ uses $\hatn_{i, p}$ 
    instead of $n_{i, p}$ in
    $\nodef_v(I_{i, \prev}, \hatn_{i,p})$ where $I_{i, \prev}$ is from $\interns_{\prev}$. 
    \item When determining the stopping condition, node $i$ samples $S_{i, p}\sim \geom(\eps/(6K))$
    and computes $\hs_{i, p} = s_{i, p} + S_{i, p}$. 
    Node $i$ uses $\hs_{i, p}$
    instead of $s_{i, p}$ in 
    $\stopf_{i}(\hs_{i, p})$.
    \item Let $\gs_{\outf_i}$ be the global sensitivity of $\outf_i$. To compute the output after satisfying the stopping condition, each node $i$ samples $Q_{i, p} \sim \geom(\eps/(6 \cdot \gs_{\outf_{i}} \cdot K))$
    and outputs $\outf_{i}(I_{i, p}) + Q_{i, p}$.
\end{itemize}

\begin{theorem}\label{alg:ledp}
$\hat{\alg}$ is $\eps$-LEDP provided \localadjust algorithm $A$.
\end{theorem}

\begin{proof}
    We implement our LEDP algorithm using the following local
    randomizers. First, by a modified version of~\cref{lem:sensitivity-adjacents} and~\cref{lem:sgd-private}, our mechanism implements an $(\eps/6)$-LR
    for computing the output of $\nodef_v$. Similarly,
    by~\cref{lem:sensitivity-st} and~\cref{lem:sgd-private},
    our mechanism implements a $(\eps/6)$-LR for obtaining
    the output of $\stopf_{i}$. Finally, by~\cref{lem:sgd-private} and because the global sensitivity of 
    $\outf_i$ is $\gs_{\outf_i}$, our mechanism 
    implements a $(\eps/3)$-LR for producing the output
    of each node. Finally, by~\cref{thm:composition} and \cref{thm:group-dp} over all calls to LRs over the $K$
    phases, our algorithm is $\eps$-LEDP.
\end{proof}
}

\section*{Acknowledgements}
We thank Talya Eden for helpful discussions. 
This research was supported by 
DOE Early Career Award \#DE-SC0018947,
NSF CAREER Award \#CCF-1845763, Google Faculty Research Award, Google Research Scholar Award, DARPA
SDH Award \#HR0011-18-3-0007, and Applications Driving Architectures
(ADA) Research Center, a JUMP Center co-sponsored by SRC and DARPA.

\appendix

\section{Proof of the Adaptive Composition Theorem}\label{app:adaptive}

Here we present the proof of~\cref{thm:composition} (\cite{DMNS06,DL09,DNPR10})
for completeness. 

\begin{proof}[Proof of~\cref{thm:composition}]
    Let $\mech$ be an $\eps$-LDP mechanism and $\pmb{y} = (y_1, \dots, y_{p-1})$ be a set of $p$ outputs for 
    adaptively chosen adjacent inputs $\pmb{x} = (x_0, x_1, \dots, x_{p-1})$ and $\pmb{x}' = (x_0', x_1', \dots, x_{p-1}')$
    to the mechanism where $y_i \in \range(\mech)$ for every $i \in [p-1]$. Let the randomness over
    $\mech$ be a discrete probability distribution.
    Below, we abuse notation and let $\mech(\pmb{x})$ denote the set of outputs obtained by running 
    $\mech$ on adaptive inputs $\pmb{x}$.
    We then have:
    
    \begin{align*}
        \frac{\prob[\mech(\pmb{x}) = \pmb{y}]}{\prob[\mech(\pmb{x'}) = \pmb{y}]} &= \left(\frac{\prob[\mech(x_0) = y_1]}{\prob[\mech(x_0') = y_1]}\right) \cdot \prod_{i = 1}^{p-1} \frac{\prob[\mech(x_i) = y_i| y_1, \dots, y_{i-1}]}{\prob[\mech(x_i') = y_i| y_1, \dots, y_{i - 1}]} \\
        &\leq \prod_{i = 1}^{p} \exp(\eps) = \exp(p\eps).
    \end{align*}
    
    The second inequality follows since $\mech$ is $\eps$-LDP, each pair of inputs are adjacent, and previous
    outputs are public information. 
\end{proof}

\section{$k$-Core Decomposition Approximation Proofs}\label{app:ledp-approx}

\begin{proof}[Proof of~\cref{lem:expectation}]
    This proof is a simple modification of the proof of Lemma 5.12 of~\cite{LSYDS22}.
    In this proof, when we refer to the level of a node $i$, we mean the level of $i$
    in $L_{\numlevels-1}$. Furthermore, all expressions are given \emph{in expectation}. 
    For simplicity, we omit the phrase \emph{in expectation} from now on in this proof. 
    Using notation from previous work, let $\kest(i)$ be the core number estimate of $i$ and $\core(i)$ be the
    core number of $i$. First, we show that 
    \begin{align}
        \text{if } \kest(i) \leq (2+\lambda)(1+\lf)^{g'}, \text{ then }
        \core(i) \leq (1+\lf)^{g' + 1}\label{expect-eq:first-inequality}
    \end{align}
    for any group $g'$. Let $T(g')$ be the topmost level of group $g'$ and 
    let $Z_{\lcur}$ be the set of nodes at level $\lcur$ and above.
    In order for $(2+\lambda)(1+\psi)^{g'}$ to be the estimate of node $i$'s core number, 
    the level of $i$ is bounded by $T(g') \leq \level(i) \leq T(g' + 1) - 1$. 
    Let $\lcur$ be $i$'s level. By a modified~\cref{inv:degree-1} for the expectation setting, 
    if $\lcur < T(g' + 1)$, then $|\adj_i \cap Z_\lcur| \leq \upexp^{\gn(\lcur)} \leq
    \upexp^{g' + 1}$. Furthermore,
    each node $w$ at the same or lower level $\lcur' \leq \lcur$ has $|\adj_w \cap Z_{\lcur'}|
    \leq (1+\lf)^{g' + 1}$. 
    
    Suppose we perform the following iterative procedure:
    starting from level $\lcur = 0$,
    remove all nodes in level $\lcur$ during this turn and set $\lcur
    \leftarrow \lcur + 1$ for the next turn and perform the removal again; we perform this
    iterative procedure until all nodes are removed. Using this procedure, the nodes in
    level $0$ are removed in the first turn, the nodes in level $1$ are
    removed in the second turn, and so on. Let $d_{\lcur}(i)$ be the induced degree of
    any node $i$ that is in the graph before the removal in the $\lcur$-th turn. Since we showed above
    that $|\adj_i \cap Z_{\lcur}| \leq (1+\lf)^{g' + 1}$ for any node
    $i$ at level $\lcur \leq T(g' + 1)$, node $i$ on level $\lcur \leq T(g' + 1)$
    after the $\lcur$-th turn has $d_{\lcur}(v) \leq \upexp^{g' + 1}$.
    Thus, when $i$ is removed on the $(\lcur+1)$-st turn, it has degree $\leq
    \upexp^{g' + 1}$. Since all nodes removed before $i$
    also had degree $\leq \upexp^{g' + 1}$ when it was removed,
    by~\cref{lem:folklore}, node $i$
    has core number $\core(i) \leq \upexp^{g' + 1}$.
    
    Now we prove our lower bound on $\kest(i)$. For the below proof, let $d_S(i)$ be the 
    induced degree of node $i$ in the induced subgraph consisting of nodes in $S$.
    We prove that for any $g' \geq 0$, 
    \begin{align}
        \text{if } \kest(i) \geq (1+\lf)^{g'}, \text{ then }
        \core(i) \geq \frac{(1+\lf)^{g'}}{\upexpold}\label{expect-eq:second-inequality}
    \end{align}
    for all nodes $i \in [n]$
    in the graph. We assume for contradiction that there exists a node $i$ where
    $\kest(i) \geq (1+\lf)^{g'}$ and $\core(i) < \frac{(1+\lf)^{g'}}{\upexpold}$.
    To consider this case, we use the \emph{pruning} process defined
    in~\cref{lem:folklore}. For a given subgraph $S$, we
    \emph{prune} $S$ by repeatedly removing all nodes $i \in S$ whose $d_S(i)
    < \frac{(1+\lf)^{g'}}{\upexpold}$. As in the proof of Lemma 5.12
    of~\cite{LSYDS22}, we only consider levels from the same group $g'$. 
    Let $j$ be the number of levels below level $T(g')$. We prove via induction that the
    number of nodes pruned from the subgraph induced by $Z_{T(g') - j}$ must
    be at least
    \begin{align}
        \left(\frac{\upexpold}{2}\right)^{j-1}\left(\lbexpexpect\right)\label{expect-eq:pruned}.
    \end{align}

    We first prove the base case when $j = 1$. In this case, we know that
    $d_{Z_{T(g') - 1}}(i) \geq (1+\lf)^{g'}$ by a modified version of~\cref{inv:degree-2} that holds in expectation.
    In order to prune $i$ from the graph, we must prune at least

    \begin{align*}
        (1+\lf)^{g'} - \frac{(1+\lf)^{g'}}{\upexpold} = \downexp^{g'} \cdot \left(1 -
            \frac{1}{\upexpold}\right)
    \end{align*}
    neighbors of $i$ from $Z_{T(g') - 1}$. We must prune at least this many neighbors
    in order to reduce the degree of $i$ to below the cutoff for pruning a node (as we show more formally below).

    Then, if fewer than $\lbexpexpect$
    neighbors of $i$ are pruned from the graph, then $i$ is not pruned from the
    graph. If $i$ is not pruned from the graph, then $i$ is part of a $\left(
    \frac{(1+\lf)^{g'}}{\upexpold}\right)$-core (by~\cref{lem:folklore})
    and $\core(i) \geq \frac{(1+\lf)^{g'}}{\upexpold}$, a contradiction. Thus, it must be
    the case that at least $\lbexpexpect$
    neighbors of $v$ are pruned in $Z_{T(g') - 1}$.
    For our induction hypothesis, we
    assume that at least the number of nodes as indicated in~\cref{expect-eq:pruned}
    is pruned for $j$ and prove this for $j + 1$.

    Each node $w$ in levels $T(g') - j$ and above has
    $d_{Z_{T(g') - j - 1}}(w) \geq
    (1+\lf)^{g'}$ by~\cref{inv:degree-2} (recall that all $j$ levels
    below $T(g')$ are in group $g'$). For simplicity of expression,
    we denote $\lbabrv \triangleq \lbexpexpect$.
    Then, in order to prune the
    $\left(\frac{\upexpold}{2}\right)^{j-1}\lbabrv$
    nodes by our induction hypothesis, we must prune at least

    \begin{align}
        &\left(\frac{\upexpold}{2}\right)^{j-1}\lbabrv \cdot
        \left(\frac{(1+\lf)^{g'}}{2}\right)
        \label{expect-eq:pruned-edges}
    \end{align}
    edges where each such edge is ``charged'' to the endpoint
    that gets pruned last. (Note that we actually need to prune at least 
    $\downexp^{g'} \cdot \left(1 -
    \frac{1}{\upexpold}\right)$ edges per pruned node
    as in the base case but 
    $\frac{\downexp^{g'}}{2}$ lower bounds this amount.)
    Each pruned node prunes less than
    $\frac{(1+\lf)^{g'}}{\upexpold}$
    edges. Thus, using~\cref{expect-eq:pruned-edges}, the number of nodes
    that must be pruned from $Z_{T(g') - j - 1}$ is at least
    \begin{align}
        \left(\frac{\upexpold}{2}\right)^{j-1} \lbabrv \cdot
        \frac{(1+\lf)^{g'}}{2\left(\frac{(1+\lf)^{g'}}
        {\upexpold}\right)} = \left(\frac{\upexpold}{2}\right)^j \lbabrv.
        \label{expect-eq:final-induction}
    \end{align}%
    \cref{expect-eq:final-induction} proves our induction step. 
    Using~\cref{expect-eq:pruned}, the number of nodes that must be pruned from
    $Z_{T(g') - 2\log_{\upexpold/2}\left(n\right)}$ is
    greater than $n$ (when $n > 2$) since $J \geq 1/2$:
    \begin{align}
        \left(\frac{\upexpold}{2}\right)^{2\log_{\upexpold/2}\left(n\right)} \cdot \lbabrv \geq \frac{n^2}{2}.
        \label{expect-eq:final-eq}
    \end{align}
    Thus, at $j =
    2\log_{\upexpold/2}\left(n\right)$, we run out of
    nodes to prune. We have reached a contradiction as we require pruning greater than 
    $n$ nodes in expectation
    assuming $\core(i) < \frac{(1+\lf)^{g'}}{\upexpold}$. 
    This contradicts the fact that more than $n$ nodes is pruned with $0$ probability.
    
    From the above, we can first obtain the inequality $\core(i) \leq \kest(i)$ from~\cref{expect-eq:first-inequality} since
    this bounds the case when $\kest(i) = (2+\lambda)(1+\psi)^{g'}$; if $\kest(i) < (2+\lambda)(1+\psi)^{g'}$
    then the largest possible value for $\kest(i)$ is $(2+\lambda)(1+\psi)^{g'-1}$ by~\cref{alg:estimate} and
    we can obtain the tighter bound of $\core(i) \leq (1+\psi)^{g'}$. We can substitute $\kest(i) = 
    (2+\lambda)(1+\psi)^{g'}$ since $(1+\psi)^{g' + 1} < (2+\lambda)(1+\psi)^{g'}$ for all $\psi \in (0, 1)$ and 
    $\lambda > 0$. 
    
    Second, by~\cref{expect-eq:second-inequality}, for any estimate $(2+\lambda)(1+\psi)^{g}$, 
    the largest $g'$ for which this estimate
    has $(1+\psi)^{g'}$ as a lower bound is $g' = g + \floor{\log_{(1+\psi)}(2+\lambda)} \geq g 
    + \log_{(1+\psi)}(2+\lambda) -1$. Substituting this $g'$ into $\frac{(1+\lf)^{g'}}{\upexpold}$ 
    results in 
    \begin{align*}
        \frac{(1+\lf)^{g + \log_{(1+\psi)}(2+\lambda) -1}}{\upexpold} = \frac{\frac{(2+\lambda)(1+\lf)^{g}}{1+\lf}}{\upexpold} = \frac{\frac{\kest(v)}{1+\lf}}{\upexpold}.
    \end{align*}
    Thus,
    we can solve $\core(v) \geq \frac{\frac{\kest(v)}{1+\lf}}{(2+\lambda)(1+\lf)}$
    and $\core(v) \leq \kest(v)$ to obtain
    \begin{align*}
        \core(v) \leq \kest(v) \leq 
        \upexpold^2 \core(v)
    \end{align*}
    which is consistent with the definition of
    a $(2+\const, 0)$-factor approximation algorithm (in expectation)
    for core number for any constant $\coren > 0$ and appropriately
    chosen constants $\lambda, \lf \in (0, 1)$ that depend on $\coren$.
\end{proof}

\section{Challenges with Using Previous Techniques}\label{app:challenges}

Let us consider a simple algorithm for obtaining the $k$-core decomposition. The 
classic algorithm of Matula and Beck~\cite{Matula83} 
for static, centralized, sequential $k$-core decomposition repeatedly \emph{peels} the 
node with smallest degree; once a node is peeled, all adjacent edges to the node
are removed and the process repeats until the graph is empty. The core number of each
node is the degree of the node when it is removed from the graph. In fact, 
many current non-DP algorithms for $k$-core decompositions use some variant of peeling.

Suppose we attempt to turn this algorithm into a DP algorithm. Let $f(v_i)$ be a function
that takes a node and returns the core number of $v_i$. On edge-adjacent graphs
$G$ and $G'$, the sensitivity of $f$ is one, $\gs_f = 1$. Furthermore, suppose 
$F(V) = \sum_{v \in V} f(v)$ is the function that returns the sum of all outputs of $f$ on 
input $v \in V$. Then, $\gs_F = n = |V|$ since \emph{every} node's core number can increase
or decrease by $1$ as a result of the deletion or addition of a single edge. (Consider a 
cycle; all nodes in the cycle have core number $2$ but removing one edge reduces the core
numbers to $1$.) 

From the above sensitivity analysis, in order to transform the Matula and Beck algorithm into
an $\eps$-edge DP algorithm, we must add $\geom(\eps/\gs_1) = \geom(\eps)$ noise to the core numbers of every
node. By the composition theorem, this results in a $n\eps$-edge DP algorithm. Thus, instead
of adding $\geom(\eps)$ noise, we must instead add $\geom(\eps/\gs_F) = \geom(\eps/n)$ noise.
However, adding this much noise results in an additive error of $\tO(n/\eps)$ (hiding $\poly(\log n)$ factors).
This is no better than simply guessing the core number!
Thus, this simple algorithm fails to provide us with good approximation factors.

One can instead consider other approximate non-DP algorithms (instead of exact algorithms) but it is
difficult to bound the $\gs_F$ of these algorithms directly. Thus, we must find a function that 
simultaneously has small sensitivity and for which it is easy to bound the $\gs$ of the function. 
Then, we need to show that the function is not queried too many times.
(We do precisely this in our result.)

\section{Dual of the Densest Subgraph LP}\label{app:dual}

The classic densest subgraph ILP introduced by Charikar~\cite{Charikar00}
associates each node $v$ with a variable $x_v \in \{0, 1\}$ where $x_v = 1$
indicates that node $v$ is in the densest subgraph (and $x_v = 0$ otherwise). Each edge is associated with a variable $y_e \in \{0, 1\}$ that indicates
whether it is part of the densest subgraph. Relaxing the variables to take real values in $0 \leq x_v \leq 1$ and $0 \leq y_e \leq 1$
allows us to obtain an LP whose optimal has been shown to equal the density of the densest subgraph.
This LP is shown below in~\cref{lp:densest}.

\begin{equation}\label{lp:densest}
\begin{matrix}
\displaystyle \max \sum_{e \in E} y_e  \\
\textrm{such that} & y_e \leq x_u, x_v, & \forall e = \{u, v\} \in E  \\
& \sum_{v \in V} x_v \leq 1, &  \\
& y_e \geq 0, x_v \geq 0 & \forall e \in E, \forall v \in V
\end{matrix}
\end{equation}

One can show that the dual of the densest subgraph LP naturally corresponds to the minimum outdegree orientation problem where edges are provided an orientation
and the goal is to minimize the maximum outdegree of any node. The LP itself corresponds to this problem in the following way.
Each edge $e = \{u, v\}$ has a load of $1$ which it wants to assign to its endpoints. The variables $f_e(u)$ and $f_e(v)$ represent these
loads. The objective is to minimize the maximum load assigned to any node for any feasible load assignment. The dual of the densest subgraph
LP is given below in~\cref{lp:dual}.

\begin{equation}\label{lp:dual}
\begin{matrix}
\displaystyle \min D  \\
\textrm{such that} & f_e(u) + f_e(v) \geq 1, & \forall e = \{u, v\} \in E  \\
& \sum_{e \ni v} f_e(v) \leq D, &  \\
& f_e(u), f_e(v) \geq 0 & \forall e = \{u, v\} \in E
\end{matrix}
\end{equation}

The dual is used in the proof of~\cref{lem:feasible-lp}.

\ifsubmit
    \bibliographystyle{ACM-Reference-Format}
\else
    \bibliographystyle{alpha}
\fi
\bibliography{ref}

\end{document}